\documentclass[10pt,journal,compsoc]{IEEEtran}
\usepackage{siunitx}

\usepackage{amssymb, amsmath, mathrsfs}					  

\usepackage{graphicx}	                                  
\usepackage{subfigure}       
\usepackage{makecell}                            

\usepackage{booktabs, multirow}		

\usepackage{xcolor}
\usepackage{colortbl}

\usepackage{amsthm}
\newtheorem{theorem}{Theorem}[section]

\theoremstyle{remark}
\newtheorem{remark}{Remark}
\theoremstyle{definition}

\usepackage{xspace}

\usepackage{xr}       

\usepackage{soul}							
\soulregister\ref7
\soulregister\cite7      

\newcommand{\etal}{\textit{et al}.}
\newcommand{\ie}{\textit{i}.\textit{e}.}
\newcommand{\eg}{\textit{e}.\textit{g}.}
\newcommand{\iid}{i.i.d.\xspace}

\newcommand{\revision}{}

\begin{document}
\title{Image Restoration Learning via Noisy Supervision in the Fourier Domain}
% \title{Weakly Supervised Image Restoration Learning with Loss in Frequency Domain}
% \title{NBare Demo of IEEEtran.cls\\ for IEEE Journals}

\author{Haosen~Liu,
        Jiahao~Liu,
        Shan~Tan,~%~\IEEEmembership{Member,~IEEE,}
        and~Edmund~Y.~Lam,~\IEEEmembership{Fellow,~IEEE}% <-this % stops a space
\thanks{Haosen Liu and Edmund Y. Lam are with the Department of Electrical and Electronic Engineering, The University of Hong Kong, Hong Kong SAR, China (e-mail: hsliu@eee.hku.hk, elam@eee.hku.hk).}% <-this % stops a space
\thanks{Jiahao Liu and Shan Tan are with the Key Laboratory of Image Processing and Intelligent Control, Ministry of Education; School of Artificial Intelligence and Automation, Huazhong University of Science and Technology, Wuhan, China (e-mail: jiahaoliu\_aia@hust.edu.cn, shantan@hust.edu.cn).}% <-this % stops a space
\thanks{Corresponding author: Edmund Y. Lam.}
}


\markboth{Submitted to IEEE TRANSACTIONS ON PATTERN ANALYSIS AND MACHINE INTELLIGENCE}{}
% \markboth{Journal of \LaTeX\ Class Files,~Vol.~14, No.~8, August~2015}%
% {Shell \MakeLowercase{\textit{et al.}}: Bare Demo of IEEEtran.cls for IEEE Journals}



\IEEEtitleabstractindextext{%
\begin{abstract}
Noisy supervision refers to supervising image restoration learning with noisy targets. It can alleviate the data collection burden and enhance the practical applicability of deep learning techniques. However, existing methods suffer from two key drawbacks. Firstly, they are ineffective in handling spatially correlated noise commonly observed in practical applications such as low-light imaging and remote sensing. Secondly, they rely on pixel-wise loss functions that only provide limited supervision information. This work addresses these challenges by leveraging the Fourier domain. We highlight that the Fourier coefficients of spatially correlated noise exhibit sparsity and independence, making them easier to handle. Additionally, Fourier coefficients contain global information, enabling more significant supervision. Motivated by these insights, we propose to establish noisy supervision in the Fourier domain. We first prove that Fourier coefficients of a wide range of noise converge in distribution to the Gaussian distribution. Exploiting this statistical property, we establish the equivalence between using noisy targets and clean targets in the Fourier domain. This leads to a unified learning framework applicable to various image restoration tasks, diverse network architectures, and different noise models. Extensive experiments validate the outstanding performance of this framework in terms of both quantitative indices and perceptual quality.
\end{abstract}

}

\maketitle

% \fi
\IEEEpeerreviewmaketitle

\section{Introduction}
\IEEEPARstart{I}{mage} restoration aims to provide an accurate estimation~$\hat{z}$ of a high-quality latent image $z$ from its degraded low-quality measurement $x$. Traditional image restoration methods rely on hand-designed image prior models. Some representative techniques include the variational methods~\cite{TV}, sparse coding~\cite{K-SVD}, collaborative filtering~\cite{BM3D}, and the low-rank methods~\cite{SAIST, WNNM}. These methods often have limited performance as their models struggle to effectively capture the complex properties of real-world images. In recent years, deep learning methods have gained increasing attention in image restoration tasks due to their larger model capacity. They have demonstrated significant improvements over traditional model-based methods in areas such as image denoising~\cite{DnCNN, RED}, super-resolution~\cite{SRCNN, DAN}, and deblurring~\cite{DeepDeblur}. However, the outstanding performance of deep learning methods is typically based on supervised learning, which requires access to a set of high-quality images as training targets. Unfortunately, obtaining such images in practice is often challenging or even impossible, limiting the practical applicability of these methods.

% #2 can only collect noisy targets in practice
In imaging systems, a fundamental trade-off exists between signal-to-noise ratio (SNR) and spatiotemporal resolution. When images are captured at high spatial resolutions or ultra-high speeds, the limited amount of incident light received at each pixel leads to low SNR. This issue becomes even more pronounced in low-light conditions. Consequently, when collecting data, one is faced with a choice between obtaining high-resolution but noisy images or clean but low-resolution images. In many situations, it is generally impractical to collect high-resolution and clean images as training targets. In this scenario, some approaches propose post-processing the captured high-resolution but noisy images to generate pseudo-clean targets. These methods either estimate the latent clean image from multiple (\eg, 150) noisy measurements of the same scene~\cite{SIDD} or apply denoising techniques such as BM3D~\cite{BM3D} to remove noise in each noisy observation~\cite{RealBlur}. However, both strategies have drawbacks compared to using noisy images directly as targets. The former approach requires the scene to be static so that the repetitive measurements is feasible, not applicable to dynamic scenes. The latter approach introduces unwanted artifacts and may inadvertently smooth out image details, as our experiments will demonstrate.
% As such, this work is aimed at developing efficient weakly-supervised learning with noisy targets for image restoration.

% #3 a brief intro to existing weakly-supervised methods
This work aims to develop efficient supervision for image restoration learning with noisy images as targets. A pioneering work in this field is the Noise2Noise~(N2N)~\cite{N2N}, which establishes a fundamental statistical equivalence, \ie, if the noise $n$ in the noisy target $y$ is zero-mean and independent of the input $x$, then
\begin{equation}    \label{eqn: N2N}
    \mathbb{E}\left[\|f_\theta(x) - y\|_2^2\right] = \|f_\theta(x) - z\|_2^2 + const,
\end{equation}
where $\|\cdot\|_2$ denotes the $\ell_2$ norm, $f_\theta(\cdot)$ denotes the network with parameters $\theta$, $z$ is the latent clean image, $x = g(z)$ is the image degraded by the degradation model $g(\cdot)$, and $y = z + n$ is the image corrupted by the zero-mean noise $n$. Note that the independence between $n$ and $x$ can naturally hold if the degraded image $x$ and the noisy image $y$ are two independent observations. Following N2N, several variants have been proposed by using this equivalence along with additional assumptions such as the spatial independence of noise (Noise2Void~\cite{N2V}, Noise2Self~\cite{N2S}, Self2Self~\cite{S2S}, \revision{Neighbor2Neighbor~\cite{Nr2Nr}}) and the prior knowledge of the noise model (Noiser2Noise~\cite{Noiser2Noise}, Recorrupted2Recorrupted~\cite{R2R}). 

% ------------------------------------------------------------------- %
% fig: overview
\begin{figure}[!t]
\centering
\includegraphics[width=\linewidth]{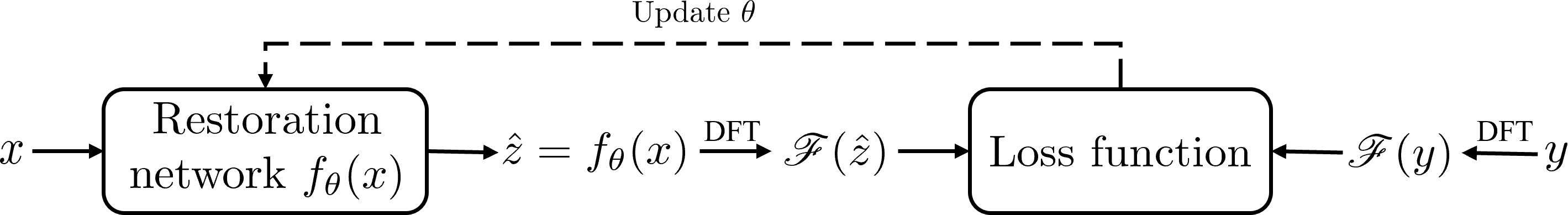}
\caption{The IR-NSF framework overview. We propose to learn a restoration network through supervision from noisy targets $y$ in the Fourier domain.}
\label{fig: overview}
\end{figure}

% \begin{figure*}[!t]
% \centering
% \includegraphics[width=0.7\linewidth]{figs/overview/overview.png}
% \caption{Framework overview. Our proposed IR-NSF is to achieve the network learning for image restoration tasks, which are to provide an estimation $\hat{z}$ of the latent clean image from its degraded observation $x$, through supervision from the noisy target $y$ in the Fourier domain.}
% \label{fig: overview}
% \end{figure*}
% ------------------------------------------------------------------- %
% #4 their drawbacks
% A common drawback shared by these works
% These works have benefited various applications. However, they still have two common drawbacks. 
 % such as low-light vision~\cite{Low-Light}, remote sensing~\cite{Remote-Sensing}, and X-ray imaging~\cite{X-ray}
These existing methods suffer from two common drawbacks. Firstly, they are ineffective when dealing with spatially correlated noise that is widely observed in practical applications such as low-light vision~\cite{Low-Light} and remote sensing~\cite{Remote-Sensing}. As noted in \cite{N2N}, the performance of N2N deteriorates as the spatial correlation of noise increases. This issue becomes even more pronounced for its variants, as their assumptions may be violated. For correlated noise, the spatial independence of noise no longer holds, and it becomes more challenging to determine the noise model. Secondly, the loss function in (\ref{eqn: N2N}) penalizes each local pixel individually, capturing only limited information and failing to provide global supervision for image restoration learning. Consequently, even if the statistical equivalence given by (\ref{eqn: N2N}) holds, N2N and its variants would still exhibit inferior performance.

% #5 Motivation
To tackle these drawbacks, we introduce the novel concept of establishing noisy supervision in the Fourier domain, driven by two key observations. Firstly, spatially correlated noise often exhibits sparsity and independence in the Fourier domain. The most representative example of such noise is the stripe noise that holds a strong spatial correlation along each row or the entire image. It is commonly observed in low-light vision~\cite{Low-Light}, remote sensing~\cite{Remote-Sensing}, infrared imaging~\cite{Infrared}, and microscopy imaging~\cite{Microscopy}. Other examples include the stationary correlated noise in X-ray imaging~\cite{X-ray} and various structure noise mentioned by~\cite{VSNR, SN-BM3D}. Secondly, each Fourier coefficient is calculated based on the entire image, encompassing global information and offering more significant supervision. When clean targets are available, using a frequency-based loss can yield better results compared to a local pixel-based loss in various tasks such as image super-resolution \cite{Fourier-Perceptual, SUDF-DDL}, image deblurring \cite{MIMO}, and image synthesis \cite{FFL}. This highlights the benefits of utilizing the Fourier domain for image restoration.
% \revision{This is because sub-sampled or blurry images often share consistent local regions with their latent ground truth images. When the pixel-based loss is adopted, this local consistency will tend to guide the network to learn an identity function rather than the super-resolving or deblurring the degraded images.  }

% #6 Contribution
To lay the foundation for establishing Fourier-based noisy supervision, we present several key results. Firstly, similar to using the central limit theorem (CLT) in the analysis of discrete cosine transform coefficients in image compression~\cite{DCT}, this work conducts a comprehensive analysis of noise properties in the Fourier domain, showing that the Fourier coefficients of many types of noise converge in distribution to the Gaussian distribution. This enables a unified treatment of different noise models under a common framework. Secondly, we investigate the impact of noise on a Fourier-based loss and establish the statistical equivalence between using noisy targets and clean targets for such loss. Leveraging this equivalence, we propose an efficient learning framework, termed image restoration learning via noisy supervision in the Fourier domain (IR-NSF), as illustrated in Fig.~\ref{fig: overview}. To validate its effectiveness, we apply it to multiple image restoration tasks, diverse network architectures, and different noise models. In all cases, IR-NSF consistently demonstrates outstanding performance in terms of both quantitative metrics and perceptual quality.

Furthermore, it is important to highlight the broader impacts of this work in fields such as data collection and unsupervised denoising. Our experiments suggest that collecting noisy targets not only relieves the burden of data collection but also leads to superior results compared to generating pseudo-clean targets. Moreover, the newly established equivalence enables the development of unsupervised denoising techniques for spatially correlated noise, exemplified by an unsupervised stripe removal designed to address the stripe noise.

In summary, our main contributions are:
\begin{itemize}
    \item \emph{Advantage of the Fourier Domain}: We highlight the benefits of establishing noisy supervision in the Fourier domain. We demonstrate that spatially-correlated noise can exhibit sparsity and independence in the Fourier domain, and Fourier coefficients provide more significant supervision.
    
    \item \emph{Consistent Distribution of Fourier Coefficients}: We show that Fourier coefficients of many types of noise converge in distribution to the Gaussian distribution. This enables us to analyze different noise models within a unified framework, enhancing the understanding of noise characteristics.
    
    \item \emph{Fundamental Statistical Equivalence}: We establish a fundamental statistical equivalence between using noisy targets and clean targets \revision{under the assumption that noise in the noisy target is zero-mean, independent of the input in the training pair, and has Fourier coefficients that follow Gaussian distributions}. It forms the basis for an efficient Fourier-based learning framework that enhances the practical applicability of diverse networks on various image restoration tasks.

    \item \emph{Broad Impacts}: \revision{Our analysis of a beam-splitting data collection system designed for dynamic scenes demonstrates that omitting the denoising step during data collection and employing the noisy supervision can yield superior image restoration results.} Additionally, we develop an effective unsupervised stripe removal by leveraging the proposed Fourier-based statistical equivalence to address the stripe noise.
\end{itemize}

% Besides, we also demonstrate that this work can potentially benefit other fields such as the data collection and the unsupervised denoising.
% The remainder of the paper is organized as follows.

\section{Related Work}
This section will provide a brief introduction to works that are highly related to this study.
% in the fields of noisy supervision for image restoration, noise model, and Fourier-related image restoration learning.
% several research fields that are highly related to this study.

\subsection{Noise Model}    \label{sec: noise_model}
The most widely adopted noise model is the additive, independent, and homogeneous Gaussian distribution~\cite{Gaussian}. However, it is insufficient for effectively modeling the signal-dependent noise such as the Poisson-Gaussian noise. To address this limitation, Foi~\etal~proposed a heteroscedastic Gaussian model~\cite{Poi-Gau}. The error introduced by approximating the Poisson noise with the Gaussian noise becomes smaller than $1\%$ when the photo flux is more than 5 photons per pixel~\cite{Poi-Gau-Accuracy}. While the heteroscedastic Gaussian model is more realistic than the homogeneous one, it still fails to account for structured noise commonly observed in various applications. For example, recent studies have demonstrated the presence of periodic noise and row noise in images captured under extremely low-light conditions~\cite{Low-Light}. Similarly, periodic stripes have been observed in whisk-room imaging systems, while non-periodic stripes occur in push-broom imaging for remote sensing~\cite{Remote-Sensing}. Notably, \cite{CMOS} emphasizes the importance of reducing stripe noise for the CMOS image sensor, as stripe-wise noise tends to be significantly more visible than pixel-wise noise. Moreover, noise affecting projections in X-ray medical imaging is commonly assumed to be stationary and correlated~\cite{X-ray}. To describe these types of noise, Fehrenbach~\etal~\cite{VSNR} and M{\"a}kinen~\etal~\cite{SN-BM3D} proposed a simple yet effective model, \ie,
\begin{equation}    \label{eqn: NM} % Noise Model (NM)
    n = h * \eta,
\end{equation}
where $h$ denotes a correlation kernel, $\eta$ denotes zero-mean \iid~noise, and $*$ denotes the convolution operator. By setting different $h$ and $\eta$, this model can represent both the uncorrelated and correlated noise~\cite{VSNR, SN-BM3D}. 

In this work, we will show that although the different types of noise mentioned above follow various distributions in the spatial domain, their Fourier coefficients all approximately follow the Gaussian distribution. This enables a unified analysis of these noise types under a common framework. Additionally, we will demonstrate that spatially correlated noise can exhibit sparsity and independence in the Fourier domain, highlighting the potential advantages of utilizing it for developing noisy supervision methods.

\subsection{Noisy Supervision for Image Restoration}
Existing methods in this field are all developed in the spatial domain, with a primary focus on the denoising task. 

Some representatives are the N2N and its variants. They are all established based on the statistical equivalence given by (\ref{eqn: N2N}), while differing in the way of constructing input-target pairs $(x, y)$. 
The N2N collects training pairs from two independent measurements of the same scene. Its variants aim at unsupervised denoising and generate input-target pairs from a single noisy image. For instance, Noise2Void~\cite{N2V}, Noise2Self~\cite{N2S}, Self2Self~\cite{S2S}, and \revision{Neighbor2Neighbor~\cite{Nr2Nr}} assume spatial independence of the noise and propose dividing each noisy image into two parts, where one serves as the input and the other as the target. On the other hand, variants such as Noiser2Noise~\cite{Noiser2Noise} and Recorrupted2Recorrupted~\cite{R2R} assume the given noise model and its parameters, generating input-target pairs by further adding noise to the captured image. However, due to the limitations of their assumptions, these variants cannot effectively handle spatially-correlated noise.

Additionally, there are other types of methods explored in this field. For instance, SDA-SURE~\cite{SURE} develops noisy supervision based on Stein's unbiased risk estimator but requires noise to be \iid~Gaussian. The AmbientGAN~\cite{AmbientGAN} and noise-robust GAN (NR-GAN)~\cite{NR-GAN} make the learning process of a generative adversarial network (GAN) robust to corruption. The SS-GMM~\cite{SS-GMM} learns the image prior with the Gaussian mixture model (GMM) from only a single noisy image. However, these methods rely on assumptions, such as prior knowledge of noise model~\cite{AmbientGAN}, rotation-invariance~\cite{NR-GAN}, and \iid~Gaussian~\cite{SS-GMM}, which are not applicable to general spatially-correlated noise. Furthermore, their performance has not been validated on the diverse image restoration tasks. Another approach, the deep image prior (DIP)~\cite{DIP}, utilizes the early-stopping strategy to regularize the generator training. But how to determine the stopping point remains unclear, undermining its practical performance. Also, it requires training a generator for each degraded image, making it extremely time-consuming.

In this work, we propose that the Fourier domain is more suitable for developing noisy supervision, and we establish a Fourier-based statistical equivalence that not only paves the way for a novel learning framework applicable to diverse image restoration tasks and noise models but also enhances the potential practical applicability of the aforementioned methods, opening up new avenues for their improvement.

\subsection{Fourier-based Supervision}
The Fourier transform has long been a fundamental tool in image processing and has recently been integrated into deep learning designs. Chen~\etal~\cite{Recombination} improved the generalization behavior of convolutional neural networks (CNN) for image classification by utilizing the amplitude-phase recombination strategy. This approach was based on the observation that humans rely heavily on phase information for recognition. In contrast, Jiang~\etal~emphasized the necessity of both amplitude and phase for image reconstruction~\cite{FFL}. They proposed a novel method that measures the frequency distance between real and fake images to supervise learning for image reconstruction and synthesis. Similarly, Fuoli~\etal ~\cite{Fourier-Perceptual} demonstrated the effectiveness of developing supervision using Fourier amplitude and phase in achieving superior perceptual results for image super-resolution. Furthermore, Fourier-based supervision has also shown promising results in the image deblurring task~\cite{MIMO}. However, these approaches typically require clean targets, which poses challenges in practical scenarios.

In this work, our goal is to develop efficient Fourier-based supervision directly from noisy targets, which are more readily available. By addressing this challenge, we aim to bridge the gap between the advancements made by these existing methods and their practical applicability.

\section{Method}
\revision{In this section, we aim to develop a statistical equivalence between noisy and clean targets by assuming that noise in the noisy target is zero-mean, independent of the input in the training pair, and has Fourier coefficients that follow Gaussian distributions. Initially, in Sec.~3.2, we conduct a noise analysis in the Fourier domain to illustrate how the assumption concerning the Fourier coefficients of noise is mild and applicable to a diverse range of noise types. Subsequently, in Sec.~3.3, we present the proof and discussions regarding this statistical equivalence.}
% This section will analyze the properties of noise in the Fourier domain. Additionally, we will prove the Fourier-based statistical equivalence of using noisy and clean targets.

\subsection{Notation and Definition}
% \revision{For the supervised learning framework, an image restoration network $f_{\theta}(x)$ is trained by minimizing the loss function $L\big(f_{\theta}(x), z\big)$ defined on the latent ground truth image $z$ and the degraded image $x$. In this work, we aim to develop a framework that can achieve image restoration learning with the loss function $L\big(f_{\theta}(x), y\big)$ utilizing the noisy image $y$ as the training target. For noise $n = y - z$ contained in the noisy target, it is assumed to be a zero-mean random variable independent with the training input $x$. }
In this work, the latent clean image and the noisy image are denoted by $z$ and $y$, respectively. The target noise is defined as $n = y - z$. 
% The image and noise
They are denoted by 2-dimensional variables of the size $U \times V$. We use $[u,v]$ to denote the position indices in the spatial domain, \eg, the noise $n$ at the position $[u,v]$ is denoted by $n[u,v]$. The Fourier transform adopted in this work refers to the 2-dimensional discrete Fourier transform (2D-DFT). For a 2-dimensional variable such as the noise $n$, its Fourier transform is denoted by $\mathscr{F}(n)$ and is calculated as 
\begin{equation}    \label{eqn: DFT}
   \mathscr{F}(n)[k,l] = \frac{1}{UV} \sum\limits_{u=0}^{U-1} \sum\limits_{v=0}^{V-1} n[u,v] e^{-j2\pi \left(\frac{ku}{U}+\frac{lv}{V}\right)},
\end{equation}
where $[k,l]$ denote the position indices in the Fourier domain. The real and imaginary coefficients of $\mathscr{F}(n)$ are denoted by $a(n)$ and $b(n)$, respectively. They are calculated as 
\begin{small}
\begin{equation}    \label{eqn: real}
   a(n)[k,l] = \frac{1}{UV} \sum\limits_{u=0}^{U-1} \sum\limits_{v=0}^{V-1} n[u,v] \cos\left(2\pi \left(\frac{ku}{U}+\frac{lv}{V}\right)\right),
\end{equation}
\end{small}
and
\begin{small}
\begin{equation}    \label{eqn: imag}
   b(n)[k,l] = \frac{-1}{UV} \sum\limits_{u=0}^{U-1} \sum\limits_{v=0}^{V-1} n[u,v] \sin\left(2\pi \left(\frac{ku}{U}+\frac{lv}{V}\right)\right).
\end{equation}
\end{small}\revision{These equations demonstrate that the Fourier coefficients of an variable are either the average or a weighted sum of its values across all positions, thereby containing the global information of this variable.
% $a(n)[0,0]$ is the mean of the variable $n$ and other coefficients are the weighted sum of variable values across all positions. As such, they contain global information of this variable.
} 

\subsection{Noise Analysis in the Fourier Domain} \label{sec: noisy_analysis}
\subsubsection{Spatially-independent noise}
Our analysis begins by considering the Poisson-Gaussian noise. For a Poisson distribution $\mathcal{P}(\lambda)$, it can be approximated by a Gaussian distribution if $\lambda$ is sufficiently large. In the field of image processing, $\lambda$ refers to the expected number of photons detected by the image sensor. It is observed in~\cite{Poi-Gau-Accuracy} that the error of approximating Poisson distribution with Gaussian distribution is less than $1\%$ when the photon flux is more than 5 photons per pixel (\ie, $\lambda \geq 5$). As such, even under the low-light conditions~\cite{Low-Light}, the Poisson-Gaussian noise can be approximately described by the heteroscedastic Gaussian model~\cite{Poi-Gau}, \ie,
\begin{equation}
    n[u,v] \sim \mathcal{N}\big(0, \alpha z[u,v] + \beta\big),
\end{equation}
where $\alpha$ and $\beta$ are noise parameters. For noise following such a model, it is relatively intuitive that each of its Fourier coefficients follows a Gaussian distribution since the linear combination of independent Gaussian variables is still a Gaussian variable.

Moreover, we observe that such a property is not only possessed by the Poisson-Gaussian noise that can be approximated by the heteroscedastic Gaussian model, but by any type of \iid~noise whose distribution might look very different from the Gaussian distribution. More specifically, we have the following theorem.
\begin{theorem} \label{theorem: CLT-DFT-iid}
    Suppose $n$ is \iid~noise with bounded cumulants. As $UV$ approaches infinity, each Fourier coefficient of $n$ will (i) converge in distribution to a Gaussian distribution and (ii) be independent with other coefficients.
\end{theorem}
\begin{proof}
    The proof is provided in the Appendix~\ref{proof: CLT-DFT-iid}.
\end{proof}
% ------------------------------------------------------------------- %
% fig: noise analysis

% \begin{figure}[!t]
% \centering
% \includegraphics[width=\linewidth]{figs/noise_analysis/noise_analysis.png}
% \caption{Histograms of noise in Spatial domain and its frequency coefficient in Fourier domain. For noise following different distributions in spatial domain, their Fourier coefficients consistently follow the Gaussian distribution.}
% \label{fig: noise_analysis}
% \end{figure}

\begin{figure}[t] \centering
    \makebox[0.02\textwidth]{}
    \makebox[0.22\textwidth]{\scriptsize \textbf{\quad \quad Spatial domain}}
    \makebox[0.22\textwidth]{\scriptsize \textbf{\quad \quad Fourier domain}}
    \vspace{.5em}
    \\
    \raisebox{0.63\height}{\makebox[0.02\textwidth]{\rotatebox{90}{\makecell{\scriptsize \textbf{Gaussian}}}}}
    \includegraphics[width=0.22\textwidth]{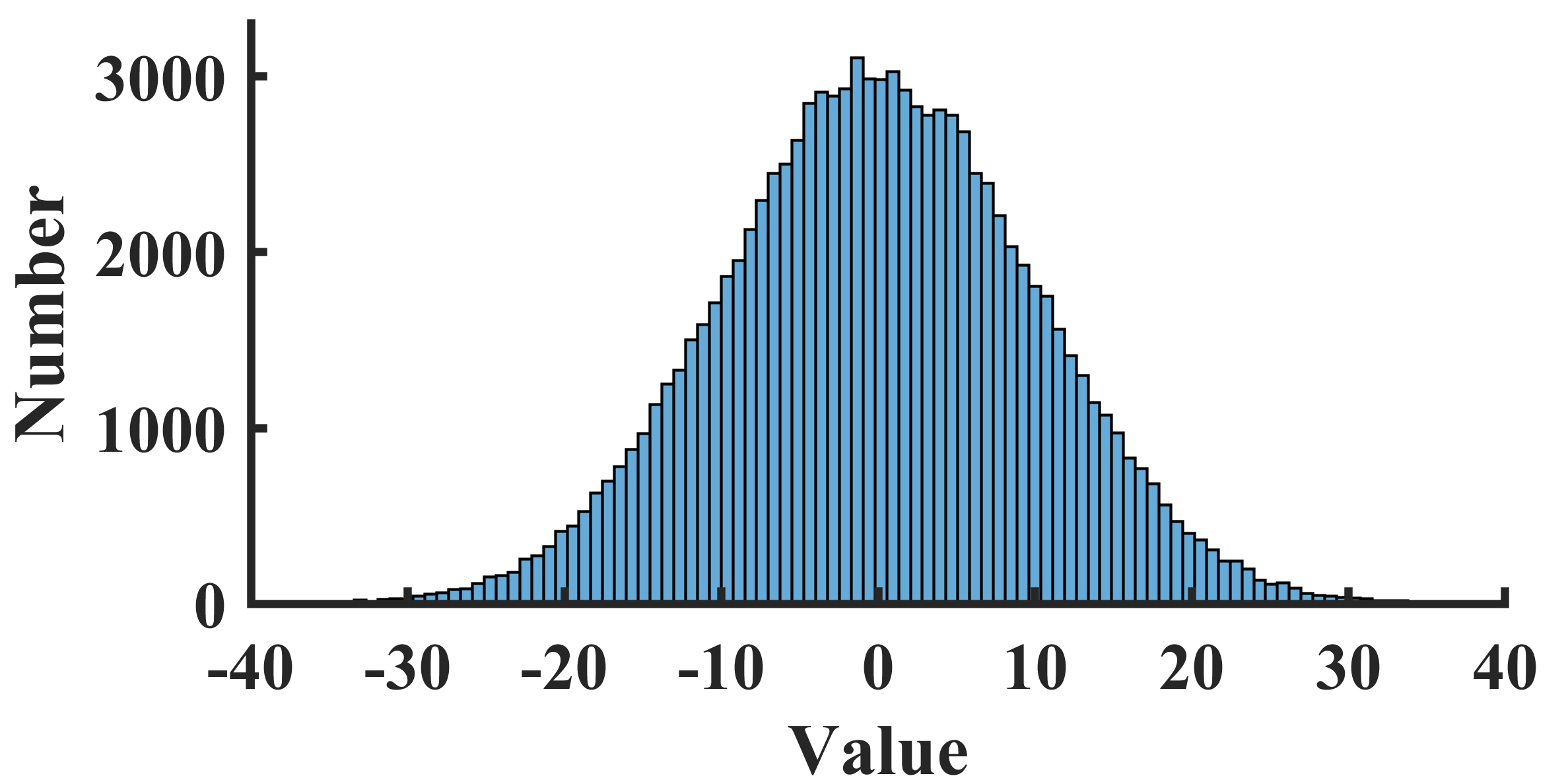}
    \includegraphics[width=0.22\textwidth]{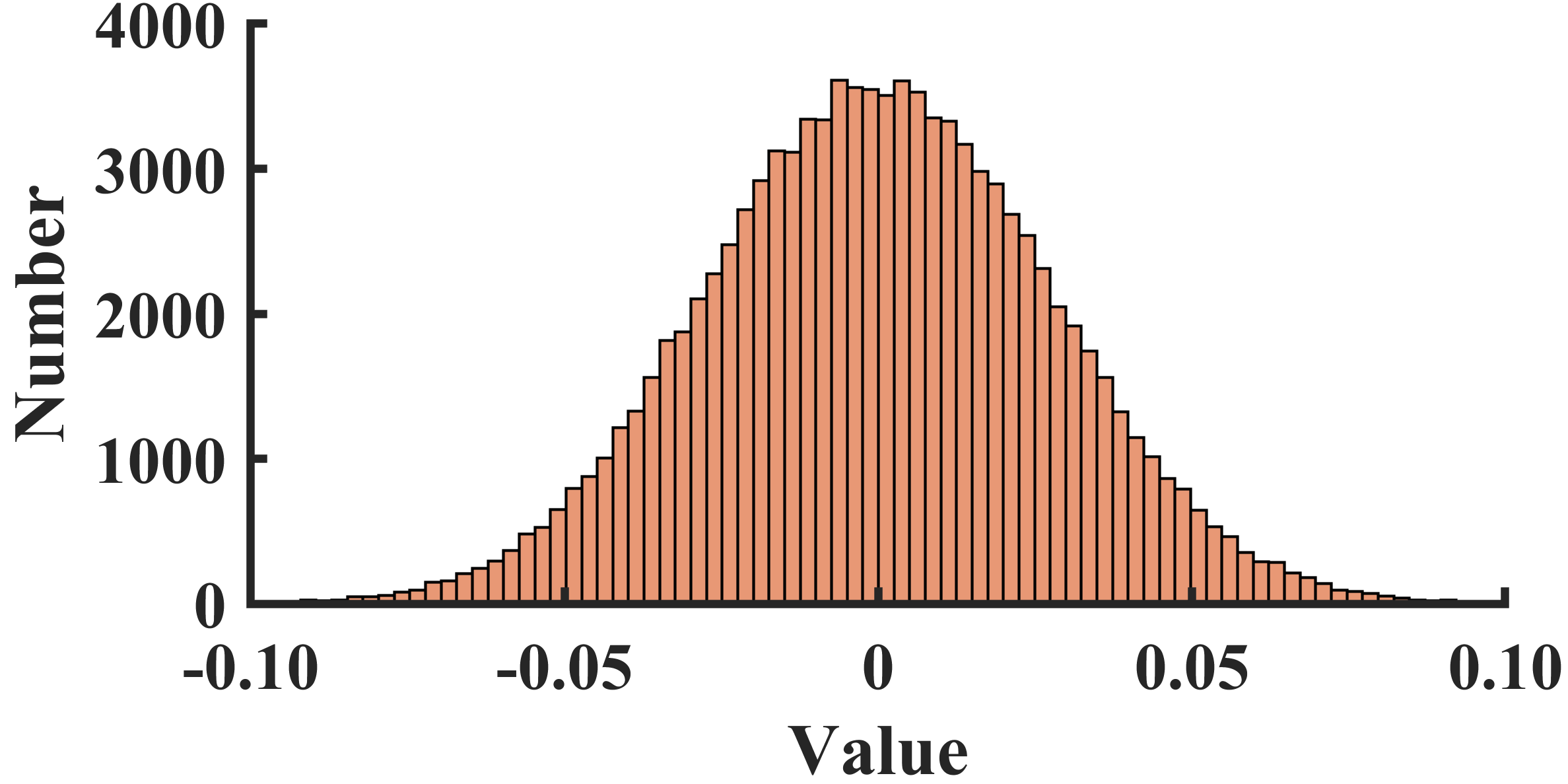}
    \\
    \raisebox{0.85\height}{\makebox[0.02\textwidth]{\rotatebox{90}{\makecell{\scriptsize \textbf{Poisson}}}}}
    \includegraphics[width=0.22\textwidth]{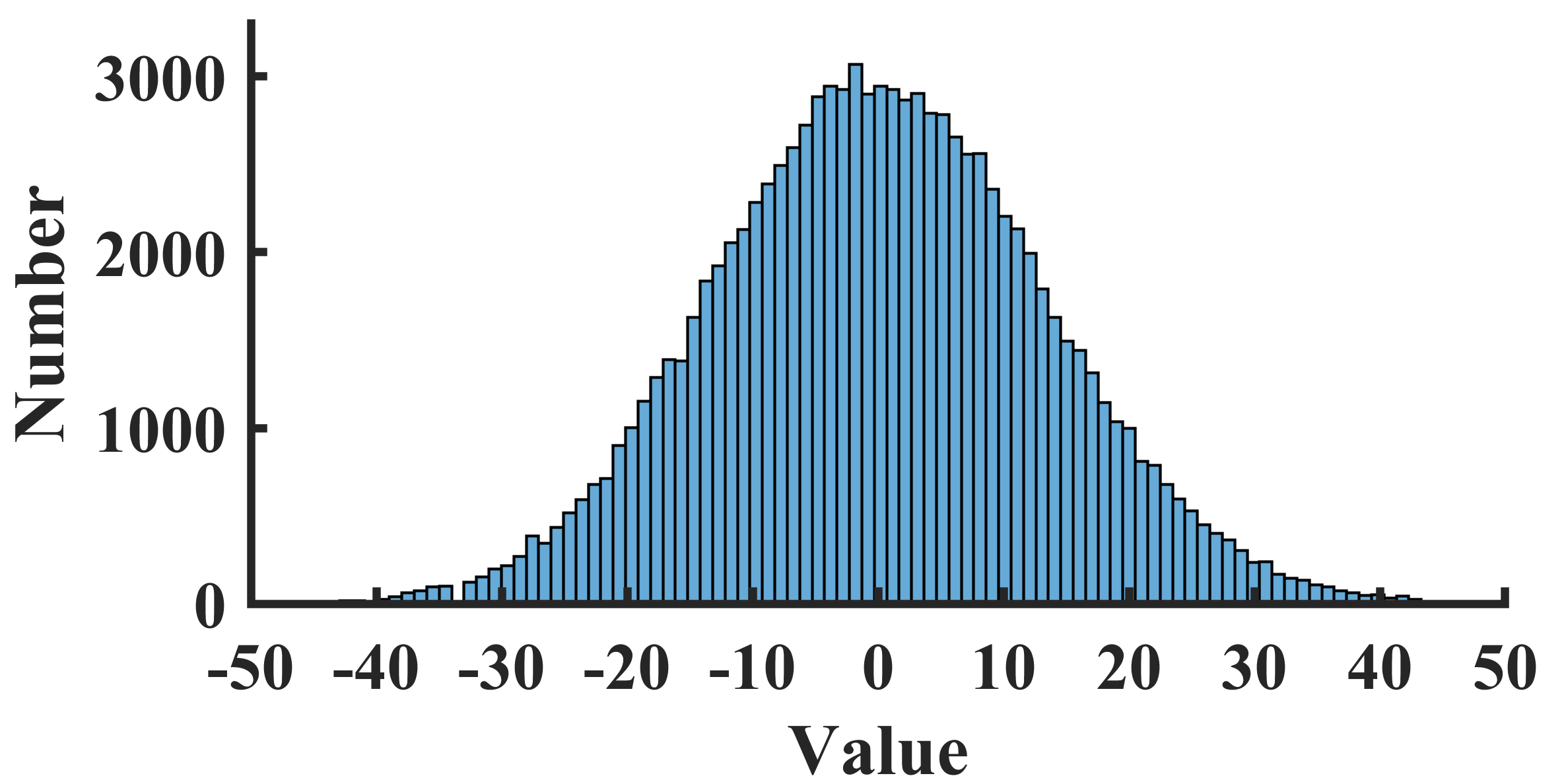}
    \includegraphics[width=0.22\textwidth]{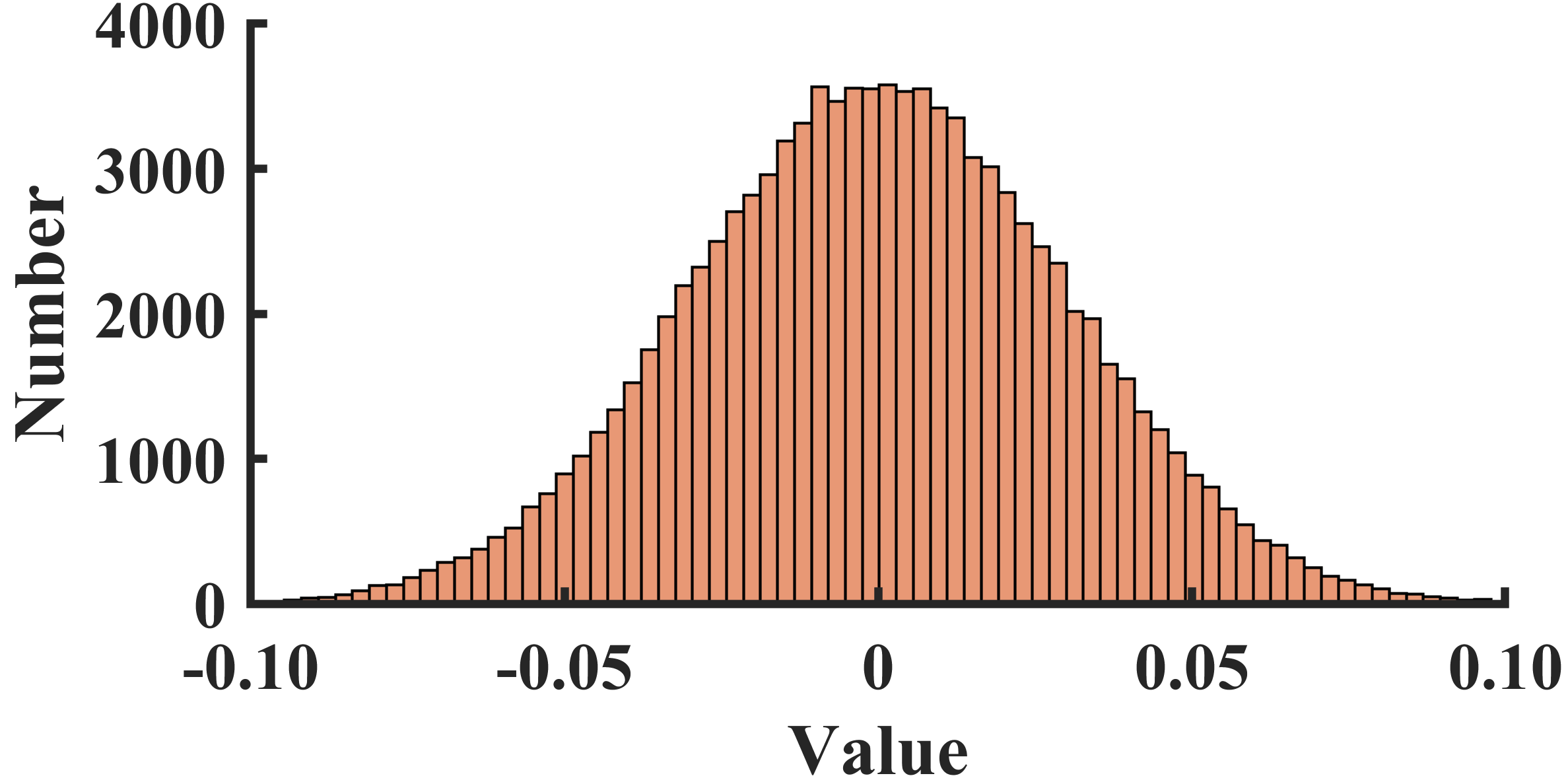}
    \\
    \raisebox{0.7\height}{\makebox[0.02\textwidth]{\rotatebox{90}{\makecell{\scriptsize \textbf{Uniform}}}}}
    \includegraphics[width=0.22\textwidth]{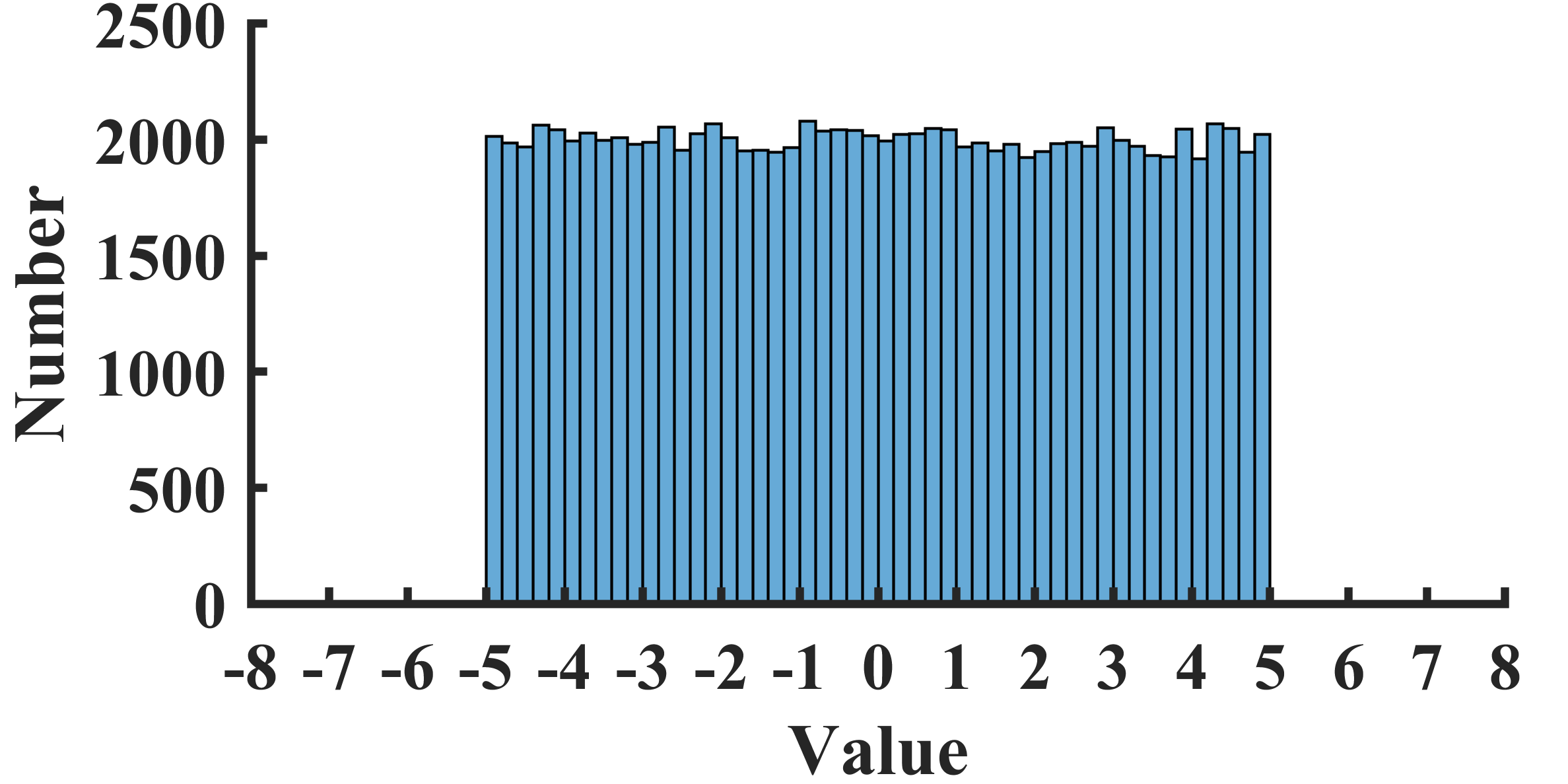}
    \includegraphics[width=0.22\textwidth]{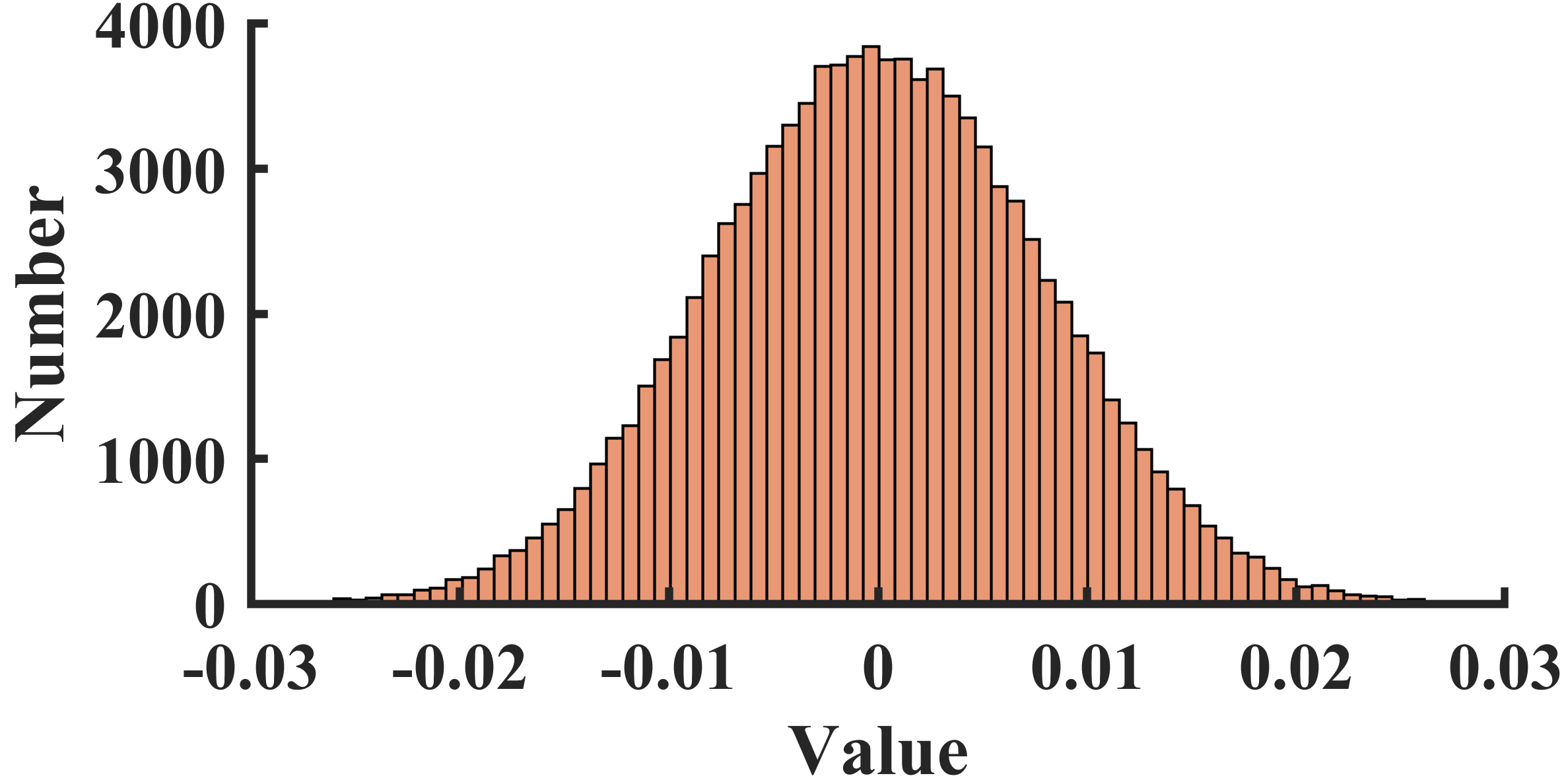}
    \\
    \raisebox{0.88\height}{\makebox[0.02\textwidth]{\rotatebox{90}{\makecell{\scriptsize \textbf{Laplace}}}}}
    \includegraphics[width=0.22\textwidth]{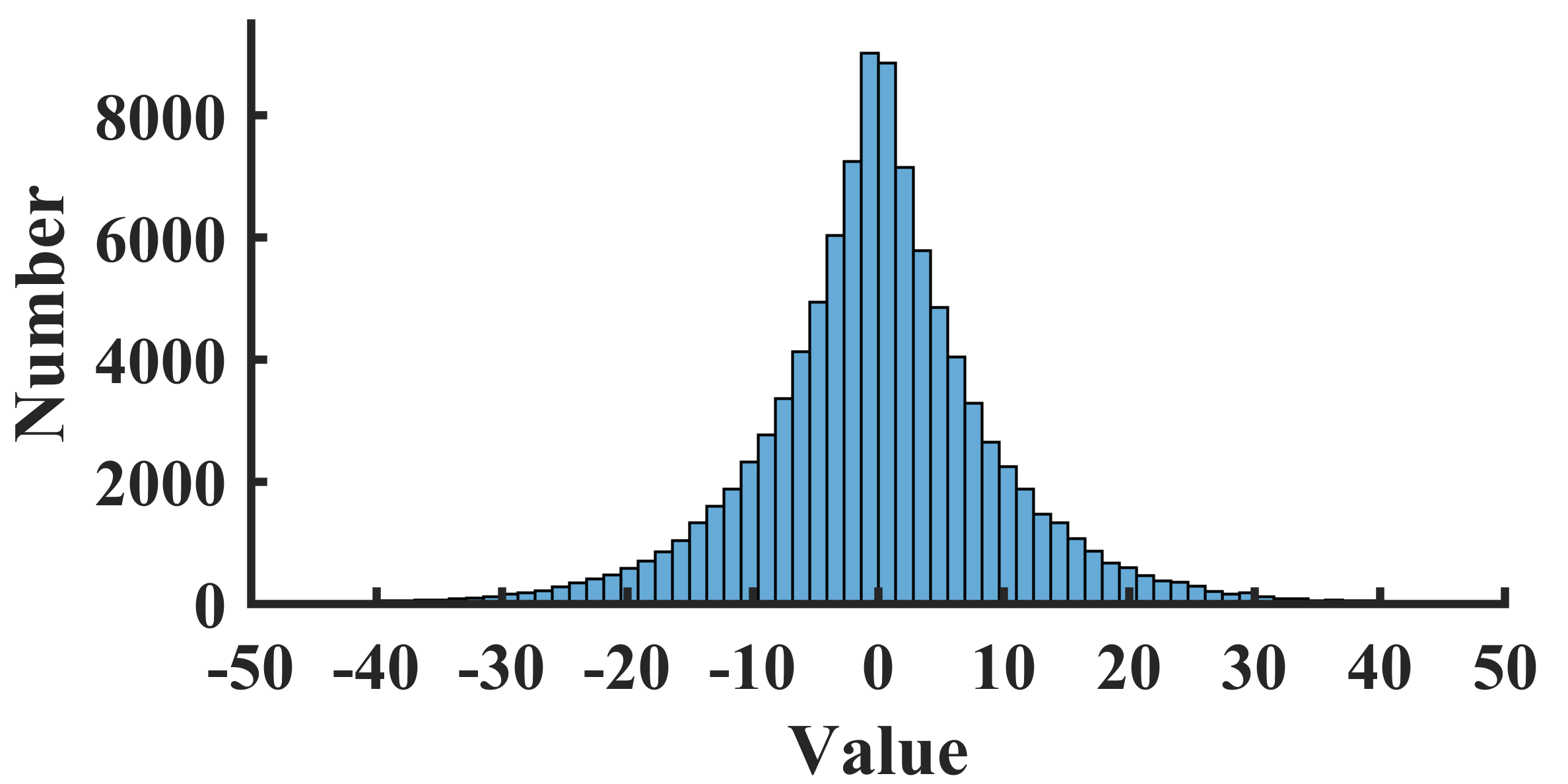}
    \includegraphics[width=0.22\textwidth]{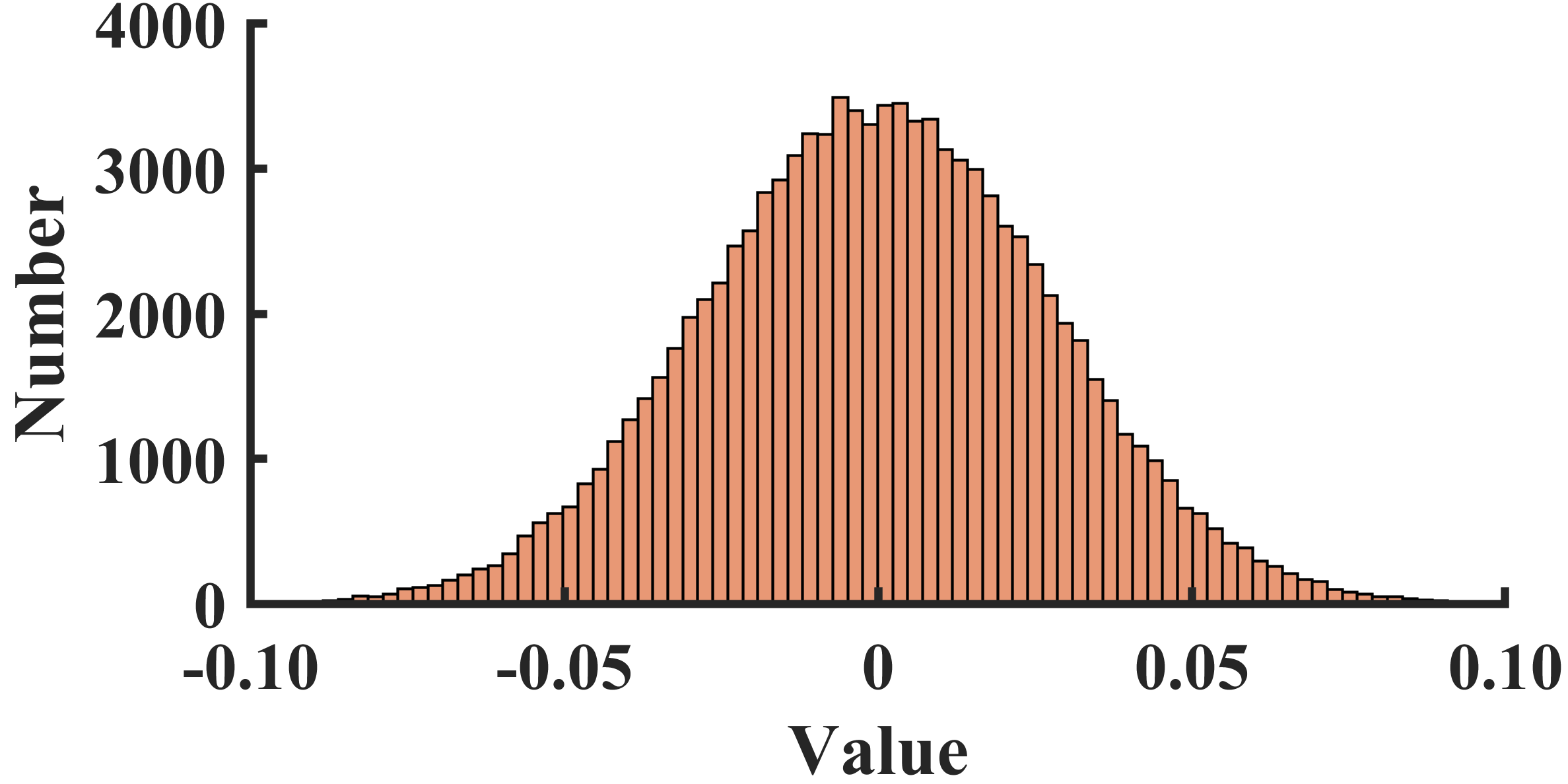}
    \\
    \caption{Histograms of noise in the spatial domain and its coefficient in Fourier domain. The histograms of the Fourier coefficients consistently approximate a Gaussian distribution for noise following different distributions in the spatial domain.}
    \label{fig: noise_analysis}
\end{figure}
% ------------------------------------------------------------------- %
\begin{remark}[Gaussian Approximation]
    Theorem~\ref{theorem: CLT-DFT-iid} indicates that, when the image size (\ie, $UV$) is sufficiently large, each Fourier coefficient of \iid~noise $n$ approximately follows a Gaussian distribution. In practical situations, $UV$ is generally much greater than $10^3$, which is large enough for this approximation. 
\end{remark}
\begin{remark}[Mean and Variance]
    A Gaussian distribution is characterized by its mean and variance. Since the Fourier transform is linear, the mean value of $\mathscr{F}(n)$ can be obtained by $\mathbb{E}[\mathscr{F}(n)] = \mathscr{F}(\mathbb{E}[n])$. Therefore, if the noise $n$ has a zero mean, its Fourier coefficients will also have a zero mean. Regarding the variance of $\mathscr{F}(n)$, it tends to be much smaller than the variance of $n$ due to the normalization factor $1 / UV$ in (\ref{eqn: DFT}). For the \iid~noise with a variance of $\sigma_n^2$, most of its Fourier coefficients will hold the variance of $\sigma_n^2 / 2UV$.
\end{remark}

To illustrate our analysis, a statistical experiment is conducted on several noise distributions, and the results are visualized in Fig.~\ref{fig: noise_analysis}. In this experiment, we at first repeatedly generate noise instances $n$ of size $255 \times 255$ from a candidate distribution. Then, we calculate $\mathscr{F}(n)$ for each noise instance. Finally, we select a fixed pixel and a Fourier coefficient to plot the histograms. In the case of Poisson noise, it was generated based on a selected image, and the mean value was subtracted to maintain a zero mean. As Fig.~\ref{fig: noise_analysis} shows, the histograms of the Fourier coefficients approximate Gaussian distributions for noise following different distributions in the spatial domain. This empirical evidence supports our analysis regarding the noise distribution in the Fourier domain.

% ------------------------------------------------------------------- %
% fig: statistaical equivalence

% \begin{figure}[!t]
% \centering

% \subfigure[Original function]{\includegraphics[width=0.48\linewidth]{figs/stripe_noise/periodic_psd.png}%
% % \label{1}
% }
% \subfigure[Equivalent function]{\includegraphics[width=0.48\linewidth]{figs/stripe_noise/row_psd.png}%
% % \label{2}
% }

% \caption{Stripe noise and their power spectral density.}
% \label{fig: psd}
% \end{figure}

\begin{figure}[t] \centering
    \makebox[0.02\textwidth]{}
    \makebox[0.21\textwidth]{\scriptsize \boldmath$n$}
    \makebox[0.22\textwidth]{\scriptsize \boldmath$\text{var}[\mathscr{F}(n)]$}
    % \vspace*{-.1em}
    \\
    \raisebox{0.7\height}{\makebox[0.02\textwidth]{\rotatebox{90}{\makecell{\scriptsize \textbf{Periodic noise}}}}}
    \includegraphics[width=0.21\textwidth]{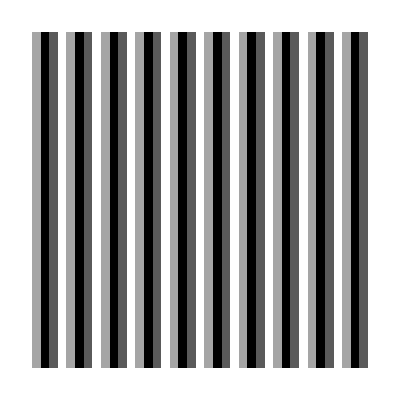}
    \includegraphics[width=0.22\textwidth]{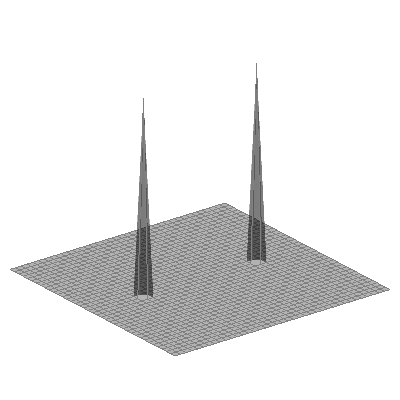}
    \\
    \raisebox{1.2\height}{\makebox[0.02\textwidth]{\rotatebox{90}{\makecell{\scriptsize \textbf{Row noise}}}}}
    \includegraphics[width=0.21\textwidth]{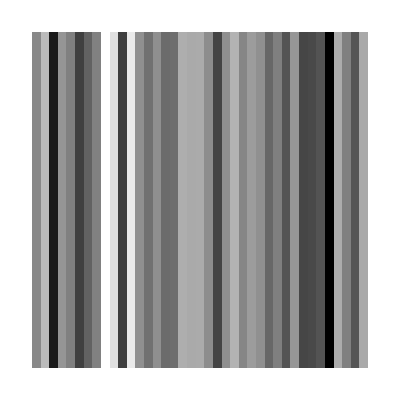}
    \includegraphics[width=0.22\textwidth]{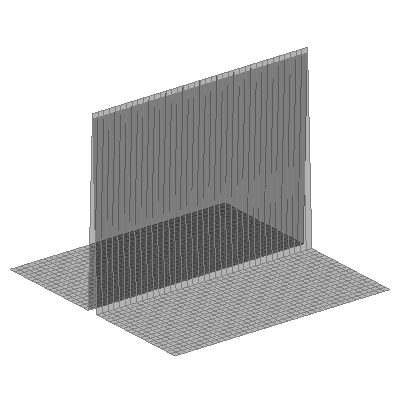}
    \\
    % \caption{Stripe noise and their power spectral density. }
    \caption{Visualization of periodic noise, row noise, and their corresponding $\text{var}[\mathscr{F}(n)]$. The row noise and periodic noise are generated in the same way as~\cite{SEID}~and~\cite{Low-Light}, respectively. Their $\text{var}[\mathscr{F}(n)]$ are zero at most of positions, only accumulating at a few positions.}
    \label{fig: psd}
\end{figure}
% ------------------------------------------------------------------- %
% convolution theorem
\subsubsection{Spatially-correlated Noise}
For the spatially-correlated noise, we focus our discussions on the stationary noise model shown in~(\ref{eqn: NM}), \ie, $n = h * \eta$, where $\eta$ is assumed to be zero-mean \iid~noise. Despite its simplicity in mathematical form, this model effectively describes various real-life noise sources such as stripe noise, band-limited noise, and pink noise~\cite{VSNR, SN-BM3D}. For these noise, we have the following theorem.
\begin{theorem} \label{theorem: CLT-DFT-stationary}
    Suppose the spatially-correlated noise $n$ can be modeled by~(\ref{eqn: NM}). As $UV$ approaches infinity, each Fourier coefficient of $n$ will (i) converge in distribution to the Gaussian distribution and (ii) be independent with other coefficients.
\end{theorem}
\begin{proof}
    The proof is provided in the Appendix~\ref{proof: CLT-DFT-stationary}.
\end{proof}
\begin{remark}[Mean and Variance]
    For the mean and variance of $\mathscr{F}(n)$,  
    \begin{equation}    
        \mathbb{E}[\mathscr{F}(n)] 
        = \mathscr{F}(h) \cdot \mathscr{F}(\mathbb{E}[\eta]) = 0,
    \end{equation}
    and
    \begin{equation}
        \text{var}[\mathscr{F}(n)] = |\mathscr{F}(h)|^2 \cdot \text{var}[\mathscr{F}(\eta)].
    \end{equation}
    An important observation from \cite{SN-BM3D} is that $|\mathscr{F}(h)|$ for diverse correlation kernels $h$ exhibits sparsity. This means that for diverse stationary noise, only a small fraction of Fourier coefficients are affected, while the remaining coefficients remain uncorrupted. To illustrate this, Fig.~\ref{fig: psd} visualizes the periodic noise, row noise, and their corresponding variances in the Fourier domain.
\end{remark}

In summary, we have concluded that the Fourier coefficients of a wide range of noise approximately follow the Gaussian distribution if the image size $UV$ is large enough. This allows for a unified treatment of different types of noise under a common framework. Additionally, we have pointed out that noise exhibiting spatial correlation can maintain both the independence and the sparsity in the Fourier domain. These findings emphasize the advantages of utilizing the Fourier domain for noisy supervision and provide insights into the development of effective noise mitigation strategies.

% ------------------------------------------------------------------- %
% fig: statistaical equivalence

% \begin{figure}[!t]
% \centering

% \subfigure[Original function]{\includegraphics[width=0.48\linewidth]{figs/equivalence/varphi.png}%
% % \label{1}
% }
% \subfigure[Equivalent function]{\includegraphics[width=0.48\linewidth]{figs/equivalence/varphi_p.png}%
% % \label{2}
% }

% \caption{Statistical equivalence. Utilizing $\varphi(t)$ with noisy targets is statistically equivalent to utilizing $(\varphi*p)(t)$ with clean targets. The original function is $\varphi(t) = |t|$. The $p(t)$ corresponds to the probability function of a Gaussian distribution with $\sigma=0.2$.}
% \label{fig: equivalence}
% \end{figure}

\begin{figure}[t] \centering
    \makebox[0.02\textwidth]{}
    \makebox[0.22\textwidth]{\scriptsize \quad \quad \boldmath$\varphi(t)=|t|$}
    \makebox[0.22\textwidth]{\scriptsize \quad \quad \textbf{Huber function}}
    \vspace{.5em}
    \\
    \raisebox{0.75\height}{\makebox[0.02\textwidth]{\rotatebox{90}{\makecell{\scriptsize \textbf{Original}}}}}
    \includegraphics[width=0.22\textwidth]{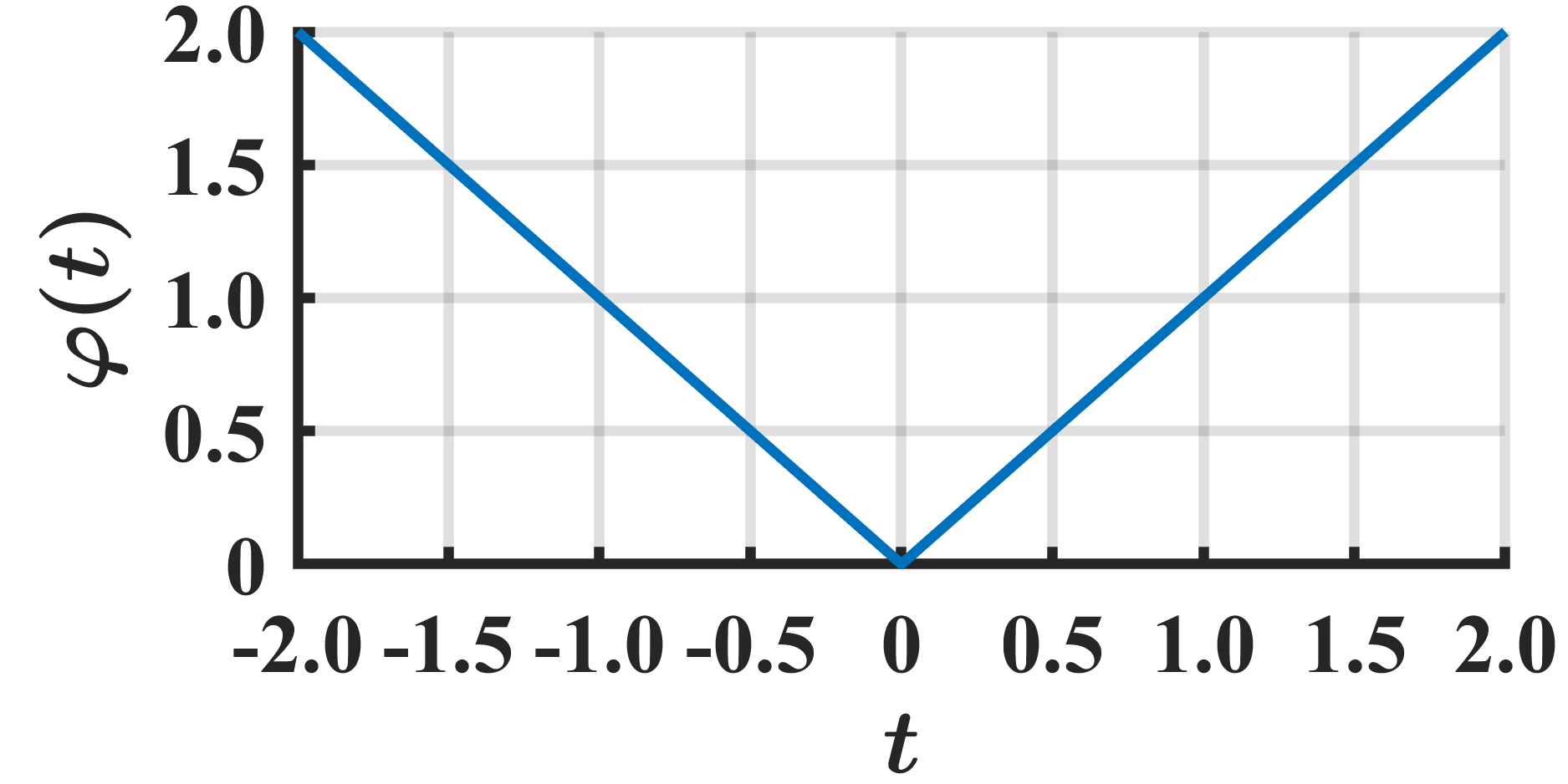}
    \includegraphics[width=0.22\textwidth]{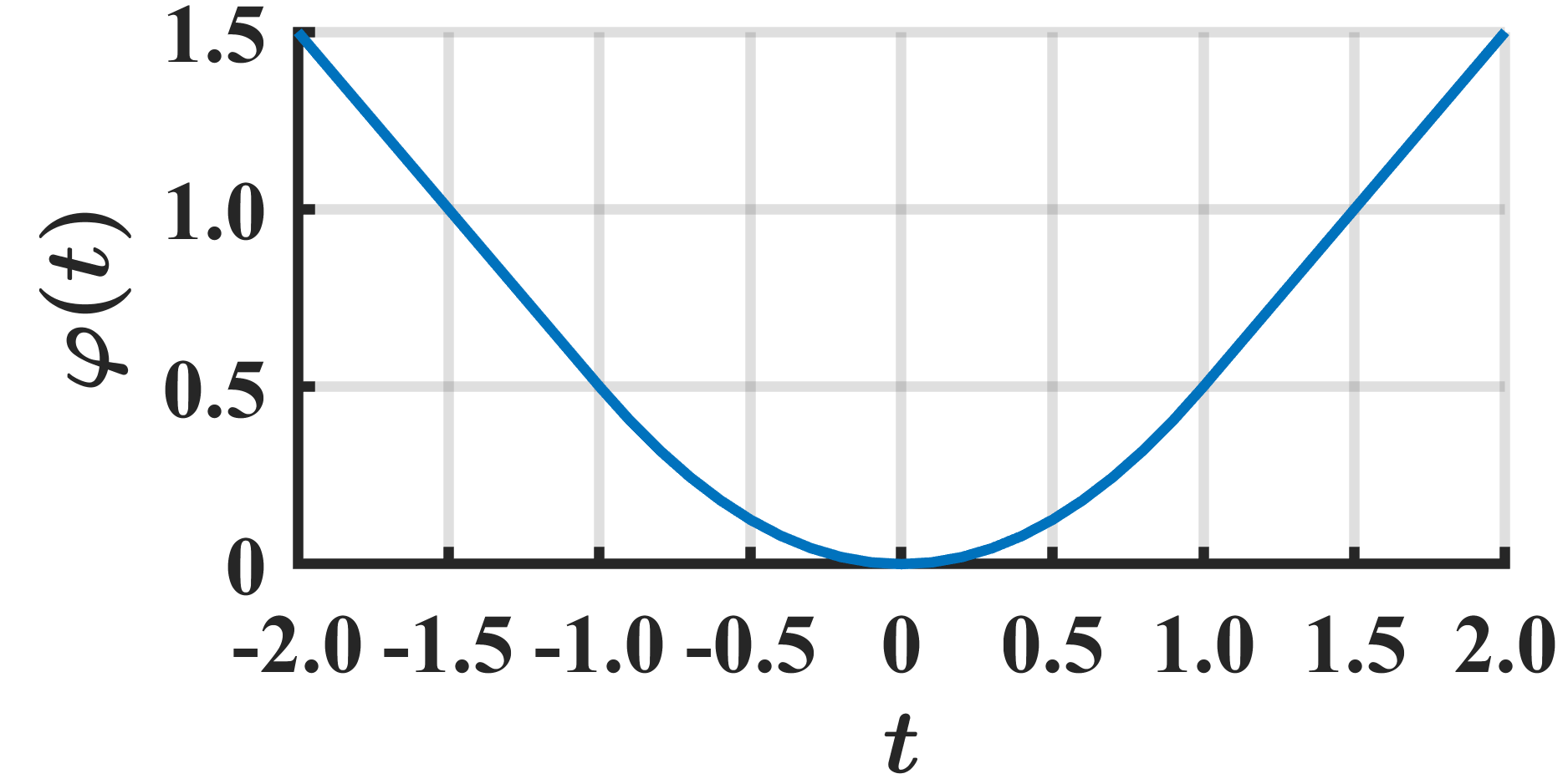}
    \\
    \raisebox{0.5\height}{\makebox[0.02\textwidth]{\rotatebox{90}{\makecell{\scriptsize \textbf{Equivalent}}}}}
    \includegraphics[width=0.22\textwidth]{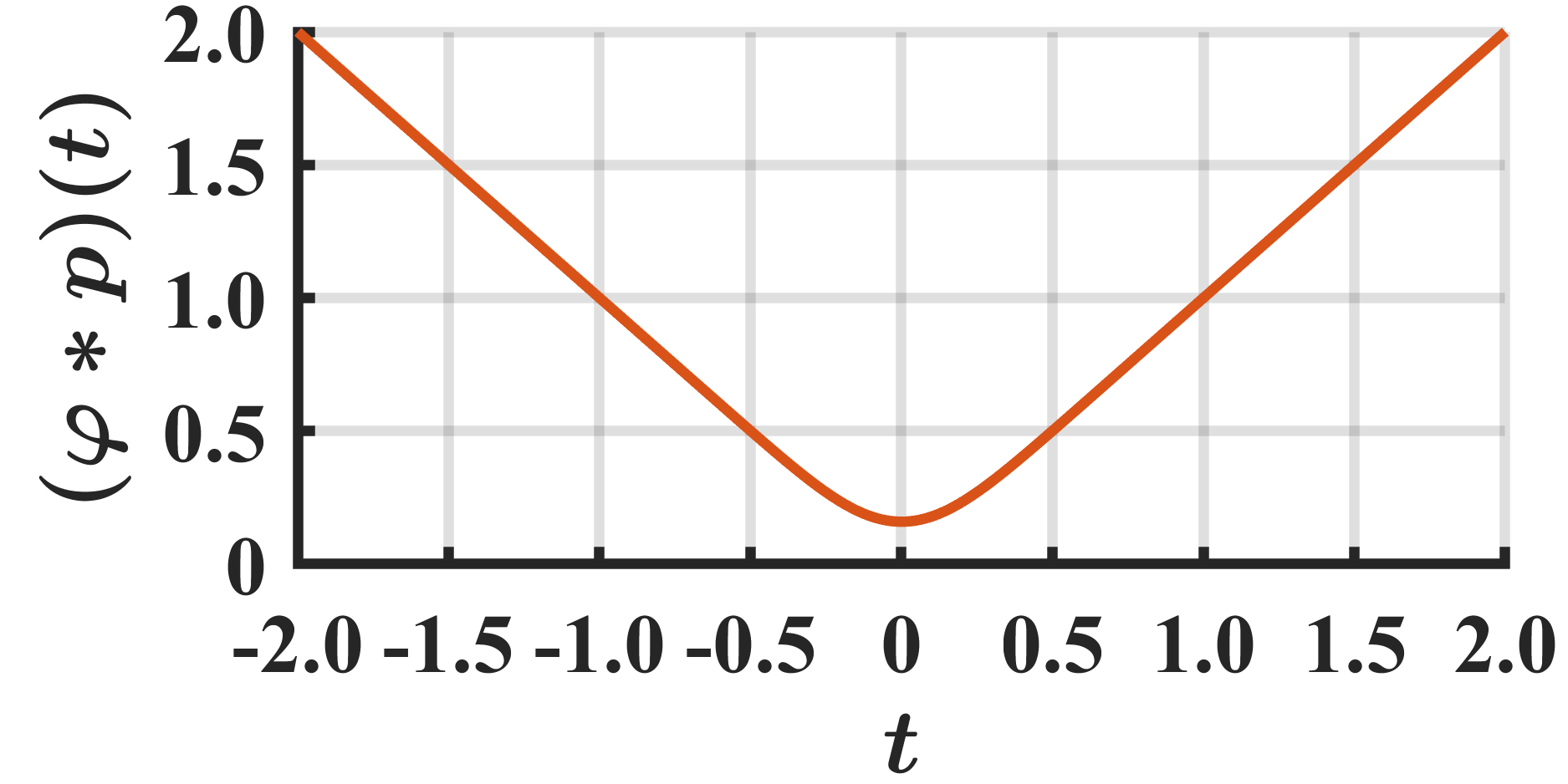}
    \includegraphics[width=0.22\textwidth]{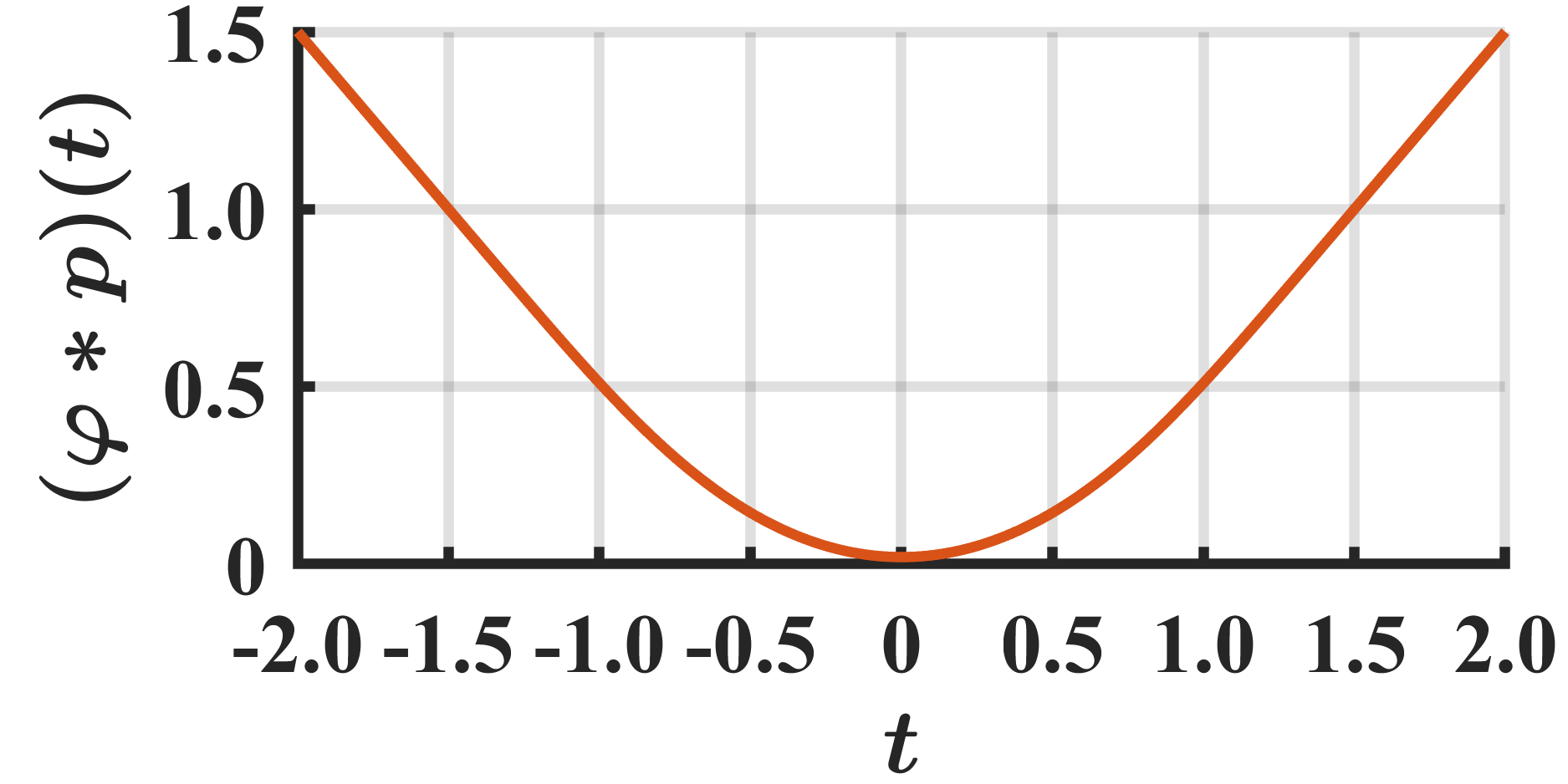}
    \\
    \raisebox{0.5\height}{\makebox[0.02\textwidth]{\rotatebox{90}{\makecell{\scriptsize \textbf{Derivative}}}}}
    \includegraphics[width=0.22\textwidth]{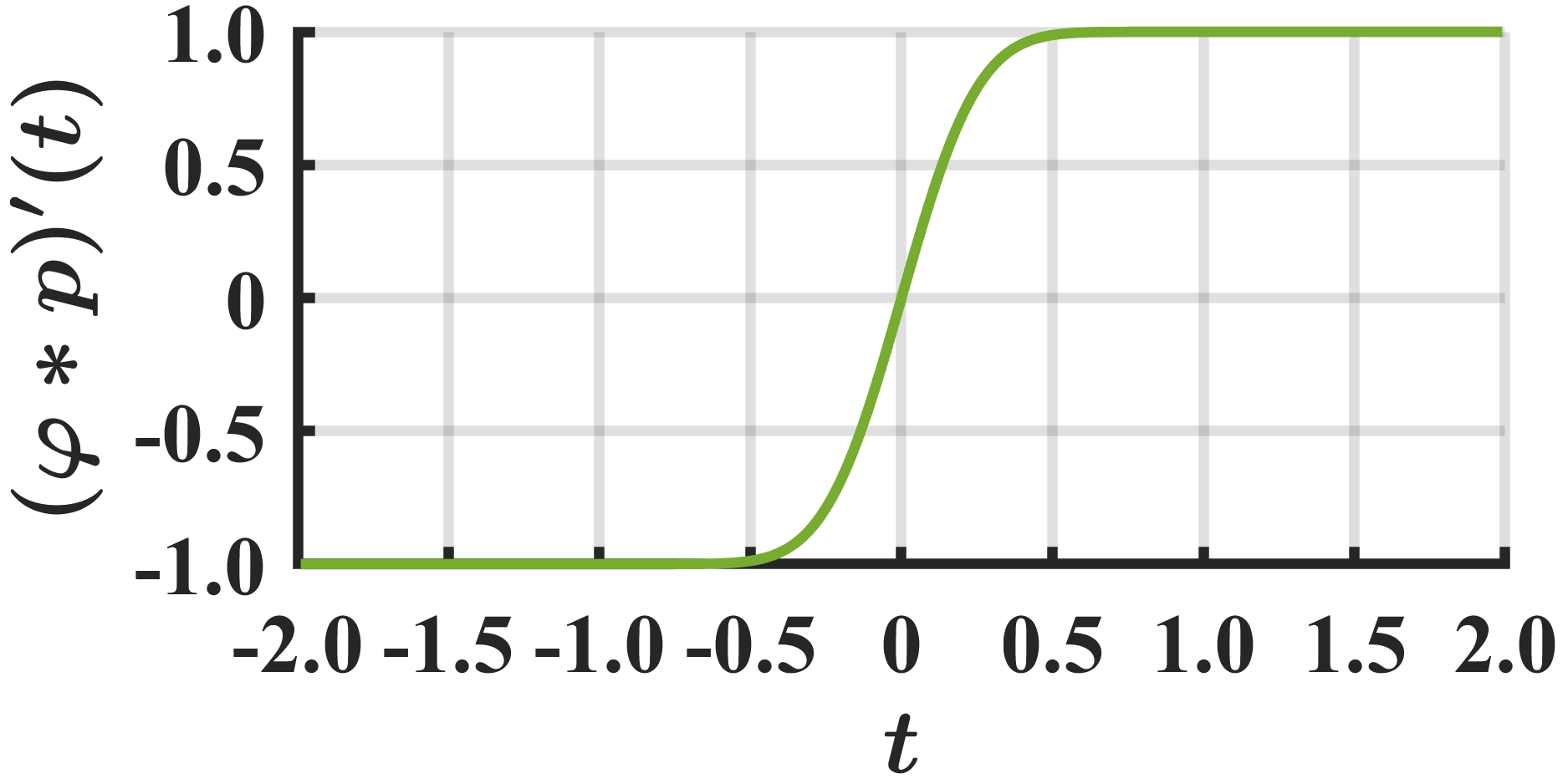}
    \includegraphics[width=0.22\textwidth]{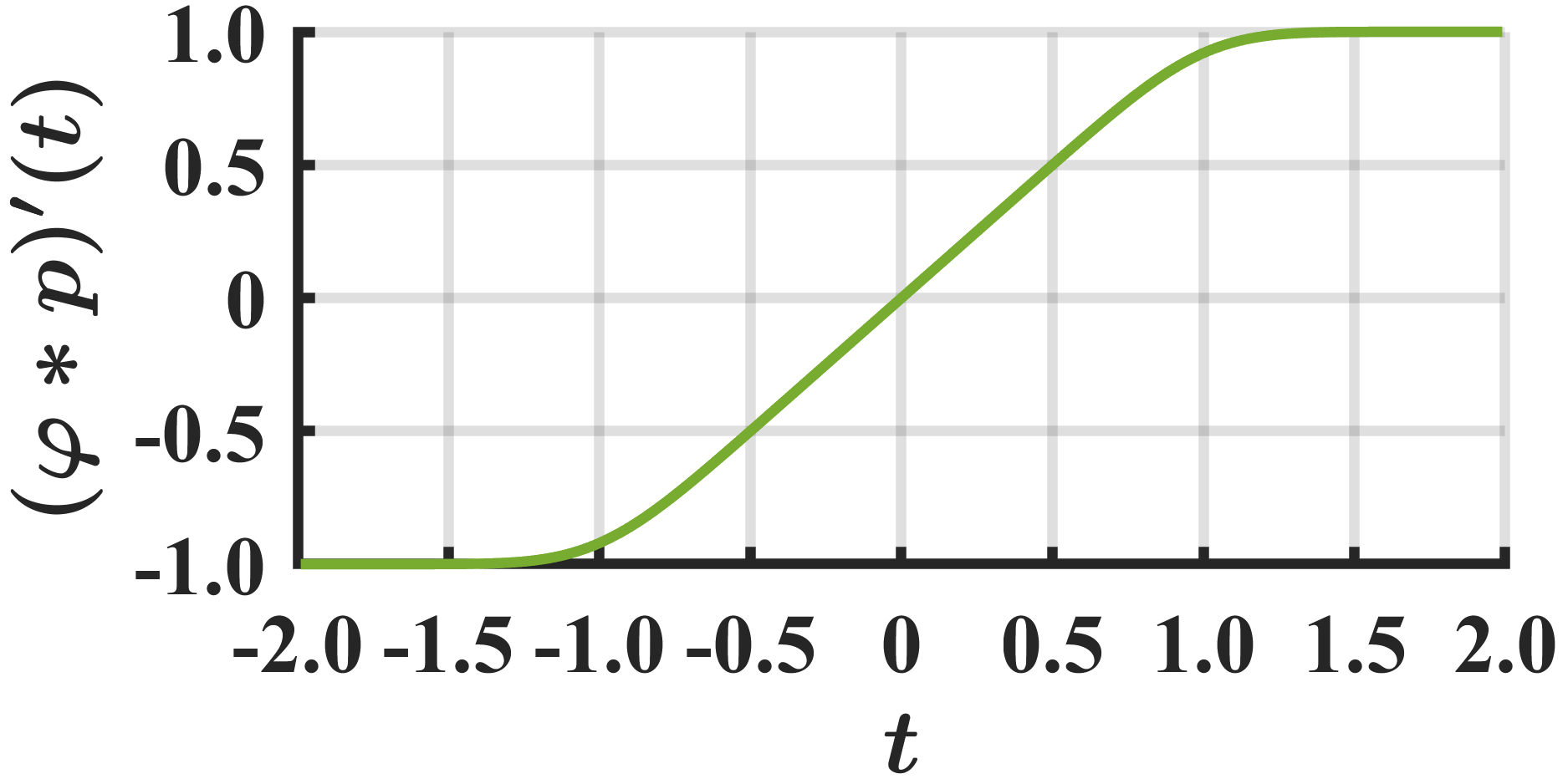}
    \\
    \caption{Statistical equivalence. Utilizing $\varphi(t)$ with noisy targets is statistically equivalent to utilizing $\phi(t)=\varphi(t)*p(t)$ with clean targets. The $p(t)$ corresponds to the zero-mean Gaussian distribution with $\sigma=0.2$.}
    \label{fig: equivalence}
\end{figure}
% ------------------------------------------------------------------- %
\subsection{Statistical Equivalence in the Fourier domain}
This subsection aims to demonstrate the efficacy of utilizing noisy targets in the Fourier domain to supervise image restoration learning, regarding which we have the following theorem.
\begin{theorem} \label{theorem: DFT-SE}
    Suppose the noise $n$ contained in the target $y$ is independent with the input $x$ and each of its Fourier coefficients (approximately) follows a zero-mean Gaussian distribution $p(t)$. Define the loss function with the network output $f_{\theta}(x)$ and the noisy target image $y=z+n$ in the Fourier domain as 
    \begin{align}   \label{eqn: loss}
        L_{\varphi}\big(f_{\theta}(x), y\big) 
        & = \sum\limits_{k=0}^{U-1} \sum\limits_{l=0}^{V-1} \varphi\Big(a\big(f_\theta(x)\big)[k,l] - a(y)[k,l]\Big)
        \notag
        \\
        & + \sum\limits_{k=0}^{U-1} \sum\limits_{l=0}^{V-1} \varphi\Big(b\big(f_\theta(x)\big)[k,l] - b(y)[k,l]\Big),
    \end{align}
    where $\varphi(t)$ denotes a penalty function satisfying $\varphi(t_1) < \varphi(t_2)$ for $ |t_1| < |t_2| $. The following two statements hold.

    % \noindent 
    (i) For the loss function given by (\ref{eqn: loss}), utilizing a penalty function $\varphi(t)$ with noisy target $y$ is statistically equivalent to utilizing its blurred version $\phi(t)=\varphi(t)*p(t)$ with clean target $z$, leading to a statistical equivalence as
    \begin{equation}    \label{eqn: DFT-SE}
      \mathbb{E}\left[L_{\varphi}\big(f_{\theta}(x), y\big)\right] 
      = 
      L_{\phi}\big(f_{\theta}(x), z\big),
    \end{equation}
    \revision{where $L_{\phi}\big(f_{\theta}(x), z\big)$ is defined similarly to (\ref{eqn: loss}), with the replacement of $\varphi(t)$ by $\phi(t)$.}
    
    % \noindent 
    (ii) The optimal point for (\ref{eqn: DFT-SE}) is $f_{\theta}^{\star}(x)=z$.
\end{theorem}
\begin{proof}
    The proof is provided in the Appendix~\ref{proof: DFT-SE}.
\end{proof}

\begin{remark}[Statistical Equivalence]
This theorem indicates that supervising the image restoration learning with noisy targets $y$ is statistically equivalent to the one with clean targets $z$, sharing the same optimal points. 
\end{remark}
\begin{remark}[Noise Assumption]
The noise assumption involved in this theorem is very mild. When the input image $x$ and the target image $y$ are captured from two independent observations, the independence between the noise $n$ and the input $x$ will be satisfied. Furthermore, based on our analysis above, the Fourier coefficients of a wide range of noise models will (approximately) follow the Gaussian distribution when the image size is large enough. Consequently, this theorem has a wide applicability. It is applicable to various noise that is commonly observed in practical applications, such as the Poisson-Gaussian noise, the \iid uniform noise, and the stripe noise.
\end{remark}
\begin{remark}[Penalty Function]
This theorem requires the penalty function to satisfy $\varphi(t_1) < \varphi(t_2)$ for $ |t_1| < |t_2| $. Some representative penalty functions satisfying this requirement include $\varphi(t) = |t|^q$ for any $q>0$ and the Huber function~\cite{Huber}. In the subsequent experimental section, we select the Huber function for image denoising and $\varphi(t)=|t|$ for image super-resolution and deblurring. In Fig.~\ref{fig: equivalence}, the original function $\varphi(t)$, the equivalent function $\phi(t)=\varphi(t)*p(t)$, and the derivative $\phi'(t)$ for these two functions are plotted to illustrate our analysis of the statistical equivalence. Clearly, it can be observed that $\varphi(t)$ and $\phi(t)$ only slightly differ in a small region, implying that they will lead to similar learning results.
\end{remark}
\begin{remark}[Practical Implementation]
The statistical equivalence shown in (\ref{eqn: DFT-SE}) involves an expectation over the noise $n$, the computation of which is time-consuming and relies on the collection of multiple independent noisy targets for each input image from the same scene. Fortunately, it is unnecessary to compute this expectation in practice. We observe that, when the network is trained from a set of training pairs, using a single noisy target for each input is enough for a good learning equivalence. The loss function we implement in the subsequent experimental section is
\begin{equation}    \label{eqn: DFT-SE-Practical}
    \sum\limits_{s=1}^{S} L_{\varphi}\big(f_{\theta}(x_s), y_s\big),
\end{equation}
where $S$ denotes the total number of training pairs, $x_s$ and $y_s$ denotes the $s$-th training image pair. Note that different training pairs can be captured from various scenes. The reasonability of replacing (\ref{eqn: DFT-SE}) by this loss function is that, for each image patch that is relatively small, there are generally lots of similar patches within the same image and across the whole training sets. Consequently, the loss function calculated on these similar image patches will implicitly have the same effect as the expectation over noise $n$.

% and a stochastic optimization method such as Adam~\cite{Adam} is adopted
\end{remark}

\section{Experimental Results}
\subsection{Experimental Setting}

\begin{figure}[t] \centering
    \makebox[0.02\textwidth]{}
    \makebox[0.15\textwidth]{\scriptsize \textbf{Denoising}}
    \makebox[0.15\textwidth]{\scriptsize \textbf{Deblurring}}
    \makebox[0.15\textwidth]{\scriptsize \textbf{Super-resolution}}
    \\
    \vspace{0.1em}
    \raisebox{1.5\height}{\makebox[0.02\textwidth]{\rotatebox{90}{\makecell{\scriptsize \textbf{Input}}}}}
    \includegraphics[width=0.15\textwidth]{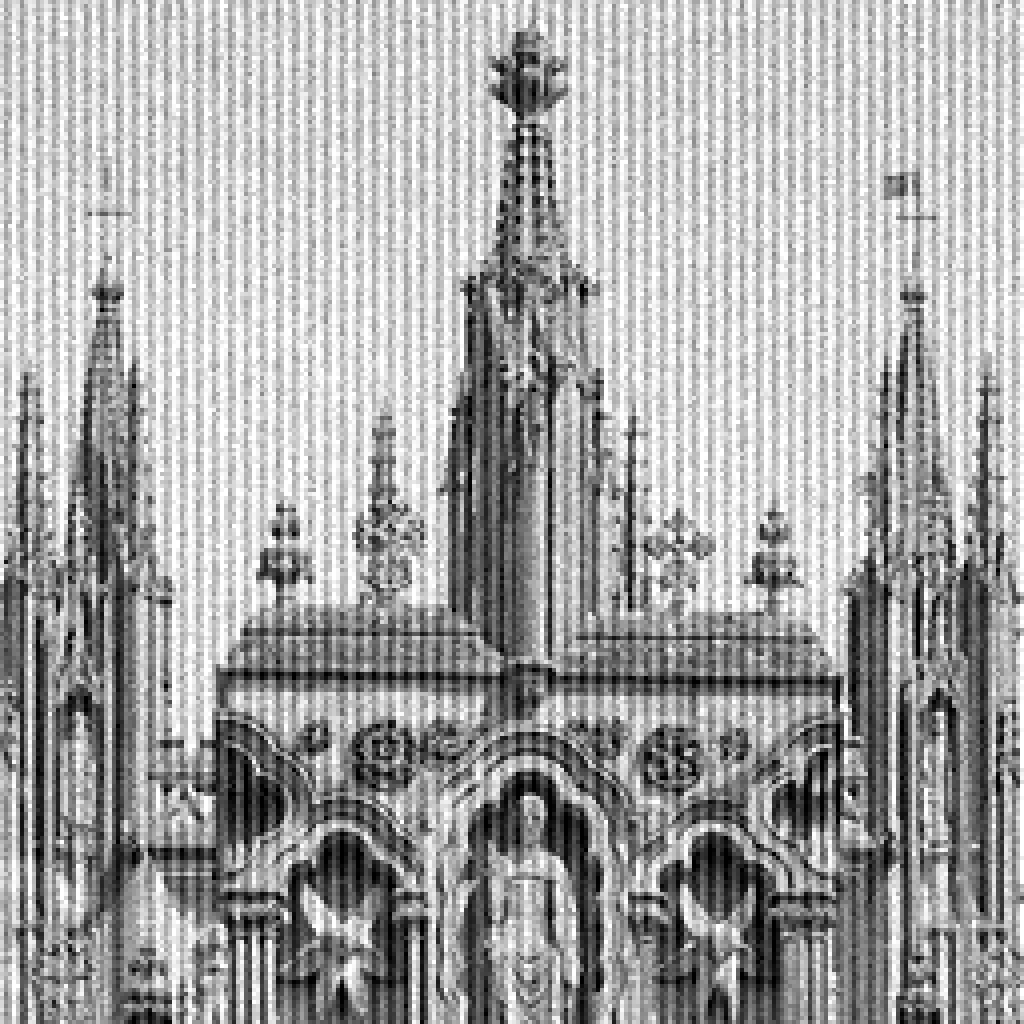}
    \includegraphics[width=0.15\textwidth]{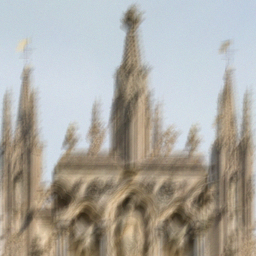}
    \raisebox{0.5\height}{\makebox[0.15\textwidth]{\centering \includegraphics[width=0.075\textwidth]{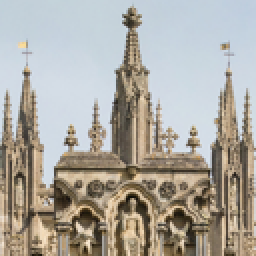}}}
    \\
    \vspace{0.2em}
    \raisebox{1.5\height}{\makebox[0.02\textwidth]{\rotatebox{90}{\makecell{\scriptsize \textbf{Target}}}}}
    \includegraphics[width=0.15\textwidth]{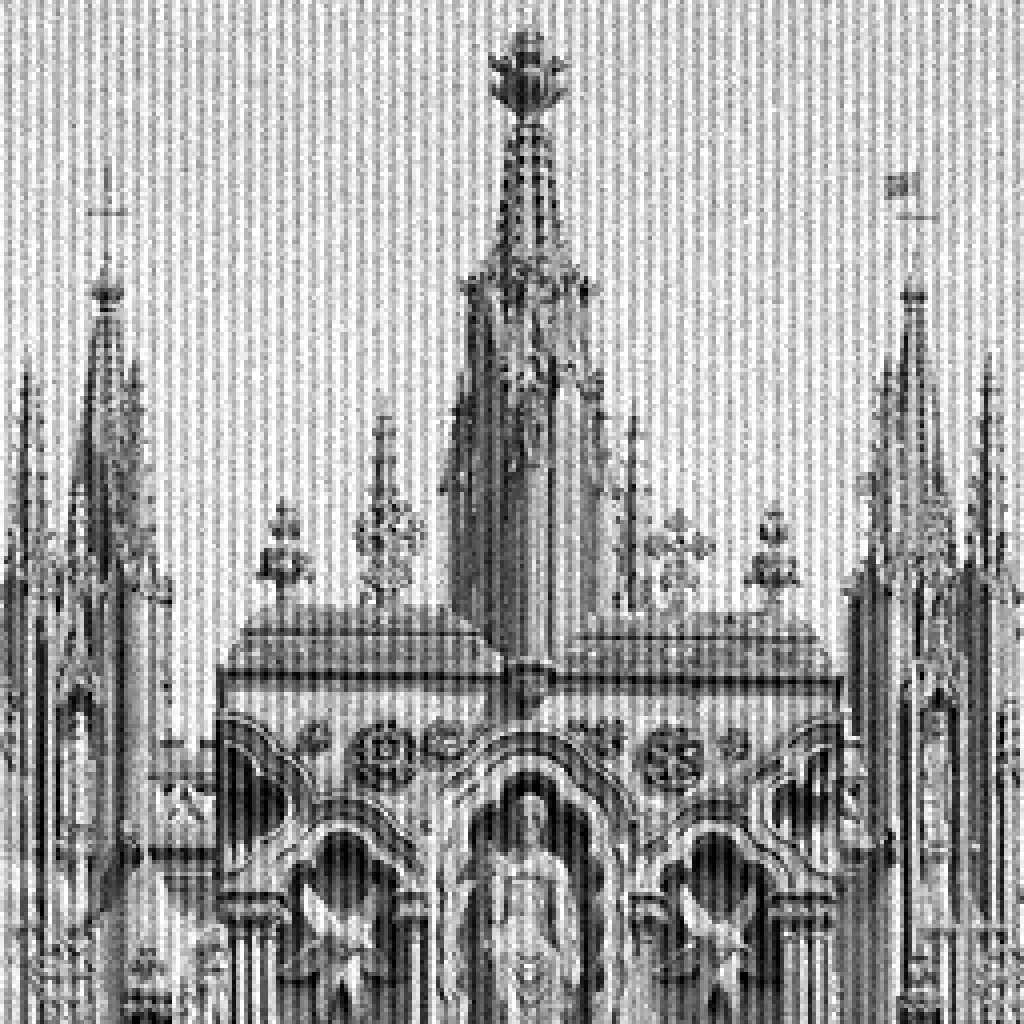}
    \includegraphics[width=0.15\textwidth]{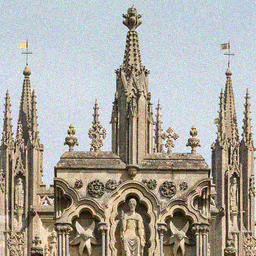}
    \includegraphics[width=0.15\textwidth]{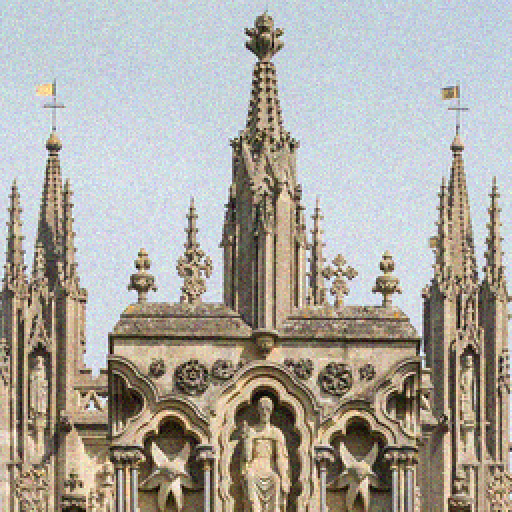}
    \\
    \caption{Illustration of training pairs for DN, SR, and DB. These images are generated from 0149 in DIV2K.}
    \label{fig: training_pair}
\end{figure}
% ------------------------------------------------------------------- %

This subsection provides a detailed introduction to the experimental settings including image restoration tasks, training/testing data, comparison methods, evaluation metrics, network architectures, and training settings.

\emph{Image Restoration Tasks.} Our proposed learning framework, referred to as IR-NSF, is evaluated on three distinct image restoration tasks: grayscale image denoising (DN), color image super-resolution (SR), and color image deblurring (DB). For DN, we consider a mixture of \iid~Gaussian noise of $\sigma=5$, Poisson noise, and the periodic noise shown in Fig.~\ref{fig: psd}. For SR, we simulate the low-resolution images by downsampling from higher resolution ones. For DB, as suggested by~\cite{RDN}, we generate blurry images with eight motion blur kernels from~\cite{Blur_Kernel} and further add \iid~Gaussian noise with a standard deviation of $\sigma=2$.

\emph{Training and Testing Data.} In our synthetic experiments, we synthesize degraded input images and noisy target images using clean images sourced from publicly available datasets. During the training phase, we employ 800 4K-resolution training images from the DIV2K dataset~\cite{DIV2K}. To accelerate the data loading process, these images are pre-processed by cropping them into 22988 image patches of size $480 \times 480$. During the testing phase, we utilize images from the Set5~\cite{Set5}, Set14~\cite{Set14}, Kodak24\footnote{http://r0k.us/graphics/kodak/}, McMaster~\cite{McMaster}, Urban100~\cite{Urban100}, and Manga109~\cite{Manga109} datasets. For DN, the source images are converted to grayscale, while for SR and DB, the color information is retained. The degraded input images are synthesized using the aforementioned degradation models, and the noisy target images are generated by adding synthetic noise of the specified model to the latent clean image. For DN, the noise model is identical for both the input and target images. For SR and DB, Gaussian noise with a standard deviation of $\sigma=10$ is utilized for generating noisy targets. \revision{Without explicit specification, the noise generation is carried out on the fly during the training process.} In Fig.~\ref{fig: training_pair}, training pairs for different tasks are illustrated.

\emph{Comparison Methods.} Our proposed IR-NSF framework is designed for image restoration learning using noisy targets. To ensure fairness, we primarily compare IR-NSF with Noise2Noise (N2N)~\cite{N2N}, which adopts the same setting as ours. To compare the statistical equivalences established by N2N and IR-NSF, we also apply both frameworks to clean targets. These alternative cases are referred to as N2C and IR-CSF, respectively.
 
\emph{Evaluation Metrics.} We employ the Peak Signal-to-Noise Ratio (PSNR) and Structural Similarity Index (SSIM) as quantitative evaluation metrics for the restored images in all tasks. However, the calculation of these metrics varies depending on the specific image restoration task. For DN and DB, the metrics are calculated using all color channels of the restored images. For SR, the calculation follows the approach in~\cite{EDSR}, where the PSNR and SSIM are measured on the Y channel of the YCbCr color space, and an equal number of border pixels are ignored corresponding to the scaling factor.

\emph{Network Architectures.} To validate that our proposed learning framework can be applied to diverse network architectures, we utilize two different networks for each image restoration task. Specifically, we employ ResNet (for DN), RDN~\cite{RDN} (for DN), EDSR~\cite{EDSR} (for SR), RCAN~\cite{RCAN} (for SR), MIMO~\cite{MIMO} (for DB), and MIMO+~\cite{MIMO} (for DB). In the case of ResNet, it is a modified version of EDSR where the up-sample block is removed, and a long skip connection is added between the input and output. The remaining networks align with the descriptions provided in their respective papers.

\emph{Training Settings.}  The training of these networks is conducted based on the source codes provided by their respective authors. The default settings are maintained, except for the modifications introduced here. For MIMO and MIMO+, the number of training epochs is adjusted to 300, with the learning rate halved every 50 epochs. All experiments are performed on the AMD Ryzen Threadripper PRO 3955WX CPU and the NVIDIA GeForce RTX 4090 GPU.
% ------------------------------------------------------------------- %
% fig: convergence curve
% ------------------------------------------------------------------------ %
% for package subfigure
% ------------------------------------------------------------------------ %
\begin{figure}[!t]
\centering

\subfigure[DN experiment with ResNet]{\includegraphics[width=0.98\linewidth]{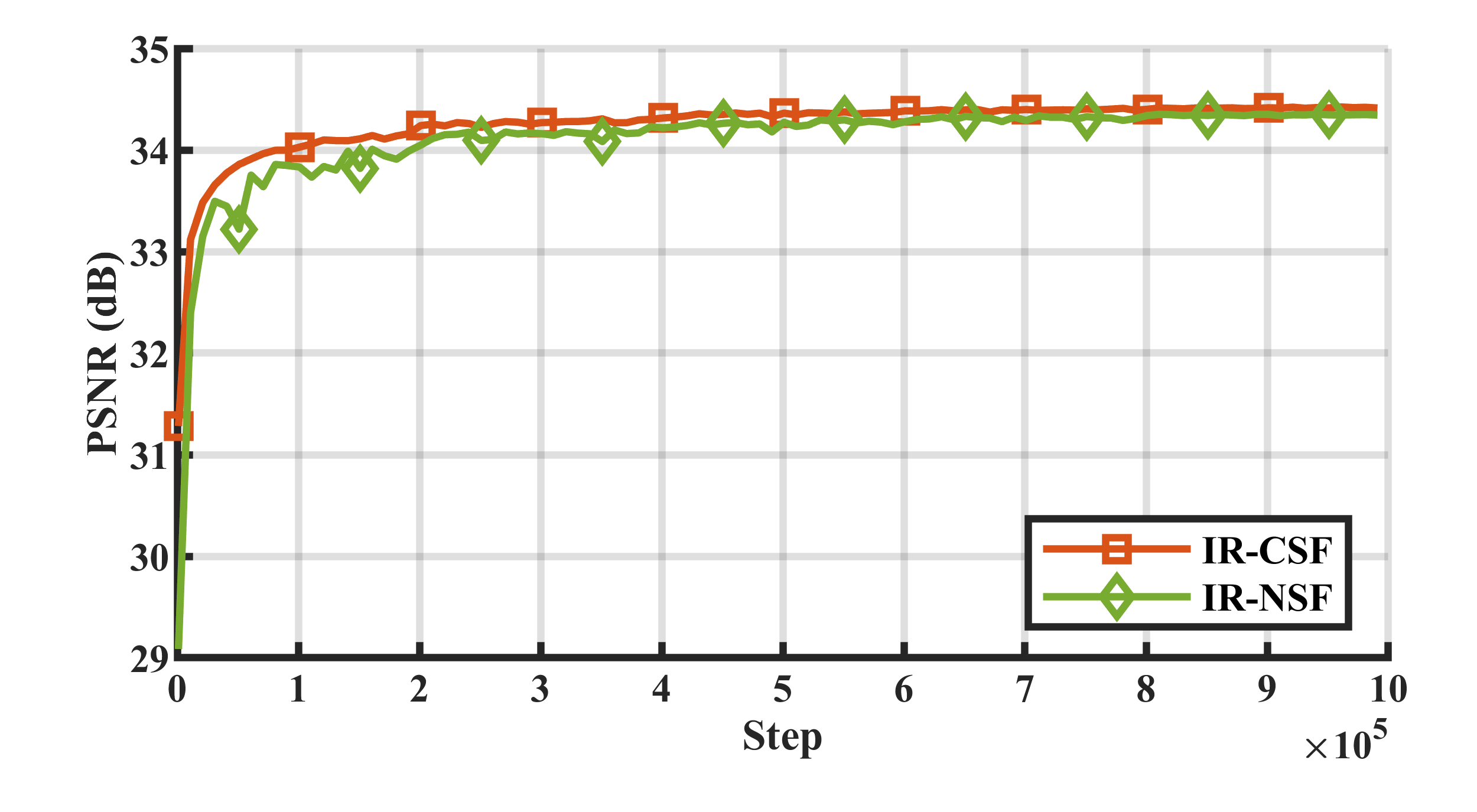}%
% \label{1}
}

\subfigure[SR experiment with EDSR]{\includegraphics[width=0.98\linewidth]{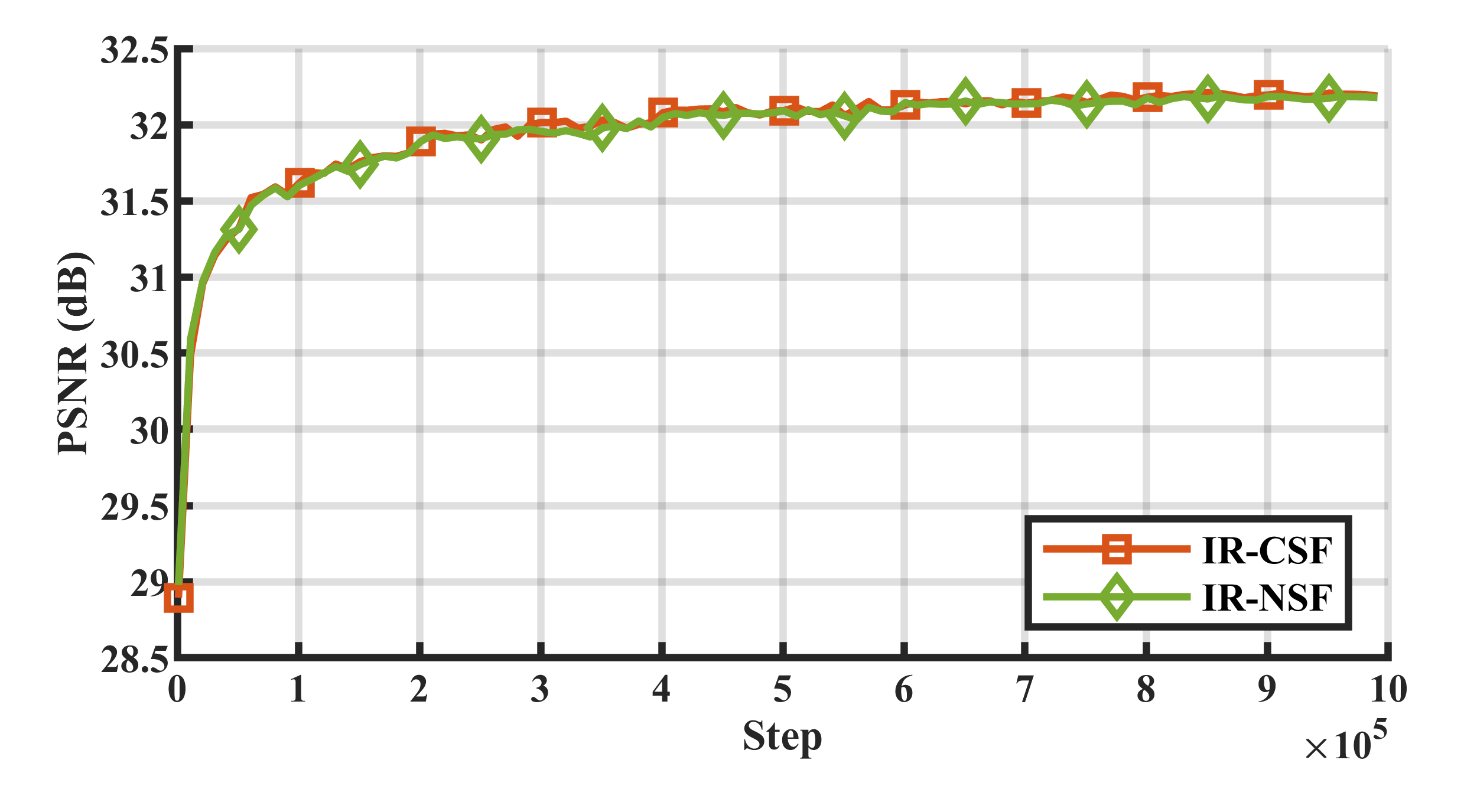}%
% \label{2}
}

\caption{PSNR (\SI{}{dB}) on Urban100 as a function of training step. In both cases, training with noisy targets (\ie, IR-NSF) and clean targets (\ie, IR-CSF) converge to almost the same results.}
\label{fig: convergence}
\end{figure}

% % ------------------------------------------------------------------------ %
% % for package subfigure
% % ------------------------------------------------------------------------ %
% \begin{figure}[!t]
% \centering

% \makebox[0.98\linewidth]{
% \includegraphics[width=0.98\linewidth]{figs/convergence_curve/DN-ResNet.eps}
% }
% \raisebox{1em}{\makebox[0.98\linewidth]{(a) DN experiment with ResNet}}

% \makebox[0.98\linewidth]{
% \includegraphics[width=0.98\linewidth]{figs/convergence_curve/SR-EDSR.eps}
% }
% \raisebox{1em}{\makebox[0.98\linewidth]{(b) SR experiment with EDSR}}

% \caption{PSNR ($dB$) on Urban100 as a function of training step.}
% \label{fig: convergence}
% \end{figure}
% ------------------------------------------------------------------- %

% ------------------------------------------------------------------- %
% fig: visual comparison of denoising (DN) task
% \input{figs/visual_cmp/fig_dn}
\begin{figure*}[t] \centering
    \makebox[0.02\textwidth]{}
    \makebox[0.18\textwidth]{\scriptsize \textbf{Ground Truth}}
    \hspace{.1em}
    \makebox[0.18\textwidth]{\scriptsize \textbf{N2C}}
    \hspace{.1em}
    \makebox[0.18\textwidth]{\scriptsize \textbf{IR-CSF}}
    \hspace{.1em}
    \makebox[0.18\textwidth]{\scriptsize \textbf{N2N}}
    \hspace{.1em}
    \makebox[0.18\textwidth]{\scriptsize \textbf{IR-NSF}}
    \vspace{.5em}
    \\
    \raisebox{1.6\height}{\makebox[0.02\textwidth]{\rotatebox{90}{\makecell{\scriptsize \textbf{ResNet}}}}}
    \includegraphics[width=0.18\textwidth]{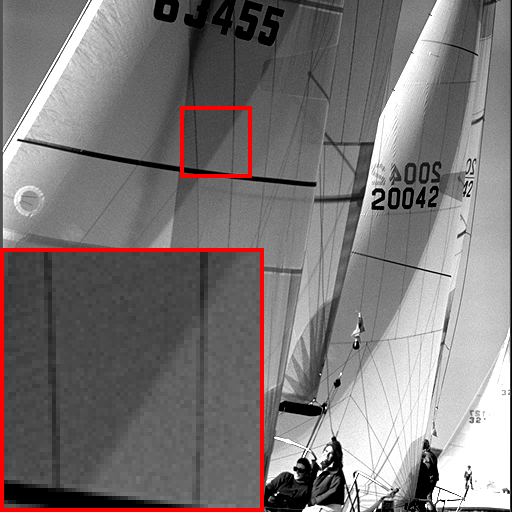}
    \hspace{.1em}
    \includegraphics[width=0.18\textwidth]{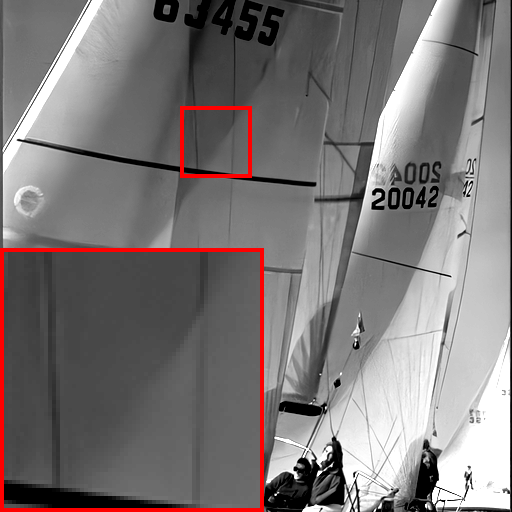}
    \hspace{.1em}
    \includegraphics[width=0.18\textwidth]{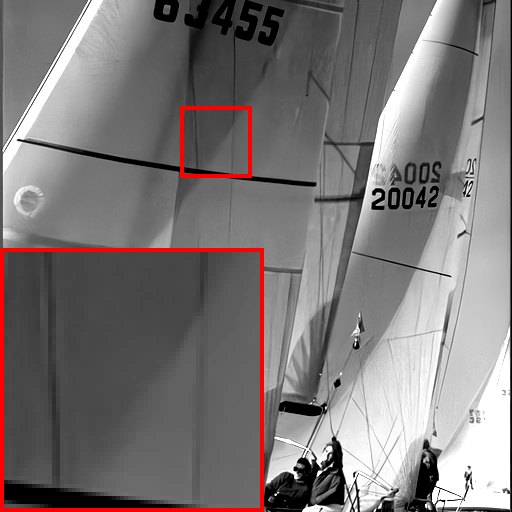}
    \hspace{.1em}
    \includegraphics[width=0.18\textwidth]{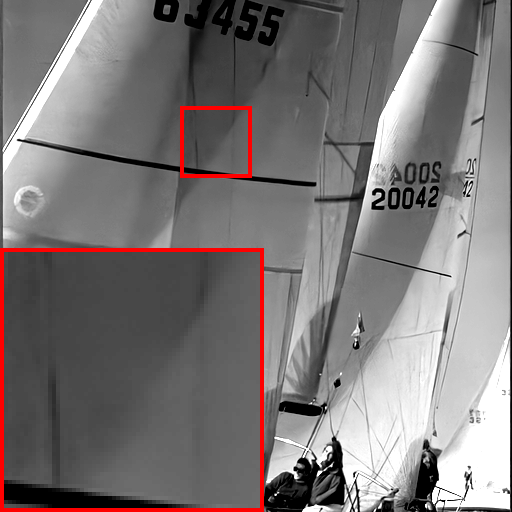}
    \hspace{.1em}
    \includegraphics[width=0.18\textwidth]{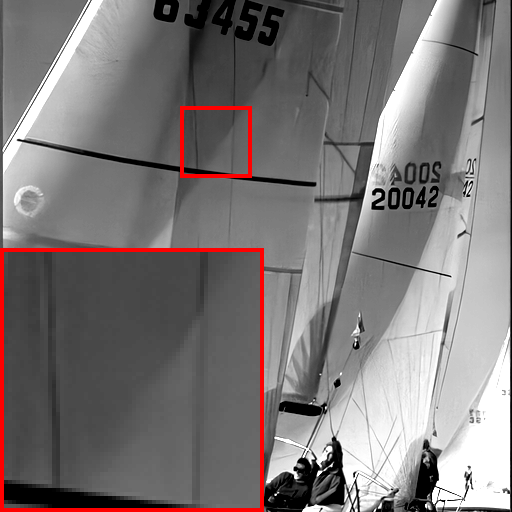}
    % \vspace{.5em}
    \\
    \makebox[0.02\textwidth]{}
    \makebox[0.18\textwidth]{\scriptsize kodim10 in Kodak24}
    \hspace{.1em}
    \makebox[0.18\textwidth]{\scriptsize PSNR: 36.98\SI{}{dB}, SSIM: 0.9215}
    \hspace{.1em}
    \makebox[0.18\textwidth]{\scriptsize PSNR: 37.00\SI{}{dB}, SSIM: 0.9219}
    \hspace{.1em}
    \makebox[0.18\textwidth]{\scriptsize PSNR: 36.69\SI{}{dB}, SSIM: 0.9164}
    \hspace{.1em}
    \makebox[0.18\textwidth]{\scriptsize PSNR: 36.97\SI{}{dB}, SSIM: 0.9215}
    \vspace{.5em}
    \\
    \raisebox{2.2\height}{\makebox[0.02\textwidth]{\rotatebox{90}{\makecell{\scriptsize \textbf{RDN}}}}}
    \includegraphics[width=0.18\textwidth]{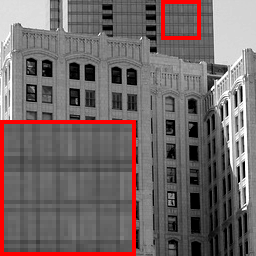}
    \hspace{.1em}
    \includegraphics[width=0.18\textwidth]{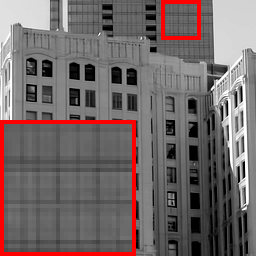}
    \hspace{.1em}
    \includegraphics[width=0.18\textwidth]{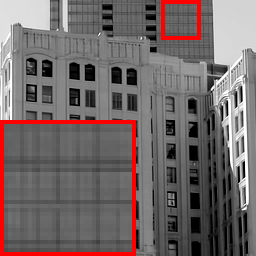}
    \hspace{.1em}
    \includegraphics[width=0.18\textwidth]{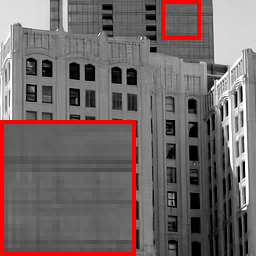}
    \hspace{.1em}
    \includegraphics[width=0.18\textwidth]{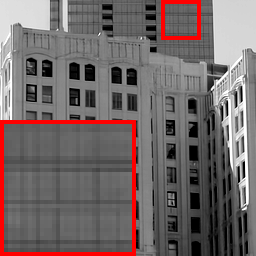}
    \\
    \makebox[0.02\textwidth]{}
    \makebox[0.18\textwidth]{\scriptsize img097 in Urban100}
    \hspace{.1em}
    \makebox[0.18\textwidth]{\scriptsize PSNR: 34.42\SI{}{dB}, SSIM: 0.9350}
    \hspace{.1em}
    \makebox[0.18\textwidth]{\scriptsize PSNR: 34.30\SI{}{dB}, SSIM: 0.9344}
    \hspace{.1em}
    \makebox[0.18\textwidth]{\scriptsize PSNR: 33.88\SI{}{dB}, SSIM: 0.9247}
    \hspace{.1em}
    \makebox[0.18\textwidth]{\scriptsize PSNR: 34.35\SI{}{dB}, SSIM: 0.9333}
    % \vspace{.5em}
    \caption{\revision{Visual results of denoising experiments. The image contrast is adjusted for better comparison.}}
    \label{fig: visual_dn}
\end{figure*}

\begin{table*}[!t]
	\centering
	\caption{Average PSNR (\SI{}{dB}) and SSIM results of denoising experiments. The results of our proposed IR-NSF are highlighted with gray color.}
	\label{tab: dn}
	\resizebox*{\textwidth}{!}{
		\begin{tabular}{*{15}{c}}
			\toprule
			\multirow{2}*{Network} & \multirow{2}*{Target} & \multirow{2}*{Method} & \multicolumn{2}{c}{Set5} & \multicolumn{2}{c}{Set14} & \multicolumn{2}{c}{Kodak24} & \multicolumn{2}{c}{McMaster} & \multicolumn{2}{c}{Urban100} & \multicolumn{2}{c}{Manga109}\\
            \cmidrule(lr){4-5}   \cmidrule(lr){6-7}   \cmidrule(lr){8-9} \cmidrule(lr){10-11} \cmidrule(lr){12-13} \cmidrule(lr){14-15}
            & & & PSNR & SSIM & PSNR & SSIM & PSNR & SSIM & PSNR & SSIM & PSNR & SSIM & PSNR & SSIM\\
			
			\midrule
			\multirow{4}*{ResNet} & \multirow{2}*{Clean} & N2C & 35.07 & 0.9315 & 33.77 & 0.9147 & 34.45 & 0.9176 & 35.92 & 0.9398 & 34.39 & 0.9477 & 36.16 & 0.9508\\
			& & IR-CSF & 35.07 & 0.9317 & 33.77 & 0.9149 & 34.45 & 0.9176 & 35.92 & 0.9399 & 34.41 & 0.9478 & 36.16 & 0.9508\\
			& \multirow{2}*{Noisy} & N2N & 34.94 & 0.9295 & 33.62 & 0.9125 & 34.32 & 0.9155 & 35.77 & 0.9380 & 34.04 & 0.9450 & 35.97 & 0.9492\\
			&  & \cellcolor{gray!20}IR-NSF & \cellcolor{gray!20}35.06 & \cellcolor{gray!20}0.9312 & \cellcolor{gray!20}33.75 & \cellcolor{gray!20}0.9146 & \cellcolor{gray!20}34.43 & 
                \cellcolor{gray!20}0.9175 & \cellcolor{gray!20}35.90 & \cellcolor{gray!20}0.9396 & \cellcolor{gray!20}34.34 & \cellcolor{gray!20}0.9474 & \cellcolor{gray!20}36.14 & \cellcolor{gray!20}0.9505\\
   
			\midrule
			\multirow{4}*{RDN} & \multirow{2}*{Clean} & N2C & 35.17 & 0.9327 & 33.92 & 0.9166 & 34.58 & 0.9196 & 36.07 & 0.9415 & 34.73 & 0.9505 & 36.37 & 0.9523\\
			& & IR-CSF & 35.18 & 0.9328 & 33.93 & 0.9169 & 34.59 & 0.9197 & 36.08 & 0.9416 & 34.74 & 0.9506 & 36.39 & 0.9524\\
			& \multirow{2}*{Noisy} & N2N & 35.03 & 0.9311 & 33.72 & 0.9142 & 34.42 & 0.9172 & 35.88 & 0.9394 & 34.29 & 0.9473 & 36.10 & 0.9502\\
			&  & \cellcolor{gray!20}IR-NSF & \cellcolor{gray!20}35.16 & \cellcolor{gray!20}0.9325 & \cellcolor{gray!20}33.89 & \cellcolor{gray!20}0.9162 & \cellcolor{gray!20}34.57 & 
                \cellcolor{gray!20}0.9193 & \cellcolor{gray!20}36.05 & \cellcolor{gray!20}0.9412 & \cellcolor{gray!20}34.66 & \cellcolor{gray!20}0.9500 & \cellcolor{gray!20}36.34 & \cellcolor{gray!20}0.9519\\
			
			\bottomrule
	  \end{tabular}
    }
\end{table*}

% ------------------------------------------------------------------- %
% ------------------------------------------------------------------- %
% fig: visual comparison of super-resolution (SR) task
% \input{figs/visual_cmp/fig_sr}
\begin{figure*}[t] \centering
    \makebox[0.02\textwidth]{}
    \makebox[0.18\textwidth]{\scriptsize \textbf{Ground Truth}}
    \hspace{.1em}
    \makebox[0.18\textwidth]{\scriptsize \textbf{N2C}}
    \hspace{.1em}
    \makebox[0.18\textwidth]{\scriptsize \textbf{IR-CSF}}
    \hspace{.1em}
    \makebox[0.18\textwidth]{\scriptsize \textbf{N2N}}
    \hspace{.1em}
    \makebox[0.18\textwidth]{\scriptsize \textbf{IR-NSF}}
    \vspace{.5em}
    \\
    \raisebox{1.8\height}{\makebox[0.02\textwidth]{\rotatebox{90}{\makecell{\scriptsize \textbf{EDSR}}}}}
    \includegraphics[width=0.18\textwidth]{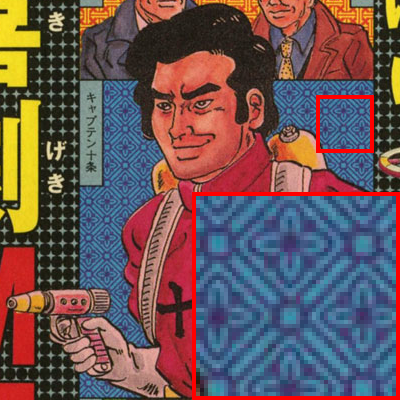}
    \hspace{.1em}
    \includegraphics[width=0.18\textwidth]{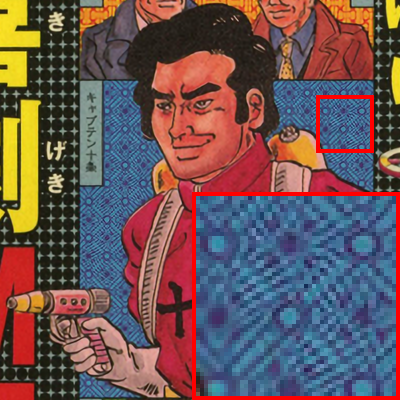}
    \hspace{.1em}
    \includegraphics[width=0.18\textwidth]{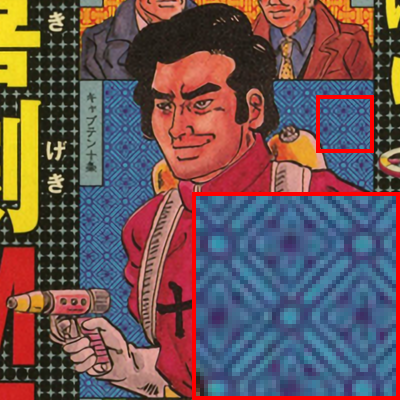}
    \hspace{.1em}
    \includegraphics[width=0.18\textwidth]{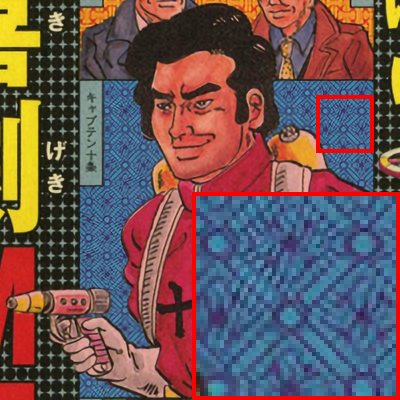}
    \hspace{.1em}
    \includegraphics[width=0.18\textwidth]{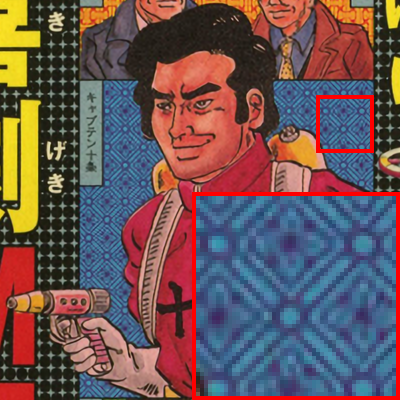}
    \\
    \makebox[0.02\textwidth]{}
    \makebox[0.18\textwidth]{\scriptsize UchuKigekiM774 in Manga109}
    \hspace{.1em}
    \makebox[0.18\textwidth]{\scriptsize PSNR: 30.72\SI{}{dB}, SSIM: 0.9219}
    \hspace{.1em}
    \makebox[0.18\textwidth]{\scriptsize PSNR: 31.91\SI{}{dB}, SSIM: 0.9389}
    \hspace{.1em}
    \makebox[0.18\textwidth]{\scriptsize PSNR: 30.47\SI{}{dB}, SSIM: 0.9193}
    \hspace{.1em}
    \makebox[0.18\textwidth]{\scriptsize PSNR: 31.76\SI{}{dB}, SSIM: 0.9370}
    \vspace{.5em}
    \\
    \raisebox{1.8\height}{\makebox[0.02\textwidth]{\rotatebox{90}{\makecell{\scriptsize \textbf{RCAN}}}}}
    \includegraphics[width=0.18\textwidth]{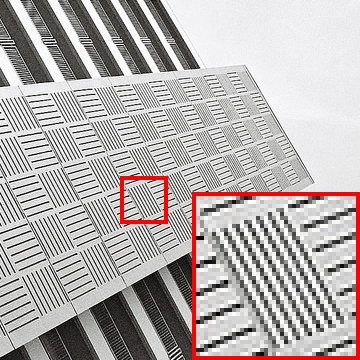}
    \hspace{.1em}
    \includegraphics[width=0.18\textwidth]{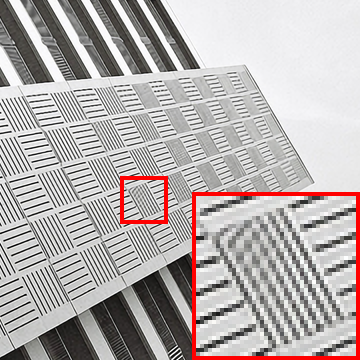}
    \hspace{.1em}
    \includegraphics[width=0.18\textwidth]{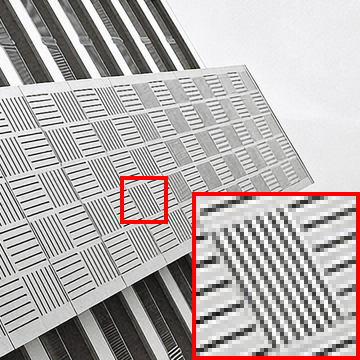}
    \hspace{.1em}
    \includegraphics[width=0.18\textwidth]{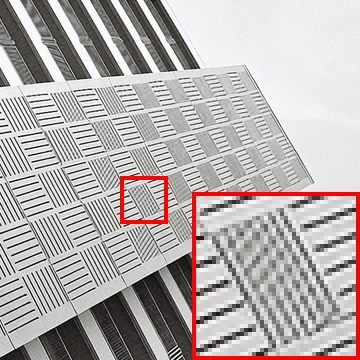}
    \hspace{.1em}
    \includegraphics[width=0.18\textwidth]{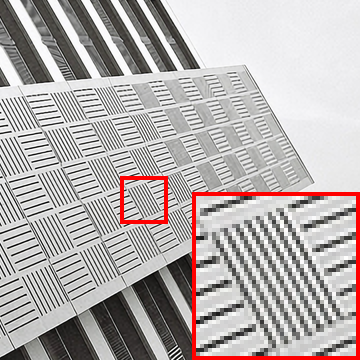}
    \\
    \makebox[0.02\textwidth]{}
    \makebox[0.18\textwidth]{\scriptsize img092 in Urban100}
    \hspace{.1em}
    \makebox[0.18\textwidth]{\scriptsize PSNR: 21.91\SI{}{dB}, SSIM: 0.8539}
    \hspace{.1em}
    \makebox[0.18\textwidth]{\scriptsize PSNR: 22.00\SI{}{dB}, SSIM: 0.8545}
    \hspace{.1em}
    \makebox[0.18\textwidth]{\scriptsize PSNR: 20.98\SI{}{dB}, SSIM: 0.8346}
    \hspace{.1em}
    \makebox[0.18\textwidth]{\scriptsize PSNR: 22.03\SI{}{dB}, SSIM: 0.8531}
    % \vspace{.5em}
    \\
    \caption{\revision{Visual results of super-resolution experiments.}}
    \label{fig: visual_sr}
\end{figure*}
\begin{table*}[!t]
	\centering
	\caption{Average PSNR (\SI{}{dB}) and SSIM results of super-resolution experiments. The results of our proposed IR-NSF are highlighted with gray color.}
	\label{tab: sr}
	\resizebox*{\textwidth}{!}{
		\begin{tabular}{*{15}{c}}
			\toprule
			\multirow{2}*{Network} & \multirow{2}*{Target} & \multirow{2}*{Method} & \multicolumn{2}{c}{Set5} & \multicolumn{2}{c}{Set14} & \multicolumn{2}{c}{Kodak24} & \multicolumn{2}{c}{McMaster} & \multicolumn{2}{c}{Urban100} & \multicolumn{2}{c}{Manga109}\\
            \cmidrule(lr){4-5}   \cmidrule(lr){6-7}   \cmidrule(lr){8-9} \cmidrule(lr){10-11} \cmidrule(lr){12-13} \cmidrule(lr){14-15}
            & & & PSNR & SSIM & PSNR & SSIM & PSNR & SSIM & PSNR & SSIM & PSNR & SSIM & PSNR & SSIM\\
			
			\midrule
			\multirow{4}*{EDSR} & \multirow{2}*{Clean} & N2C  & 37.95 & 0.9602 & 33.54 & 0.9175 & 34.01 & 0.9229 & 38.32 & 0.9636 & 32.05 & 0.9277 & 38.37 & 0.9763\\
			& & IR-CSF & 38.07 & 0.9607 & 33.68 & 0.9181 & 34.08 & 0.9237 & 38.45 & 0.9642 & 32.20 & 0.9286 & 38.75 & 0.9772\\
			& \multirow{2}*{Noisy} & N2N & 37.94 & 0.9603 & 33.54 & 0.9173 & 33.99 & 0.9227 & 38.29 & 0.9635 & 32.04 & 0.9277 & 38.27 & 0.9762\\
			&  & \cellcolor{gray!20}IR-NSF & \cellcolor{gray!20}38.04 & \cellcolor{gray!20}0.9606 & \cellcolor{gray!20}33.65 & \cellcolor{gray!20}0.9183 & \cellcolor{gray!20}34.06 & \cellcolor{gray!20}0.9234 & \cellcolor{gray!20}38.42 & \cellcolor{gray!20}0.9642 & \cellcolor{gray!20}32.20 & \cellcolor{gray!20}0.9288 & \cellcolor{gray!20}38.55 & \cellcolor{gray!20}0.9769\\
   
			\midrule
			\multirow{4}*{RCAN} & \multirow{2}*{Clean} & N2C & 38.17 & 0.9611 & 33.96 & 0.9202 & 34.29 & 0.9260 & 38.72 & 0.9655 & 32.90 & 0.9350 & 39.12 & 0.9779\\
			& & IR-CSF & 38.27 & 0.9615 & 34.02 & 0.9217 & 34.42 & 0.9273 & 38.91 & 0.9664 & 33.14 & 0.9370 & 39.44 & 0.9785\\
			& \multirow{2}*{Noisy} & N2N & 38.11 & 0.9597 & 33.95 & 0.9204 & 34.34 & 0.9264 & 38.73 & 0.9656 & 32.88 & 0.9351 & 39.11 & 0.9775\\
			&  & \cellcolor{gray!20}IR-NSF & \cellcolor{gray!20}38.25 & \cellcolor{gray!20}0.9614 & \cellcolor{gray!20}34.00  & \cellcolor{gray!20}0.9211 & \cellcolor{gray!20}34.46  & \cellcolor{gray!20}0.9274 & \cellcolor{gray!20}38.84  & \cellcolor{gray!20}0.9661 & \cellcolor{gray!20}33.10  & \cellcolor{gray!20}0.9366 & \cellcolor{gray!20}39.33  & \cellcolor{gray!20}0.9784\\
			
			\bottomrule
	  \end{tabular}
    }
\end{table*}

% ------------------------------------------------------------------- %
% ------------------------------------------------------------------- %
% fig: visual comparison of deblurring (DB) task
% \input{figs/visual_cmp/fig_db}
\begin{figure*}[t] \centering
    \makebox[0.02\textwidth]{}
    \makebox[0.18\textwidth]{\scriptsize \textbf{Ground Truth}}
    \hspace{.1em}
    \makebox[0.18\textwidth]{\scriptsize \textbf{N2C}}
    \hspace{.1em}
    \makebox[0.18\textwidth]{\scriptsize \textbf{IR-CSF}}
    \hspace{.1em}
    \makebox[0.18\textwidth]{\scriptsize \textbf{N2N}}
    \hspace{.1em}
    \makebox[0.18\textwidth]{\scriptsize \textbf{IR-NSF}}
    \vspace{.5em}
    \\
    \raisebox{1.7\height}{\makebox[0.02\textwidth]{\rotatebox{90}{\makecell{\scriptsize \textbf{MIMO}}}}}
    \includegraphics[width=0.18\textwidth]{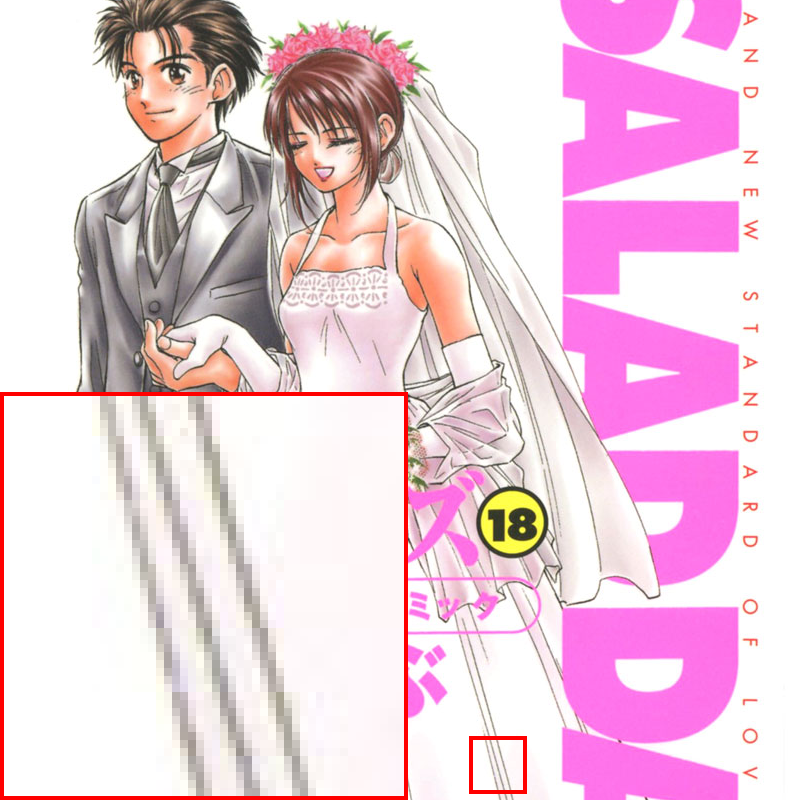}
    \hspace{.1em}
    \includegraphics[width=0.18\textwidth]{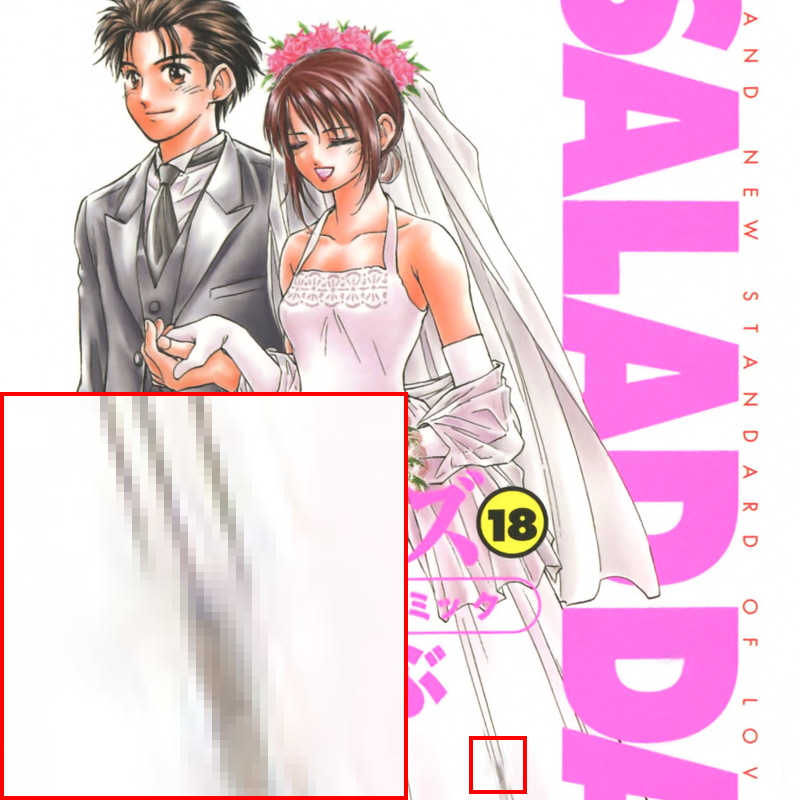}
    \hspace{.1em}
    \includegraphics[width=0.18\textwidth]{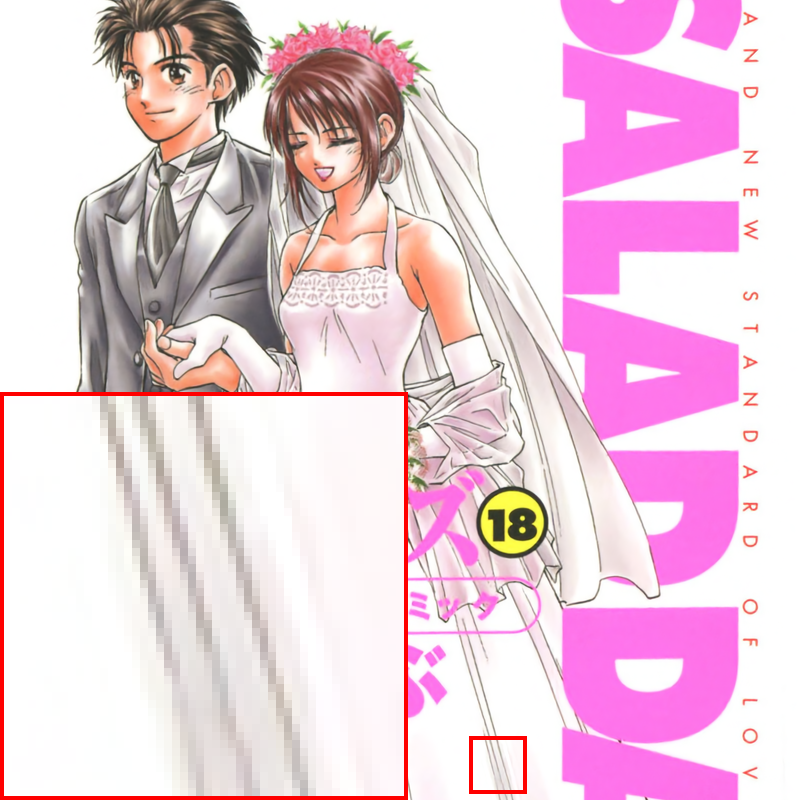}
    \hspace{.1em}
    \includegraphics[width=0.18\textwidth]{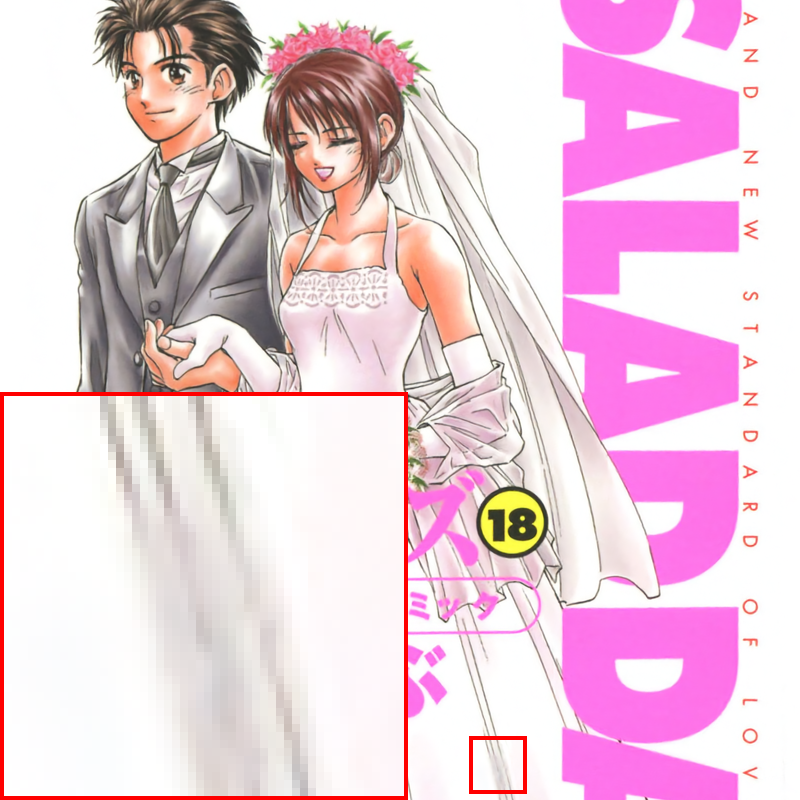}
    \hspace{.1em}
    \includegraphics[width=0.18\textwidth]{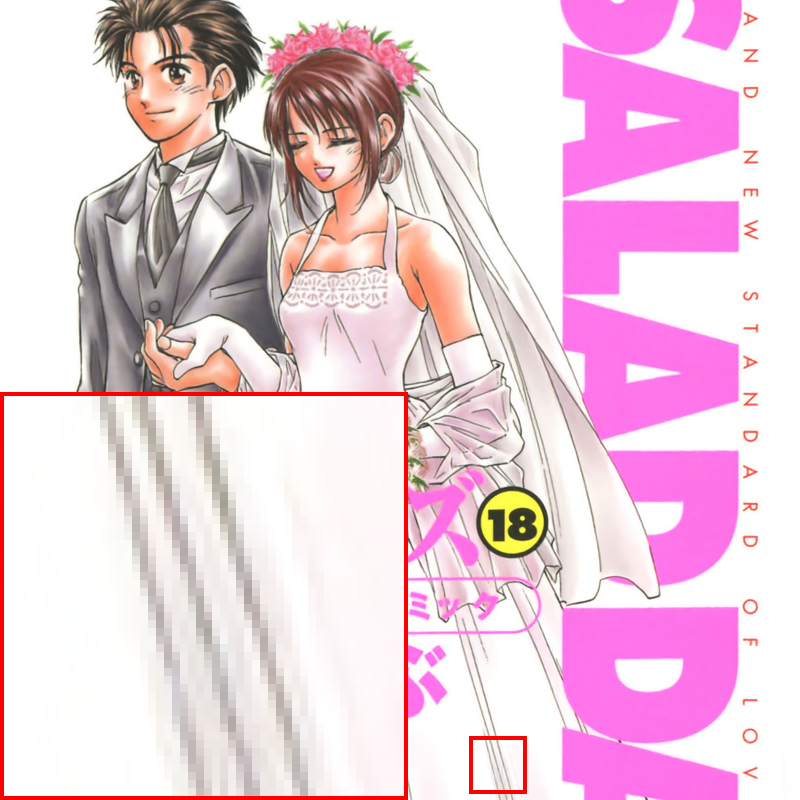}
    \\
    \makebox[0.02\textwidth]{}
    \makebox[0.18\textwidth]{\scriptsize SaladDays$\_$vol18 in Manga109}
    \hspace{.1em}
    \makebox[0.18\textwidth]{\scriptsize PSNR: 36.43\SI{}{dB}, SSIM: 0.9663}
    \hspace{.1em}
    \makebox[0.18\textwidth]{\scriptsize PSNR: 37.06\SI{}{dB}, SSIM: 0.9690}
    \hspace{.1em}
    \makebox[0.18\textwidth]{\scriptsize PSNR: 36.65\SI{}{dB}, SSIM: 0.9668}
    \hspace{.1em}
    \makebox[0.18\textwidth]{\scriptsize PSNR: 37.07\SI{}{dB}, SSIM: 0.9684}
    \vspace{.5em}
    \\
    \raisebox{1.6\height}{\makebox[0.02\textwidth]{\rotatebox{90}{\makecell{\scriptsize \textbf{MIMO}}}}}
    \includegraphics[width=0.18\textwidth]{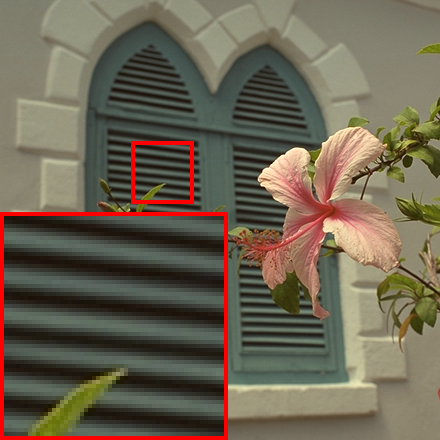}
    \hspace{.1em}
    \includegraphics[width=0.18\textwidth]{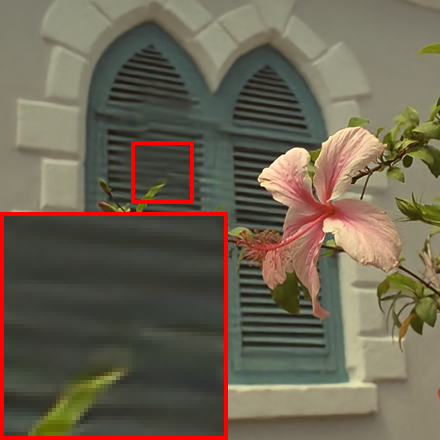}
    \hspace{.1em}
    \includegraphics[width=0.18\textwidth]{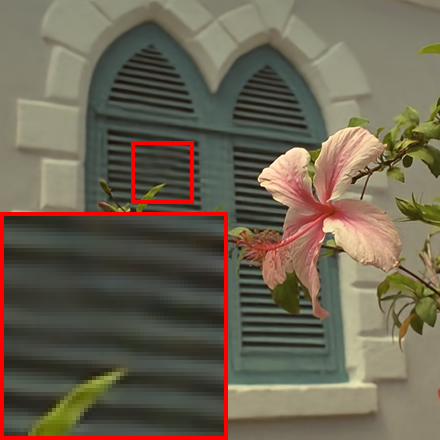}
    \hspace{.1em}
    \includegraphics[width=0.18\textwidth]{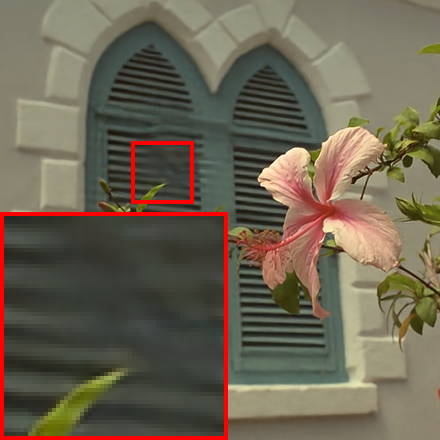}
    \hspace{.1em}
    \includegraphics[width=0.18\textwidth]{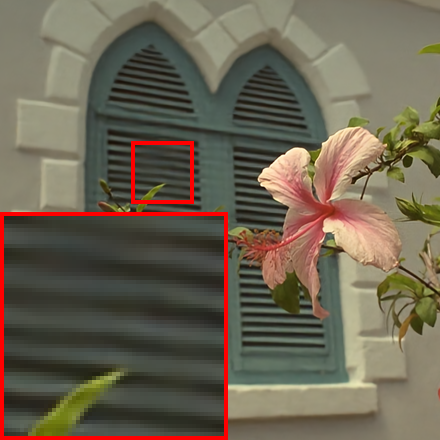}
    \\
    \makebox[0.02\textwidth]{}
    \makebox[0.18\textwidth]{\scriptsize kodim07 in Kodak24}
    \hspace{.1em}
    \makebox[0.18\textwidth]{\scriptsize PSNR: 31.43\SI{}{dB}, SSIM: 0.9194}
    \hspace{.1em}
    \makebox[0.18\textwidth]{\scriptsize PSNR: 34.85\SI{}{dB}, SSIM: 0.9531}
    \hspace{.1em}
    \makebox[0.18\textwidth]{\scriptsize PSNR: 32.71\SI{}{dB}, SSIM: 0.9333}
    \hspace{.1em}
    \makebox[0.18\textwidth]{\scriptsize PSNR: 34.72\SI{}{dB}, SSIM: 0.9530}
    % \vspace{.5em}
    \caption{\revision{Visual results of deblurring experiments.} }
    \label{fig: visual_db}
\end{figure*}
\begin{table*}[!t]
	\centering
	\caption{Average PSNR (\SI{}{dB}) and SSIM results of deblurring experiments. The results of our proposed IR-NSF are highlighted with gray color.}
	\label{tab: db}
	\resizebox*{\textwidth}{!}{
		\begin{tabular}{*{15}{c}}
			\toprule
			\multirow{2}*{Network} & \multirow{2}*{Target} & \multirow{2}*{Method} & \multicolumn{2}{c}{Set5} & \multicolumn{2}{c}{Set14} & \multicolumn{2}{c}{Kodak24} & \multicolumn{2}{c}{McMaster} & \multicolumn{2}{c}{Urban100} & \multicolumn{2}{c}{Manga109}\\
            \cmidrule(lr){4-5}   \cmidrule(lr){6-7}   \cmidrule(lr){8-9} \cmidrule(lr){10-11} \cmidrule(lr){12-13} \cmidrule(lr){14-15}
            & & & PSNR & SSIM & PSNR & SSIM & PSNR & SSIM & PSNR & SSIM & PSNR & SSIM & PSNR & SSIM\\
			
			\midrule
			\multirow{4}*{MIMO} & \multirow{2}*{Clean} & N2C & 32.75 & 0.9162 & 31.56 & 0.8954 & 33.09 & 0.9186 & 34.56 & 0.9325 & 30.75 & 0.9218 & 34.61 & 0.9431\\
			& & IR-CSF & 32.95 & 0.9179 & 32.05 & 0.9004 & 33.64 & 0.9224 & 34.74 & 0.9344 & 31.30 & 0.9274 & 34.98 & 0.9455\\
			& \multirow{2}*{Noisy} & N2N & 32.79 & 0.9157 & 31.75 & 0.8967 & 33.12 & 0.9189 & 34.49 & 0.9322 & 30.84 & 0.9221 & 34.65 & 0.9430\\
			&  & \cellcolor{gray!20}IR-NSF & \cellcolor{gray!20}32.96 & \cellcolor{gray!20}0.9179 & \cellcolor{gray!20}32.05  & \cellcolor{gray!20}0.9000 & \cellcolor{gray!20}33.59  & \cellcolor{gray!20}0.9223 & \cellcolor{gray!20}34.73  & \cellcolor{gray!20}0.9344 & \cellcolor{gray!20}31.31  & \cellcolor{gray!20}0.9274 & \cellcolor{gray!20}34.88  & \cellcolor{gray!20}0.9454\\
   
			\midrule
			\multirow{4}*{MIMO+} & \multirow{2}*{Clean} & N2C & 33.13 & 0.9186 & 32.10 & 0.9000 & 33.52 & 0.9214 & 34.71 & 0.9339 & 31.45 & 0.9277 & 35.12 & 0.9460\\
			& & IR-CSF & 33.47 & 0.9222 & 32.35 & 0.9025 & 34.00 & 0.9243 & 35.03 & 0.9365 & 32.02 & 0.9342 & 35.47 & 0.9486\\
			& \multirow{2}*{Noisy} & N2N & 33.11 & 0.9183 & 31.96 & 0.8985 & 33.70 & 0.9221 & 34.69 & 0.9338 & 31.46 & 0.9285 & 35.05 & 0.9457\\
			&  & \cellcolor{gray!20}IR-NSF & \cellcolor{gray!20}33.43 & \cellcolor{gray!20}0.9215 & \cellcolor{gray!20}32.28  & \cellcolor{gray!20}0.9020 & \cellcolor{gray!20}34.03  & \cellcolor{gray!20}0.9240 & \cellcolor{gray!20}34.95  & \cellcolor{gray!20}0.9359 & \cellcolor{gray!20}32.08  & \cellcolor{gray!20}0.9338 & \cellcolor{gray!20}35.37  & \cellcolor{gray!20}0.9480\\
			
			\bottomrule
	  \end{tabular}
    }
\end{table*}

% ------------------------------------------------------------------- %

\subsection{Convergence}
In Fig.~\ref{fig: convergence}, the convergence curves of using clean targets (IR-CSF) and noisy targets (IR-NSF) are compared. The convergence curve is plotted as PSNR on Urban100 against the training step. Clearly, for the SR experiment shown in Fig.~\ref{fig: convergence}(b), the convergence curves for IR-NSF and IR-CSF are almost consistent everywhere. As for the DN experiment shown in Fig.~\ref{fig: convergence}(a), there are obvious differences between the two curves at the early training stage. This is because noisy targets for DN consist of periodic noise, which is more difficult to tackle than the \iid~Gaussian noise in the target for the SR experiment. However, the extra difficulty posed by the periodic noise can be successfully overcome in the late training stage. As Fig.~\ref{fig: convergence}(a) shows, the curves of IR-CSF and IR-NSF finally converge to close points. Overall, experiments shown in Fig.~\ref{fig: convergence} demonstrate that training with noisy targets (IR-NSF) has an equivalent effect to training with clean targets (IR-CSF).

\subsection{Image Restoration}
\subsubsection{Denoising}
For the denoising experiments, we provide visual results in Fig.~\ref{fig: visual_dn} and report the average PSNR and SSIM results in Table~\ref{tab: dn}. From Fig.~\ref{fig: visual_dn} and Table~\ref{tab: dn}, it can be observed that IR-CSF and IR-NSF achieve similar performances. In contrast, there is a noticeable performance gap between N2C and N2N. On Urban100, the PSNR of N2N is over 0.3\SI{}{dB} worse than those of the other three methods. Also, N2N results in the loss of more image details, as illustrated in the zoomed regions in Fig.~\ref{fig: visual_dn}.  These results demonstrate that, compared to training in the spatial domain, training in the Fourier domain can reduce the gap between using clean and noisy targets. This improvement can be attributed to the fact that the periodic noise affects all spatial pixels but only influences a small portion of the Fourier coefficients.
% more severe blurring effects (1st to 3th rows of Fig.~\ref{fig: visual_dn}) and more artifacts (4th row of Fig.~\ref{fig: visual_dn}).

\subsubsection{Super-Resolution}
For the super-resolution experiments, we provide visual results in Fig.~\ref{fig: visual_sr} and report average PSNR and SSIM results in Table~\ref{tab: sr}. From Fig.~\ref{fig: visual_sr} and Table~\ref{tab: sr}, it can be observed that our proposed IR-NSF achieves a higher PSNR (up to 0.2\SI{}{dB} higher) and better visual results than N2N, despite there being almost no performance gap between training with noisy targets and clean targets in both the spatial domain and the Fourier domain for this experimental setting. This result validates that establishing supervision in the Fourier domain is more effective than in the spatial domain. The reason behind this could be attributed to Fourier coefficients containing global information, while spatial pixels only contain limited local information, allowing the former to provide more significant supervision.

\subsubsection{Deblurring}
For the deblurring experiments, we provide visual results in Fig.~\ref{fig: visual_db} and report average PSNR and SSIM results in Table~\ref{tab: db}. From Table~\ref{tab: db}, it can be observed that IR-NSF outperforms N2N over 0.3\SI{}{dB}. Moreover, as Fig.~\ref{fig: visual_db} shows, edges in images recovered by IR-NSF are clearer and sharper compared to those obtained with N2N. Similar to the SR experiments, there is almost no performance gap between IR-NSF/N2N and IR-CSF/N2C, demonstrating that the advantage of IR-NSF over N2N, in this case, is due to the Fourier-based supervision being more effective for the image deblurring task compared to spatial domain supervision.

% ------------------------------------------------------------------- %
\begin{table}[!t]
	\centering
	\caption{\revision{Average PSNR ({\SI{}{dB}}) and SSIM results on real-world denoising dataset SIDD~\cite{SIDD}. The results of our proposed IR-NSF are highlighted with gray color. The data collection burden is measured by the number of noisy images collected from each scene. In SIDD, it generates each ground truth clean image with a sequence of 150 captured images.}}
	\label{tab: SIDD}
	\resizebox*{0.49\textwidth}{!}{
		\begin{tabular}{*{6}{c}}
			\toprule
			Metrics & N2C & IR-CSF & Nr2Nr & N2N & \cellcolor{gray!20}IR-NSF\\
			\midrule
            Data collection burden & 150  & 150 & 1  & 2  & \cellcolor{gray!20}2  \\
            PSNR & 51.25  & 51.26 & 50.68  & 51.14  & \cellcolor{gray!20}51.14  \\
            SSIM & 0.9916 & 0.9917 & 0.9900 & 0.9914 & \cellcolor{gray!20}0.9914 \\
			\bottomrule
	  \end{tabular}
    }
\end{table}
\begin{table}[!t]
	\centering
	\caption{\revision{Quantitative results on real-world deblurring dataset ReLoBlur~\cite{ReLoBlur}. PSNR$_w$ and SSIM$_w$ denote weighted PSNR and weighted SSIM that are only calculated on locally blurry region. The results of our proposed IR-NSF are highlighted with gray color.}}
	\label{tab: ReLoBlur}
	% \resizebox*{0.49\textwidth}{!}{
		\begin{tabular}{*{5}{c}}
			\toprule
		    Method &    PSNR    & SSIM & PSNR$_w$ & SSIM$_w$ \\
            \midrule
            N2N    &   34.19   &  0.9186 & 27.18 & 0.8505\\
            \cellcolor{gray!20}IR-NSF    & \cellcolor{gray!20}34.43  & \cellcolor{gray!20}0.9250 & \cellcolor{gray!20}27.39 & \cellcolor{gray!20}0.8630\\
            % PSNR/SSIM    &   34.19/0.9186   & \cellcolor{gray!20}34.43/0.9250 \\
   
			\bottomrule
	  \end{tabular}
    % }
\end{table}

% \begin{table}[!t]
% 	\centering
% 	\caption{\revision{Average PSNR ({\SI{}{dB}}) and SSIM results on real-world deblurring dataset ReLoBlur~\cite{ReLoBlur}. The results of our proposed IR-NSF are highlighted with gray color.}}
% 	\label{tab: ReLoBlur}
% 	% \resizebox*{0.49\textwidth}{!}{
% 		\begin{tabular}{*{3}{c}}
% 			\toprule
% 		    Metrics &   N2N     & \cellcolor{gray!20}IR-NSF\\
%             \midrule
%             PSNR    &   34.19   & \cellcolor{gray!20}34.43  \\
%             SSIM    &   0.9186  & \cellcolor{gray!20}0.9250 \\
%             % PSNR/SSIM    &   34.19/0.9186   & \cellcolor{gray!20}34.43/0.9250 \\
   
% 			\bottomrule
% 	  \end{tabular}
%     % }
% \end{table}

\begin{figure}[t] \centering
    \subfigure[Blurry]{
        \includegraphics[width=0.23\textwidth]{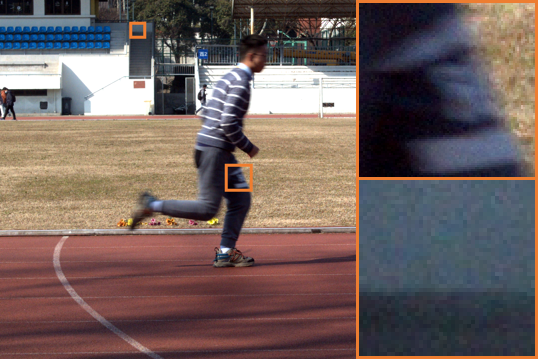}
    } 
    \hspace{-0.8em}
    \subfigure[Sharp]{
        \includegraphics[width=0.23\textwidth]{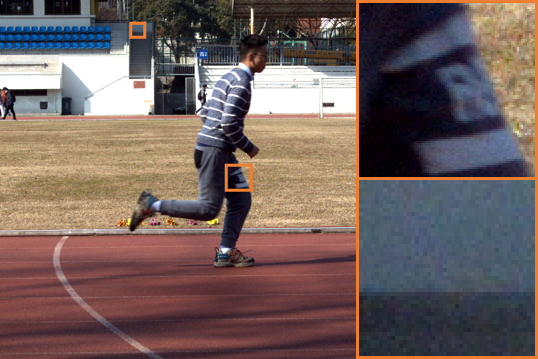}
    }
    \\
    \vspace{-0.5em}
    \subfigure[N2N]{
        \includegraphics[width=0.23\textwidth]{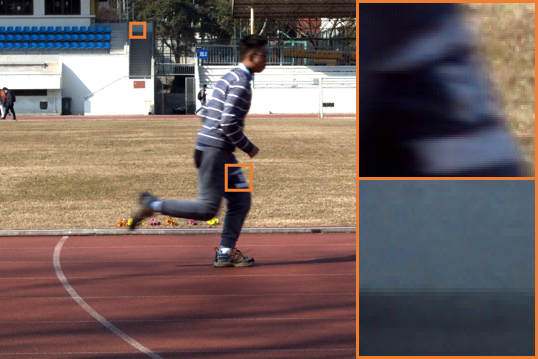}
    }
    \hspace{-0.8em}
    \subfigure[IR-NSF]{
        \includegraphics[width=0.23\textwidth]{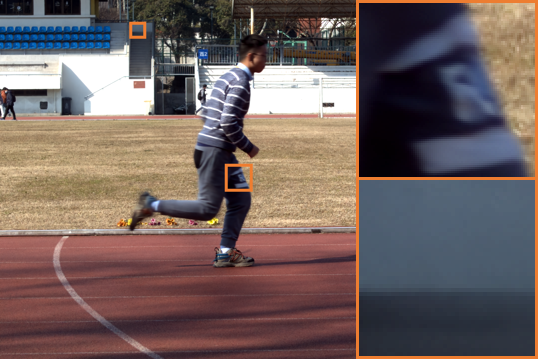}
    }
    \vspace{-0.5em}
    \caption{\revision{Visual results of deblurring experiments on ReLoBlur. The image contrast is adjusted for better comparison. In each image, the zoomed region on the top right corner is from the moving object, the one on the bottom right is from the static background. Noise can be observed in the zoomed background regions of both the (a) blurry image and the (b) sharp image. (c) N2N~\cite{N2N} (PSNR={33.79\SI{}{dB}}, SSIM=0.8904). (d) IR-NSF (PSNR={34.10\SI{}{dB}}, SSIM=0.9052).}} 
    \label{fig: visual_ReLoBlur}
    \vspace{-1em}
\end{figure}
% ------------------------------------------------------------------- %
\subsection{{Real-world experiments}}
\revision{
Besides synthetic experiments, the proposed method is validated through real-world experiments on two publicly available datasets. One is the smartphone image denoising dataset (SIDD)~\cite{SIDD}, and the other is the real local motion blur dataset (ReLoBlur)~\cite{ReLoBlur}. 
}

\subsubsection{{Denoising experiment on SIDD}}
 % and can be approximately treated as the heteroscedastic Gaussian noise~\cite{SIDD}
\revision{
The SIDD~\cite{SIDD} is captured with five smartphone cameras from 10 static scenes under different illumination conditions. For each scene instance, it provides two noisy images and their corresponding clean counterparts in both the RAW format and the sRGB format. The noise within each noisy image is spatially independent. The clean image is estimated from a sequence of 150 captured images. Following~\cite{Nr2Nr}, the SIDD Medium and Benchmark Datasets in RAW format are adopted for training and testing respectively. The noisy raw images are packed into 4-channel tensors for the network to process and the denoised raw image is obtained by unpacking the network output.  In addition to the comparison methods utilized in synthetic experiments, the SOTA unsupervised denoising method, Neighbor2Neighbor (Nr2Nr)~\cite{Nr2Nr}, is also incorporated in this real-world denoising experiment. For all methods, the network architecture is the ResNet, of which the training setting is the same as the synthetic experiment. Table~\ref{tab: SIDD} reports the PSNR and SSIM results of the comparison methods, which are measured by submitting the denoised raw images to the online server\footnote{www.kaggle.com/competitions/sidd-benchmark-raw-psnr; www.kaggle.com/competitions/sidd-benchmark-raw-ssim}. As Table~\ref{tab: SIDD} shows, our proposed IR-NSF performs very closely to that of IR-CSF, while reducing the data collection burden from 150 to 2. Besides, it achieves a significantly higher PSNR (\ie, {0.46\SI{}{dB}}) over Nr2Nr that requires single noisy images. This demonstrates that the developed noisy-supervision learning framework is effective in real-world denoising tasks. It can also be observed that the IR-NSF and N2N achieve the same metrics. This implies that the impact of the spatially independent noise on the training process in the Fourier domain is similar to that in the spatial domain.
}

\subsubsection{Deblurring experiment on ReLoBlur}
\revision{
The ReLoBlur~\cite{ReLoBlur} is captured by utilizing a beam splitter to divert incident light from the target scene to two synchronized cameras that are with different exposure times. For each scene, it consists of the static background region and the moving object. The images captured by paired cameras are blurry and sharp, respectively. To enhance the content consistency of these paired images, they are post-processed by color correction, photometrical alignment, and geometrical alignment. It is worth noting that the noise reduction operator is not involved in these post-processing steps, and thus the noise can be observed in both the blurry and sharp images. In other words, training targets provided by ReLoBlur are generally noisy, especially for scenes under relatively low lighting conditions. Utilizing this dataset, we conduct a comparison between N2N and the proposed IR-NSF. In this case, we adopt the same network architecture and training settings as~\cite{ReLoBlur}. For the metrics, in addition to PSNR and SSIM, we also follow~\cite{ReLoBlur} and adopt the weighted PSNR (PSNR$_w$) and weighted SSIM (SSIM$_w$) that are only calculated on locally blurry regions. Although the sharp image is somewhat noisy, these metrics can still evaluate the deblurring performance since the blurring effect generally has a stronger influence on these metrics than the noise. From Table~\ref{tab: ReLoBlur}, it can be observed that the proposed IR-NSF outperforms N2N in both PSNR and SSIM on the ReLoBlur dataset. In Fig.~\ref{fig: visual_ReLoBlur}, a visual comparison is provided. As the zoomed regions show, IR-NSF can obtain clearer and sharper results than N2N, demonstrating the superiority of IR-NSF on the real-world deblurring task.
}

\section{Potential Impacts}
\subsection{Impact on Data Collection} \label{sec: data_collection}
\revision{
Using the proposed IR-NSF, we conduct an investigation into a beam-splitting data collection system, as employed by RealBlur~\cite{RealBlur}, for acquiring training pairs from dynamic scenes. The experimental results reveal that directly utilizing noisy targets is better than generating pseudo-clean targets through denoising methods. Besides, we point out that, when the target is corrupted with spatially correlated noise such as the periodic noise, the proposed IR-NSF requires less training data than the N2N.
}

% ------------------------------------------------------------------- %
\begin{figure}[!t]
\centering

\includegraphics[width=0.98\linewidth]{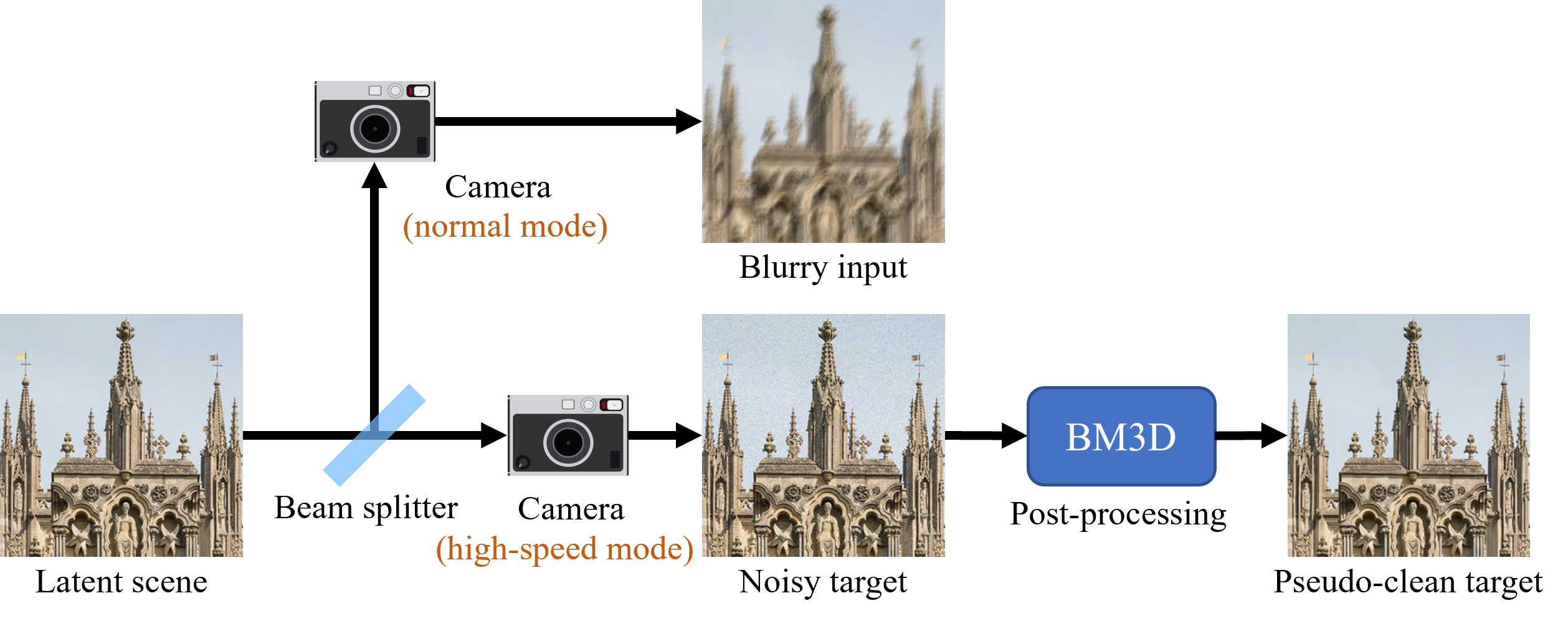}

\caption{Illustration of the data collection pipeline adopted by existing datasets such as the RealBlur~\cite{RealBlur}. The incident light from the latent scene is split into two parts by a beam splitter. One is captured by a camera in the normal mode, leading to a blurry image with a relatively higher SNR. The other is captured by a camera in the high-speed mode, leading to a noisy image that is with sharp edges. This noisy image is further post-processed by denosing methods such as BM3D~\cite{BM3D} to generate pseudo-clean targets.}
\label{fig: pipeline}
\end{figure}

\begin{figure}[t] \centering
    \makebox[0.02\textwidth]{}
    \makebox[0.15\textwidth]{\scriptsize \textbf{IR-CSF w/ Clean}}
    \makebox[0.15\textwidth]{\scriptsize \textbf{IR-CSF w/ Pseudo-clean}}
    \makebox[0.15\textwidth]{\scriptsize \textbf{IR-NSF w/ Noisy}}
    \\
    \vspace{0.1em}
    \raisebox{1.5\height}{\makebox[0.02\textwidth]{\rotatebox{90}{\makecell{\scriptsize \textbf{Target}}}}}
    \includegraphics[width=0.15\textwidth]{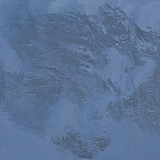}
    \includegraphics[width=0.15\textwidth]{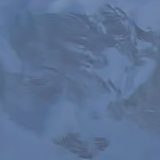}
    \includegraphics[width=0.15\textwidth]{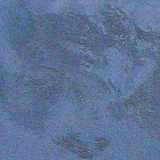}
    \\
    \vspace{0.2em}
    \raisebox{1.3\height}{\makebox[0.02\textwidth]{\rotatebox{90}{\makecell{\scriptsize \textbf{Result}}}}}
    \includegraphics[width=0.15\textwidth]{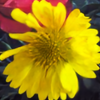}
    \includegraphics[width=0.15\textwidth]{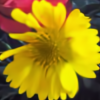}
    \includegraphics[width=0.15\textwidth]{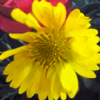}
    \\
    \caption{Visual comparison of targets and results under different training settings. The first row shows training targets and the second row shows their corresponding testing result. Since the pseudo-clean targets are blurred, their corresponding testing results are also blurry. The example images from top to bottom are cropped from: 0002 in DIV2K and flowers in Set14.}
    \label{fig: visual_dc}
\end{figure}
% ------------------------------------------------------------------- %
% ------------------------------------------------------------------- %
\begin{figure}[!t]
\centering

\includegraphics[width=0.98\linewidth]{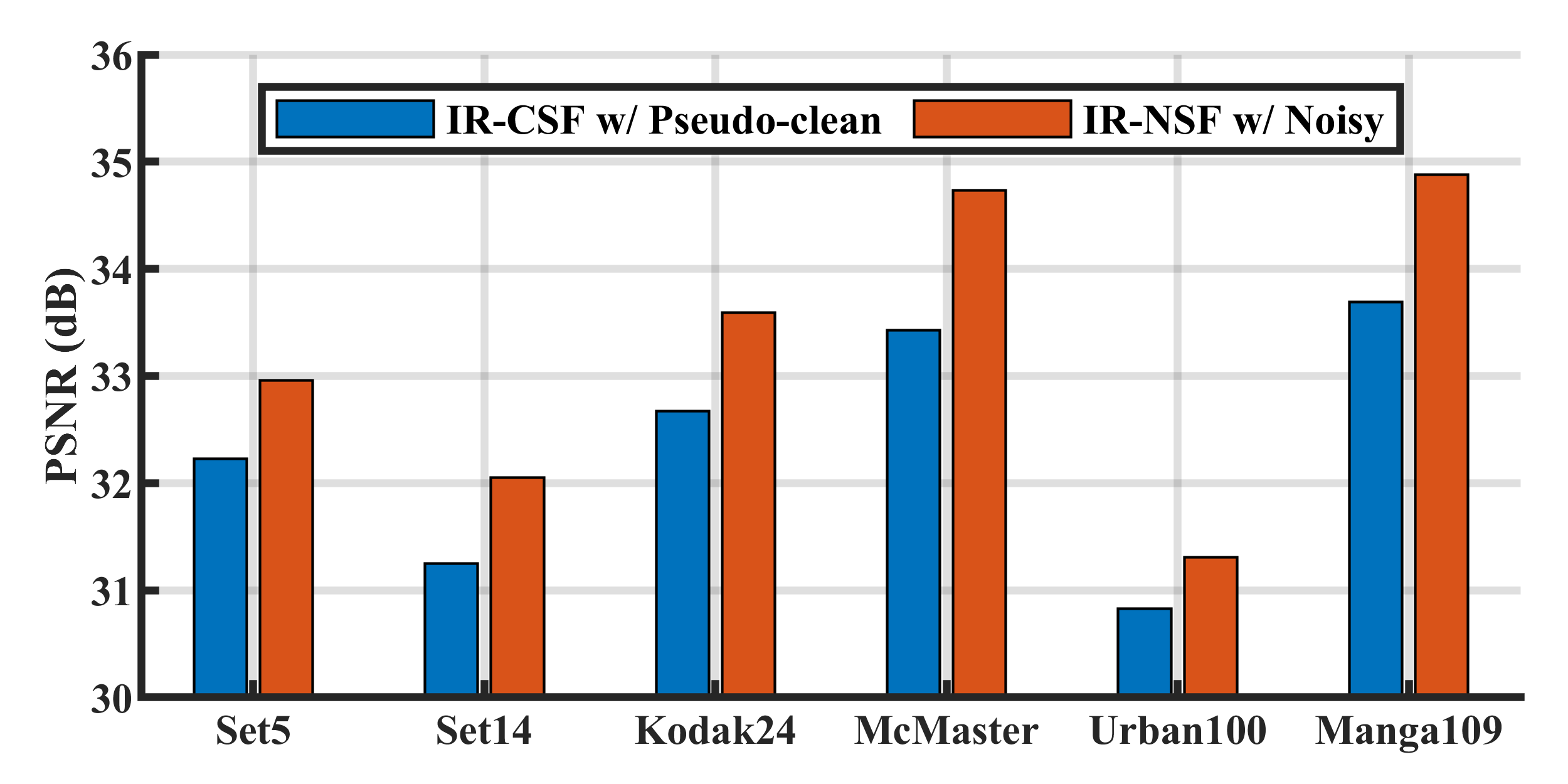}

\caption{Average PSNR (\SI{}{dB}) comparison between training with noisy targets and pseudo-clean targets.}
\label{fig: psnr_dc}
\end{figure}
% ------------------------------------------------------------------- %
% ------------------------------------------------------------------- %
% ------------------------------------------------------------------------ %
% for package subfigure
% ------------------------------------------------------------------------ %
\begin{figure}[!t]
\centering

\subfigure[Poisson-Gaussian noise]{\includegraphics[width=0.98\linewidth]{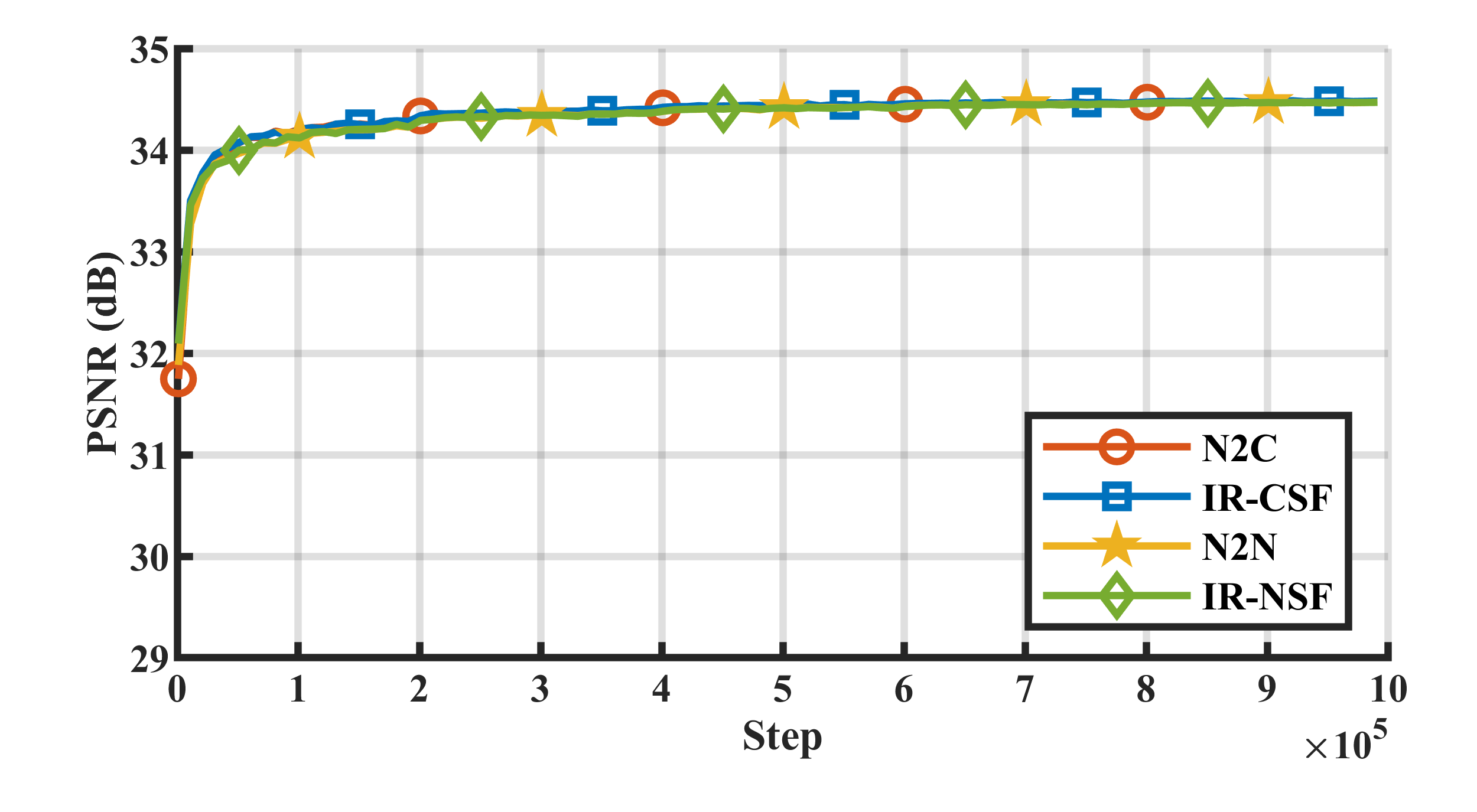}%
% \label{1}
}

\subfigure[Poisson-Gaussian noise and periodic noise]{\includegraphics[width=0.98\linewidth]{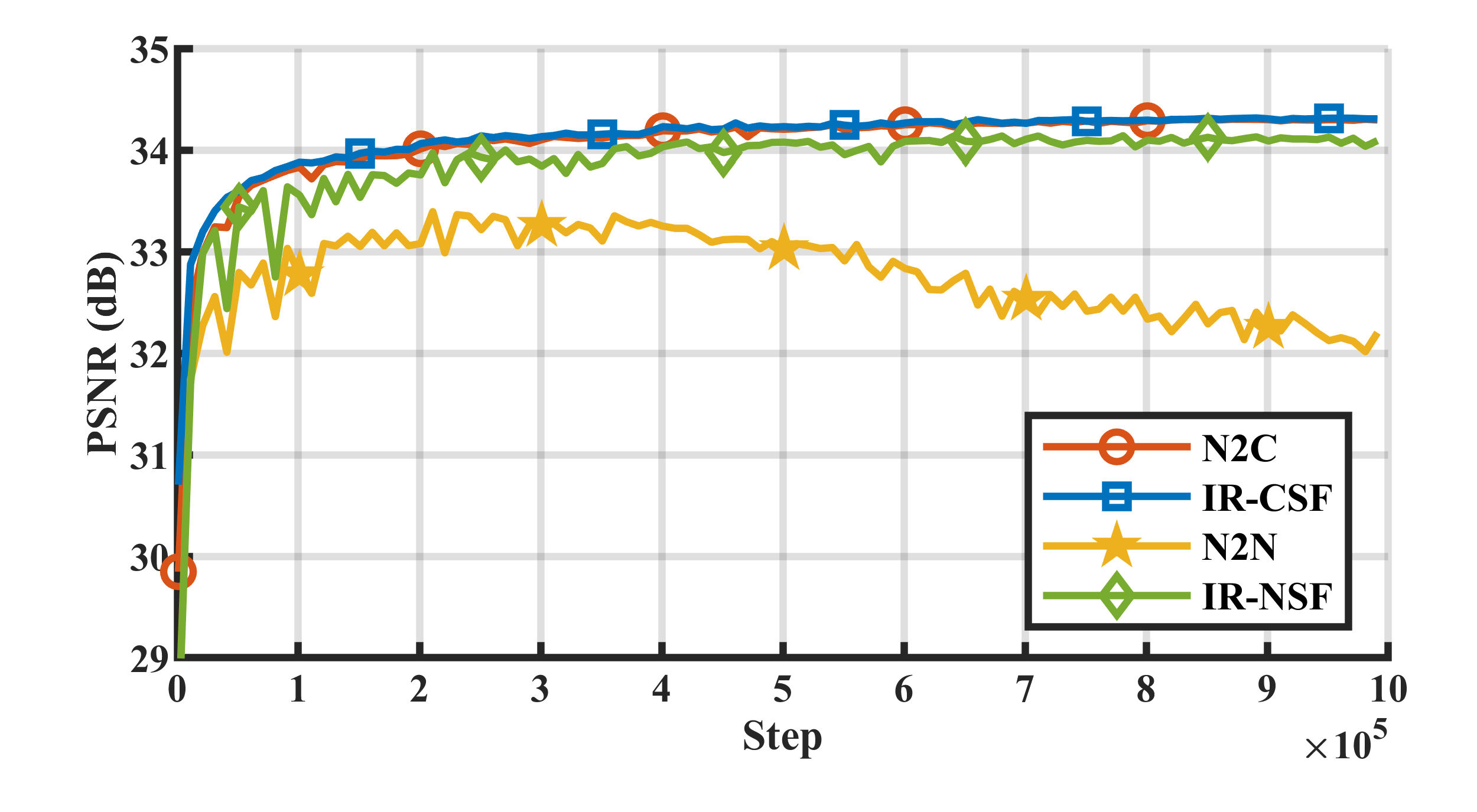}%
% \label{2}
}

\caption{\revision{PSNR ({\SI{}{dB}}) on Urban100 as a function of training step. For (a), the noise type is the Poisson-Gaussian noise and the training set is the DIV2K. For (b), the noise type is the mixture of Poisson-Gaussian and periodic noise. The training set is the combination of DIV2K and Flickr2K.}}
\label{fig: convergence_data_budget}
\end{figure}

\begin{figure}[t] \centering
    \subfigure[Ground Truth]{
        \includegraphics[width=0.15\textwidth]{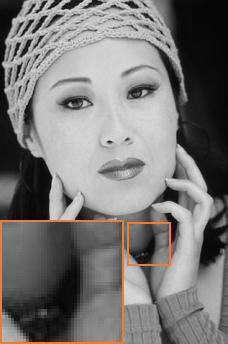}
    } 
    \hspace{-1em}
    \subfigure[N2N]{
        \includegraphics[width=0.15\textwidth]{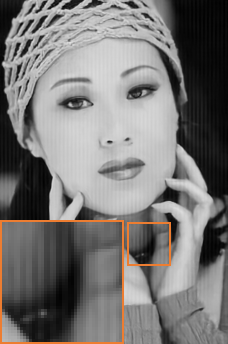}
    }
    \hspace{-1em}
    \subfigure[IR-NSF]{
        \includegraphics[width=0.15\textwidth]{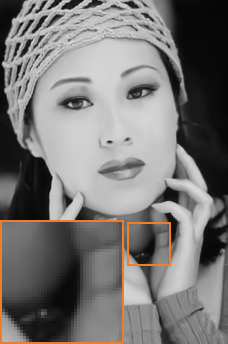}
    }
    \vspace{-0.5em}
    \caption{\revision{Visual results of denoising experiments under limited training data. The statistical equivalence for N2N becomes invalid in this case. Consequently, the periodic noise in the training target appears in its recovered image. (a) Ground Truth. (b) N2N (PSNR={33.27\SI{}{dB}}, SSIM=0.8984). (c) IR-NSF (PSNR={34.68\SI{}{dB}}, SSIM=0.9507).}} 
    \label{fig: visual_data_budget}
\end{figure}
% ------------------------------------------------------------------- %
\subsubsection{{Comparison of utilizing noisy target and generating pseudo-clean targe}}
In imaging systems, the SNR of its detected image will decrease with the increase of spatial and/or temporal resolution. Consequently, when collecting data, one is faced with a choice between obtaining high-resolution but noisy images or clean but low-resolution images. It is usually impractical to directly capture clean and high-resolution images as training targets. To address this challenge, some approaches (\eg, RealBlur~\cite{RealBlur}) adopt the data collection pipeline illustrated in Fig.~\ref{fig: pipeline} to collect training image pairs. This pipeline utilizes a beam splitter and two paired cameras to capture the blurry image and the noisy image from the same scene. The noisy image is further post-processed by denoising methods such as the BM3D~\cite{BM3D} to generate the pseudo-clean target. However, this post-processing step can introduce undesired artifacts and result in information loss. For instance, as illustrated in Fig.~\ref{fig: visual_dc}, the pseudo-clean target loses texture details and exhibits slight blurriness. Therefore, we claim that generating pseudo-clean targets is inferior to directly utilizing noisy targets. 

To substantiate this claim, we compare the performance of training with noisy targets and pseudo-clean targets for the image deblurring task. In the training mode with noisy targets, we employ the same settings as those used in the synthetic deblurring experiment with our proposed IR-NSF and the network MIMO. Noisy targets are generated by adding \iid~Gaussian noise of $\sigma = 10$ to clean images. In the alternative mode, noisy targets are processed using BM3D to generate training targets. Following RealBlur~\cite{RealBlur}, the noise level parameter of BM3D is set to $1.5 \sigma$. The comparison results are presented in Fig.~\ref{fig: visual_dc} and Fig.~\ref{fig: psnr_dc}. These figures clearly demonstrate that the result obtained with noisy targets appears sharper and contains more details than the one obtained with pseudo-clean targets, achieving a significantly higher (\ie, over 1\SI{}{dB}) average PSNR. This indicates that, with our proposed IR-NSF, directly training with noisy targets can not only simplify post-processing procedures but also yield superior results compared to generating pseudo-clean targets.

\subsubsection{Comparison of IR-NSF and N2N under limited data}
\revision{
The IR-NSF and N2N are established on the statistical equivalences of training on noisy targets and clean targets. The term `\emph{statistical}' implies that such an equivalence could be invalid when the number of training data is limited. In this subsection, we conduct a comparison of N2N and IR-NSF from the perspective of the requirement on the number of training data. The task used for this comparison is the image denoising task. Two types of noise have been considered. The first noise type is the {\iid} Poisson-Gaussian noise. In this case, we employ 800 images from DIV2K as the source images. For each source image, two noisy images are independently generated to serve as training pairs. As such, there are 800 training pairs for this case in total. For the two noisy images in each pair, we do not specify which is the input and which is the target. Instead, their roles are randomly and dynamically determined at each training iteration to maximize the use of them. The second noise type is a mixture of Poisson-Gaussian and periodic noise. In this case, 2650 source images from Flickr2K~\cite{Flickr2K} are further selected, and thus there are 3450 training pairs in total. In Fig.~\ref{fig: convergence_data_budget}, the convergence curves of N2N and IR-NSF are plotted. The plots of N2C and IR-CSF are also provided for reference. In Fig.~\ref{fig: convergence_data_budget}(a), we observe that all of the convergence curves are almost consistent everywhere. This indicates that, when the target image is corrupted by {\iid} Poisson-Gaussian, the statistical equivalences utilized by N2N and IR-NSF will be effective on a dataset including an equal or larger amount of images than DIV2K. However, things become different when the target image is further corrupted by periodic noise. In Fig~\ref{fig: convergence_data_budget}(b), the convergence curve of N2N shows a noticeable deviation from that of N2C, indicating that the statistical equivalence for N2N becomes invalid in this case, even though more training data have been adopted. Consequently, as shown in Fig.~\ref{fig: visual_data_budget}, the periodic noise in the corrupted target will also appear in the image denoised by N2N. By contrast, it can be observed that there is just a minor performance gap between IR-NSF and IR-CSF. Also, the noise pattern is suppressed in the image processed by IR-NSF. These results demonstrate that the statistical equivalence for IR-NSF is still effective in this case. Therefore, we claim that, compared to N2N, IR-NSF has less requirement on the number of training data when the target is corrupted by the spatially correlated noise such as the periodic noise.
% the statistical equivalence for N2N becomes invalid when the target is corrupted by the periodic noise,  
}

% ------------------------------------------------------------------- %
\begin{table}[!t]
	\centering
	\caption{\revision{Average PSNR ({\SI{}{dB}}) and SSIM results of VSNR~\cite{VSNR}, SEID~\cite{SEID}, and USR. These methods are adopted to remove the stripe noise. The Gaussian noise is post-processed by BM3D~\cite{BM3D}. The results of USR are highlighted with gray color. The best results are highlighted in bold.}}
	\label{tab: destripe}
	\resizebox*{0.49\textwidth}{!}{
		\begin{tabular}{*{8}{c}}
			\toprule
			\multirow{2}*{Noise} & \multirow{2}*{Method} 
            & \multicolumn{2}{c}{Set5} & \multicolumn{2}{c}{Kodak24}  & \multicolumn{2}{c}{McMaster} \\
              \cmidrule(lr){3-4}         \cmidrule(lr){5-6}         \cmidrule(lr){7-8} 
            && PSNR      & SSIM       & PSNR     & SSIM          & PSNR      & SSIM\\
			
			\midrule
			\multirow{3}*{Row}  
            & VSNR  & 38.52 & \textbf{0.9915} & 36.16 & 0.9845 & 38.08 & 0.9841 \\
			& SEID  & 35.90 & 0.9474 & 35.69 & 0.9818 & 35.05 & 0.9565 \\
			& \cellcolor{gray!20}USR   & \cellcolor{gray!20}\textbf{39.99} & \cellcolor{gray!20}0.9905 & \cellcolor{gray!20}\textbf{38.51} & \cellcolor{gray!20}\textbf{0.9953} & \cellcolor{gray!20}\textbf{39.84} & \cellcolor{gray!20}\textbf{0.9888} \\

            \midrule
			\multirow{3}*{\shortstack{Row \\ $\&$ \\ Gaussian}}  
            & VSNR  & 34.40 & 0.9375 & 33.30 & 0.9481 & 34.81 & 0.9418 \\
            & SEID  & 34.81 & 0.9134 & 34.92 & 0.9533 & 36.22 & 0.9499 \\
			& \cellcolor{gray!20}USR   & \cellcolor{gray!20}\textbf{35.89} & \cellcolor{gray!20}\textbf{0.9470} & \cellcolor{gray!20}\textbf{35.08} & \cellcolor{gray!20}\textbf{0.9533} & \cellcolor{gray!20}\textbf{36.23} & \cellcolor{gray!20}\textbf{0.9524} \\

   %          \midrule
			% \multicolumn{2}{c}{Time}  &  \multicolumn{2}{c}{0} &  \multicolumn{2}{c}{0} &  \multicolumn{2}{c}{0} \\
   
			\bottomrule
	  \end{tabular}
    }
\end{table}

% \begin{table}[!t]
% 	\centering
% 	\caption{FN2V}
% 	% \label{tab: fn2v}
% 	% \resizebox*{\textwidth}{!}{
% 		\begin{tabular}{*{7}{c}}
% 			\toprule
% 			\multirow{2}*{Method} & \multicolumn{2}{c}{VSNR} & \multicolumn{2}{c}{N2V} & \multicolumn{2}{c}{FN2V}\\
%             \cmidrule(lr){2-3}   \cmidrule(lr){4-5}   \cmidrule(lr){6-7} 
%             & PSNR & SSIM & PSNR & SSIM & PSNR & SSIM\\
			
% 			\midrule
% 			Set5     & 36.51 & 0.9670 & 24.57 & 0.5518 & 36.49 & 0.9832 \\
% 			Set14    & 35.36 & 0.9723 & 24.20 & 0.5695 & 36.74 & 0.9900 \\
% 			Kodak24  & 36.66 & 0.9864 & 24.28 & 0.5234 & 36.19 & 0.9931 \\
% 			McMaster & 36.95 & 0.9656 & 24.66 & 0.5360 & 38.26 & 0.9852 \\
% 		  Urban100 & 36.82 & 0.9702 & 24.17 & 0.6171 & 38.92 & 0.9902 \\
% 		  Manga109 & 35.32 & 0.9545 & 24.95 & 0.5753 & 38.55 & 0.9909 \\
   
% 			\bottomrule
% 	  \end{tabular}
%     % }
% \end{table}

\begin{figure}[t] \centering
    \subfigure[Ground Truth]{
        \includegraphics[width=0.23\textwidth]{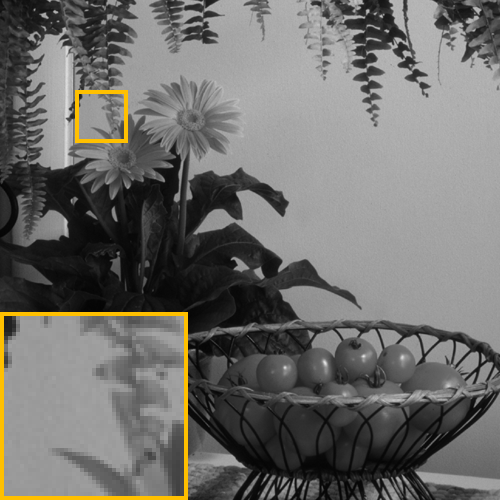}
    } 
    \hspace{-1em}
    \subfigure[VSNR]{
        \includegraphics[width=0.23\textwidth]{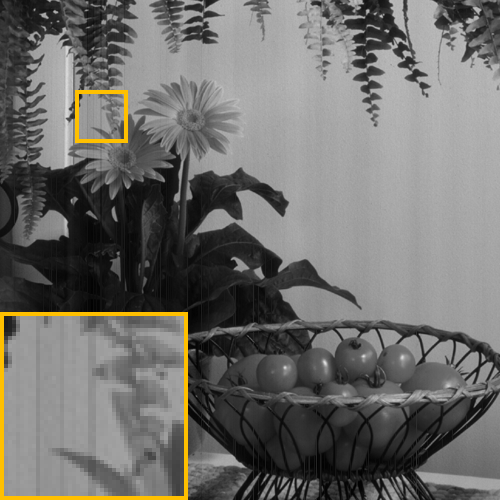}
    }
    \\
    \vspace{-0.5em}
    \subfigure[SEID]{
        \includegraphics[width=0.23\textwidth]{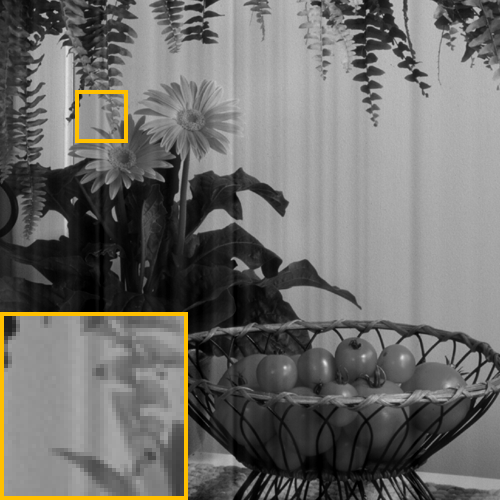}
    }
    \hspace{-1em}
    \subfigure[USR]{
        \includegraphics[width=0.23\textwidth]{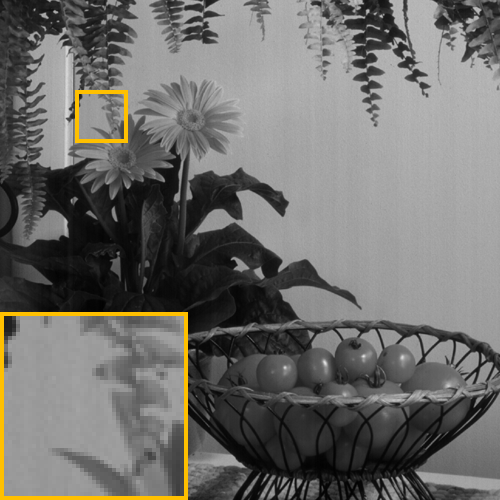}
    }
    \vspace{-0.5em}
    \caption{\revision{Visual comparison among VSNR~\cite{VSNR}, SEID~\cite{SEID}, and USR. (a) Ground Truth. (b) VSNR (PSNR={36.93\SI{}{dB}}, SSIM=0.9746). (c) SEID (PSNR={33.47\SI{}{dB}}, SSIM=0.9481). (d) USR (PSNR={39.92\SI{}{dB}}, SSIM=0.9915).}} 
    \label{fig: visual_destripe}
\end{figure}
% ------------------------------------------------------------------- %
\subsection{Impact on Unsupervised Denoising Methods} \label{sec: unsupervised_denoising}
In the field of image denoising, unsupervised methods refer to approaches that directly learn from noisy images to be processed. Unlike IR-NSF and N2N, which rely on two independent observations of the same scene to collect noisy-noisy image pairs, unsupervised methods only need to observe each scene once to acquire single noisy images. \revision{For removing pixel-wise noise, the unsupervised denoising methods (\eg, Neighbor2Neighbor~\cite{Nr2Nr}) have achieved great success.} However, as introduced in Section~\ref{sec: noise_model}, existing unsupervised denoising methods are generally ineffective when dealing with spatially-correlated noise.
\revision{
This impedes their applications in scenarios such as extremely low-light imaging and remote sensing, where the stripe noise illustrated in Fig.~\ref{fig: psd} is a common occurrence. To overcome this problem, an unsupervised stripe removal (USR) is developed based on the proposed statistical equivalence in the Fourier domain to address stripe-wise noise. In scenarios involving both pixel-wise and stripe-wise noise, it can serve as a pre-processing destriping step for the denoising technique that can only tackle pixel-wise noise.}

\revision{Given a set of unpaired images that are corrupted by stripe noise, the training pairs of USR are generated by the noise-swapping strategy. At each training iteration, the USR first processes the noisy image $y_i = z_i + n_i$ with the network under training $f_\theta(\cdot)$ to obtain an estimate of the latent clean image $\hat{z}_i$, for which the gradient calculation is stopped to stabilize learning. Then, the stripe noise is estimated by $\hat{n}_i = y_i - \hat{z}_i$. Finally, the paired noisy image for $y_i$ is obtained by $\tilde{x}_i = \hat{z}_i + \epsilon \cdot \hat{n}_j$, where $\hat{n}_j$ is the stripe estimated from another collected noisy image. The parameter $\epsilon > 1$ is an amplification factor incorporated to tackle the case where the stripe noise is under-estimated, which is usually encountered during the training process. After swapping the noise, the $\tilde{x}_i$ and $y_i$ are adopted to serve as the training input and target, respectively. As shown in Fig.~\ref{fig: psd}, the stripe noise only affects Fourier coefficients $a(n)[k,l]$ and $b(n)[k, l]$ at $k=0$, implying that there is no need to process other coefficients. As such, the loss function is only calculated on the coefficients at $k=0$, \ie, 
} 
\begin{align}   \label{eqn: loss_USR}
    L_{\varphi}\big(f_{\theta}(\tilde{x}), y\big) 
    & = \sum\limits_{l=0}^{V-1} \varphi\Big(a\big(f_\theta(\tilde{x})\big)[0,l] - a(y)[0,l]\Big)
    \notag
    \\
    & + \sum\limits_{l=0}^{V-1} \varphi\Big(b\big(f_\theta(\tilde{x})\big)[0,l] - b(y)[0,l]\Big).
\end{align}
\revision{This ensures that the trained network will only eliminate the stripe noise, capable of keeping the image structure and other possible pixel-wise noise almost unchanged. 

The performance of USR is validated on both the synthetic and real-world experiments. In synthetic experiments, we evaluate its performance of removing stripes under scenarios with and without the interference of Gaussian noise. The training data is synthesized with the source images from the DIV2K. The test sets are Set5, Kodak24, and McMaster. The comparison methods are VSNR~\cite{VSNR} and SEID~\cite{SEID}. The VSNR is an effective stationary noise removal algorithm. The SEID is a recently proposed state-of-the-art stripe estimation framework. In the experiment, these two methods and USR are adopted to remove the stripe noise. The BM3D is adopted to further post-process the destriped image if there is Gaussian noise. As Table~\ref{tab: destripe} shows, the proposed USR achieves the highest PSNR and SSIM in almost all cases. From Fig.~\ref{fig: visual_destripe}, it can be observed that there are stripes remaining in the images processed by VSNR and SEID. By contrast, the stripe noise in the image processed by USR is successfully eliminated. 

The effectiveness of USR is also validated through real-world experiments, where the unsupervised denoising method Neighbor2Neighrbor\cite{Nr2Nr} is the comparison method and the remote sensing images extracted from the compact high resolution imaging spectrometer (CHRIS) dataset\footnote{https://www.brockmann-consult.de/beam/data/products/} are utilized as both the training set and test set. As Fig.~\ref{fig: visual_CHRIS}(a) illustrates, images from CHIRS are corrupted by both the stripe noise and the pixel-wise noise. In such a scenario, Neighbor2Neighbor can suppress pixel-wise noise but is entirely ineffective against stripe noise. In contrast, USR is effective at eliminating stripe noise while leaving the pixel-wise noise unchanged. Since the stripe noise is typically more noticeable than pixel-wise noise, the result of USR is visually more appealing than that of Neighbor2Neighbor. But the best solution is to combine both methods collaboratively. Utilizing USR as a pre-processing step to eliminate stripe noise and Neighbor2Neighbor to suppress pixel-wise noise offers an effective solution for handling mixed noise in an unsupervised manner.

% Due to the existence of the spatially correlated stripe noise, the unsupervised method designed for pixel-wise noise, such as the Neighbor2Neighbor, is invalid in this case, as shown in Fig.~\ref{fig: visual_CHRIS}(b). To tackle this problem, the USR can be applied to training a network to obtain the destriped image as shown in Fig.~\ref{fig: visual_CHRIS}(c), before applying Neighbor2Neighbor to training the other network to remove the pixel-wise noise to obtain the final result as shown in Fig.~\ref{fig: visual_CHRIS}(d).
}
\begin{figure}[t] \centering
    \subfigure[Original]{
        \includegraphics[width=0.23\textwidth]{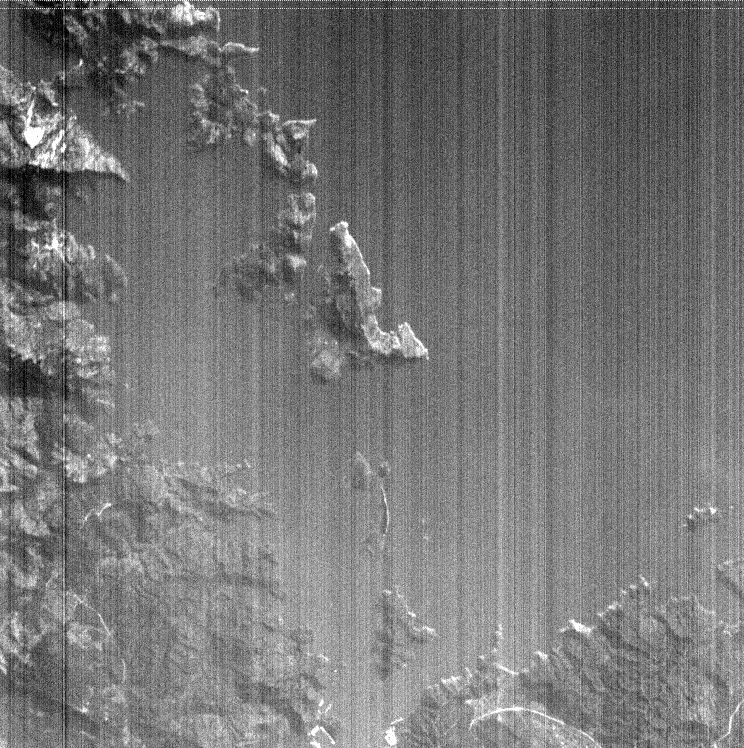}
    } 
    \hspace{-1em}
    \subfigure[Neighbor2Neighbor]{
        \includegraphics[width=0.23\textwidth]{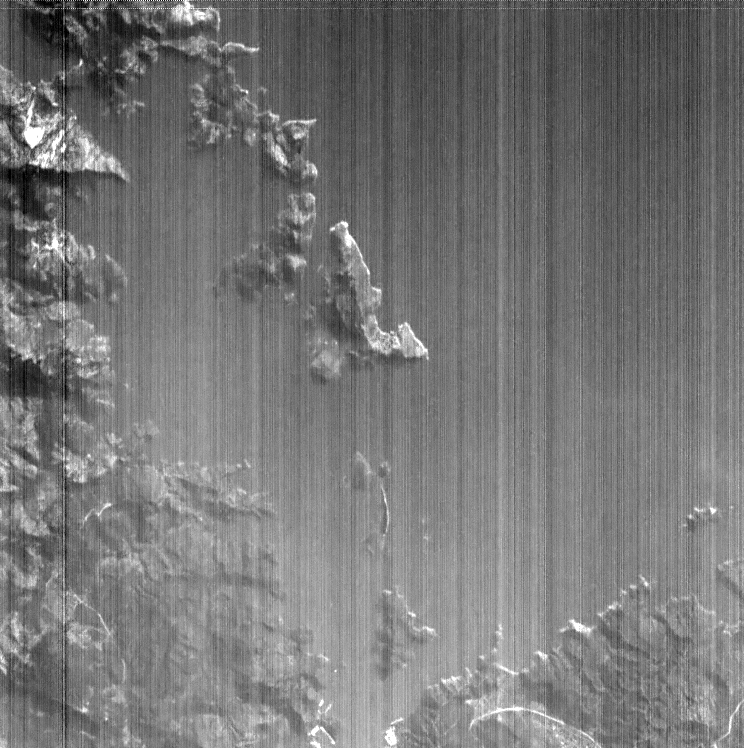}
    }
    \\
    \vspace{-0.5em}
    \subfigure[USR]{
        \includegraphics[width=0.23\textwidth]{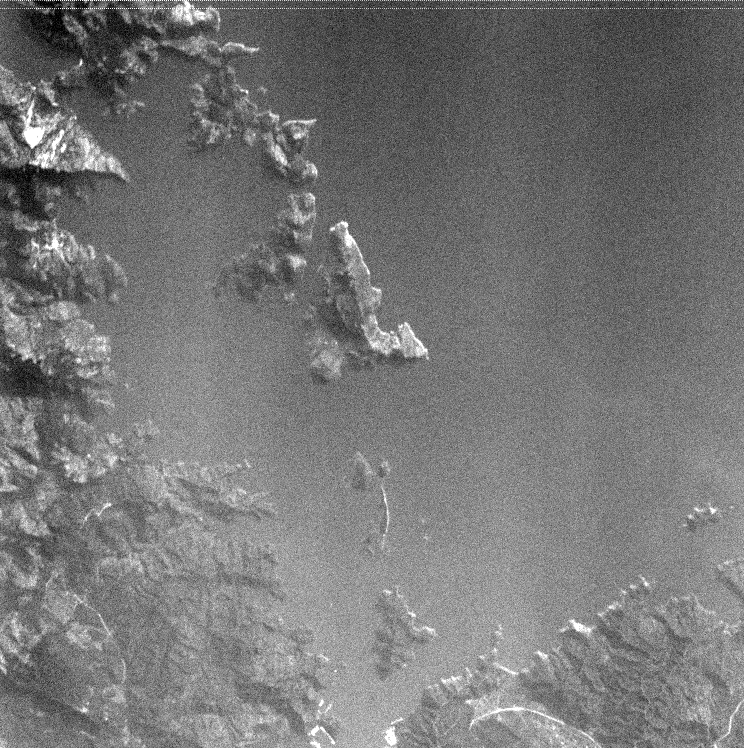}
    }
    \hspace{-1em}
    \subfigure[USR $\&$ Neighbor2Neighbor]{
        \includegraphics[width=0.23\textwidth]{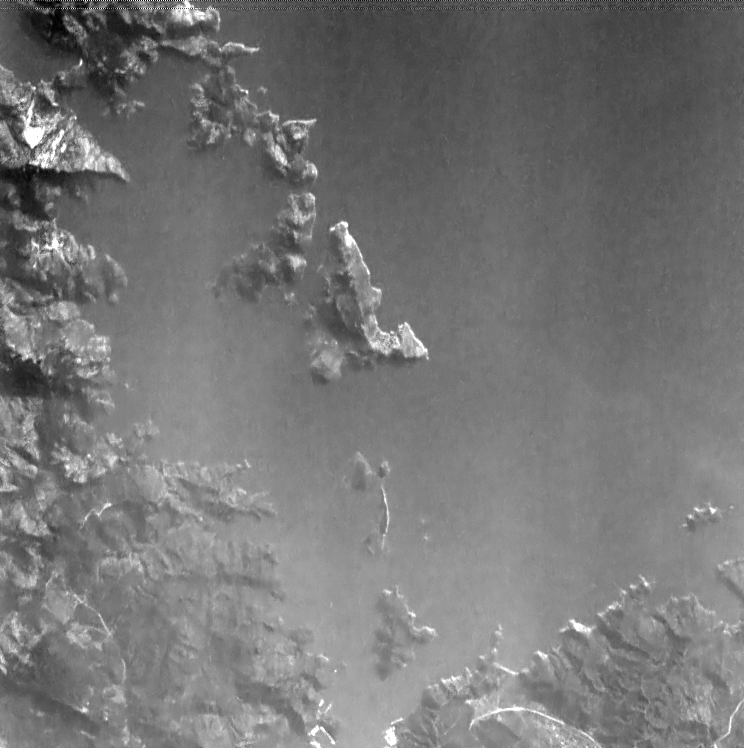}
    }
    \vspace{-0.5em}
    \caption{\revision{Visual results on real-world remote sensing dataset CHRIS. Compared to Neighbor2Neighbor~\cite{Nr2Nr}, which focuses on suppressing pixel-wise noise, the proposed USR can produce a visually more appealing destriped image. By collaboratively utilizing these two methods,  an effective unsupervised solution can be achieved for addressing the mixture of pixel-wise and stripe-wise noise.}} 
    \label{fig: visual_CHRIS}
\end{figure}

\section{Conclusion}
Motivated by the key insights that the Fourier coefficient contains global information and that spatially correlated noise, such as periodic noise and row noise, exhibits sparsity and independence in the Fourier domain, this work proposes a novel concept of establishing noisy supervision in the Fourier domain. To lay the foundation for this concept,  a rigorous analysis of noise properties in the Fourier domain is conducted, revealing the consistent Gaussian distribution of Fourier coefficients for a wide range of noise types. This finding enables a unified treatment of different noise models under a common framework. Furthermore, it is proven that utilizing noisy targets is statistically equivalent to utilizing clean targets for image restoration learning. Based on this equivalence, an efficient learning framework called IR-NSF is proposed and extensively validated on multiple image restoration tasks using diverse network architectures. The experimental results consistently demonstrate the superiority of IR-NSF. 

Additionally, this work highlights the potential impacts of its findings on the data collection procedure and the unsupervised denoising. Specifically, it is illustrated that directly collecting noisy targets is a preferable approach compared to using a denoising method to generate pseudo-clean targets. Moreover, the newly developed unsupervised stripe noise removal showcases that the statistical equivalence based on Fourier analysis can introduce new ideas and perspectives to advance the field of unsupervised denoising. These potential impacts emphasize the broader significance and applicability of this research beyond the proposed learning framework, extending to other areas of image restoration and related fields.

% or
% \appendix  % for no appendix heading

\appendices
\section{Proof of the Theorem~\ref{theorem: CLT-DFT-iid}} \label{proof: CLT-DFT-iid}

\begin{proof}
	% ~
	For a random variable $X$, we denote its moment generating function by
	\begin{equation}
		M_{X}(t) = \mathbb{E}\left[ e^{tX} \right].
	\end{equation}
	The natural logarithm of $M_{X}(t)$ is called the cumulant generating function $K_{X}(t)$. Using the Taylor expansion, there is  
	\begin{equation}
		K_{X}(t) = \ln{\mathbb{E}\left[ e^{tX} \right]} = \sum\limits_{m=1}^{\infty} \kappa_m[X] \frac{t^m}{m!},
	\end{equation}
	where the $m$-th-order cumulant $\kappa_m[X]$ is determined by the $m$-th-order derivative of $K_{X}(t)$ at zero, \ie,
	\begin{equation}
		\kappa_m[X] = K_{X}^{(m)}(0).
	\end{equation}
	An important property of the cumulants is that, for independent variables $X_1$ and $X_2$, there is 
	\begin{equation}    \label{eqn: cumulant_property}
		\kappa_m[c_1 X_1 + c_2X_2] = c_1^m \kappa_m[X_1] + c_2^m \kappa_m[X_2],
	\end{equation}
	where $c_1$ and $c_2$ denote constant numbers. It is also worth noting that the first-order and second-order cumulants of $X$ are the mean and variance of $X$, \ie, 
	\begin{equation}
		\kappa_1[X] = \mu[X], \quad  \kappa_2[X] = \sigma^2[X],
	\end{equation}
	where $\mu[\cdot]$ and $\sigma^2[\cdot]$ denote the mean and variance operators, respectively.
	
	\textbf{For the statement (i):} For $k=l=0$, the real Fourier coefficient of $n$ is
	\begin{equation}
		a(n)[0,0] = \frac{1}{UV} \sum\limits_{u=0}^{U-1} \sum\limits_{v=0}^{V-1} n[u,v], 
	\end{equation} 
	which is a sum of \iid~noise. Based on the central limit theorem, it converges to a Gaussian distribution as $UV$ approaches infinity. The imaginary coefficient is $b(n)[0,0] = 0$. It can be treated as a degenerate Gaussian variable with variance $0$.
	
	In other cases (\ie, $kl \neq 0$), to simplify notation, we ignore the position indices $[k, l]$ and adopt $w[u,v]$ to denote the sine function and the cosine function interchangeably. Based on the properties of the sine and cosine functions, there are
	\begin{equation}
		\sum\limits_{u=0}^{U-1} \sum\limits_{v=0}^{V-1} w[u,v] = 0,
	\end{equation}
	and 
	\begin{equation}
		\sum\limits_{u=0}^{U-1} \sum\limits_{v=0}^{V-1} w^2[u,v] = \frac{UV}{2}.
	\end{equation}
	
	As such, each Fourier coefficient can be denoted by a weighted sum of \iid~noise, \ie
	\begin{equation}
		X = \frac{1}{UV} \sum\limits_{u=0}^{U-1} \sum\limits_{v=0}^{V-1} n[u,v]w[u,v].
	\end{equation}
	The objective is now to show that, as $UV$ approaches infinity, the cumulant generating function of normalized $X$ converges to that of the standard Gaussian distribution, \ie, 
	\begin{equation}
		\lim_{UV \to \infty} K_{\tilde{X}}(t) = \frac{t^2}{2},
	\end{equation}
	where $\tilde{X}$ is a normalized version of $X$, \ie,
	\begin{equation}
		\tilde{X} = \frac{X - \mu[X]}{\sigma[X]}.
	\end{equation}
	Denote by $\mu_n$ and $\sigma^2_n$ the mean and variance of the \iid~noise $n$, we have
	\begin{equation}    \label{eqn: mu_X}
		\mu[X] = \frac{\mu_n}{UV} \sum\limits_{u=0}^{U-1} \sum\limits_{v=0}^{V-1} w[u,v] = 0,
	\end{equation}
	and
	\begin{equation}    \label{eqn: sigma_X}
		\sigma^2[X] = \frac{\sigma_n^2}{(UV)^2} \sum\limits_{u=0}^{U-1} \sum\limits_{v=0}^{V-1} w^2[u,v] = \frac{\sigma_n^2}{2UV}.
	\end{equation}
	Using the property of cumulants shown in (\ref{eqn: cumulant_property}), we can further have
	\begin{equation}
		\kappa_m[\tilde{X}] 
		= \frac{\kappa_m[X]}{\left(\sigma[X]\right)^m} 
		% = \frac{\kappa_m[n]}{(UV\sigma[X])^m} \sum\limits_{u=0}^{U-1} \sum\limits_{v=0}^{V-1} (w[u,v])^m 
		= \frac{\kappa_m[n] \cdot \sum\limits_{u=0}^{U-1} \sum\limits_{v=0}^{V-1} \left(w[u,v]\right)^m}{\left(UV \sigma_n^2 / 2\right)^{m/2}}.
	\end{equation}
	Given that the cumulants of \iid~noise are bounded (\ie, $ \big| \kappa_m[n] \big| \leq C$) and there is $\big| w[u, v] \big| \leq 1$, for $m > 2$, the following inequality holds:
	\begin{equation}
		\lim_{UV \to \infty}  \big| \kappa_m[\tilde{X}] \big| 
		\leq \lim_{UV \to \infty} \frac{CUV}{(UV \sigma_n^2 / 2)^{m/2}}
		= 0.
	\end{equation}
	This inequality indicates that, for $m > 2$, we have
	\begin{equation}
		\lim_{UV \to \infty}  \kappa_m[\tilde{X}] = 0.
	\end{equation}
	Furthermore, for $m=1$ and $m=2$, there are $\kappa_1[\tilde{X}] = \mu[\tilde{X}] = 0$ and $\kappa_2[\tilde{X}] = \sigma^2[\tilde{X}] = 1$. Therefore, we finally have
	\begin{equation}
		\lim_{UV \to \infty} K_{\tilde{X}}(t) 
		= \lim_{UV \to \infty} \sum\limits_{m=1}^{\infty} \kappa_m[\tilde{X}] \frac{t^m}{m!}
		= \frac{t^2}{2}.
	\end{equation}
	
	\textbf{For the statement (ii):}
	Following similar steps, we can also prove the claim that every linear combination of any two Fourier coefficients of \iid~noise converges to a Gaussian distribution as $UV$ approaches infinity. Based on the definition of the multivariate Gaussian distribution in~\cite{probability}, this claim indicates that the joint distribution of any two Fourier coefficients of \iid~noise converges to a multivariate Gaussian distribution as $UV$ approaches infinity. 
	
	Furthermore, according to the orthogonality of the 2D-DFT, Fourier coefficients of \iid~noise are uncorrelated with each other. For random variables jointly following a multivariate Gaussian distribution, they are independent if they are uncorrelated. Therefore, Fourier coefficients of \iid~noise becomes independent with each other as $UV$ approaches infinity.
\end{proof}

\section{Proof of the Theorem~\ref{theorem: CLT-DFT-stationary}} \label{proof: CLT-DFT-stationary}
\begin{proof}
	\textbf{For the statement (i):}
	A key property of the Fourier transform is that the convolution of two functions in spatial domain is the same as the product of their respective Fourier transforms in the Fourier domain, \ie,
	\begin{equation}    \label{eqn: DFT_conv}
		\mathscr{F}(n) 
		= \mathscr{F}(h * \eta) 
		= \mathscr{F}(h) \cdot \mathscr{F}(\eta).
	\end{equation}
	Consequently, there is
	\begin{align}       \label{eqn: DFT_conv_coefficient}
		a(n) &= a(h) \cdot a({\eta}) - b(h) \cdot b({\eta}), 
		\notag
		\\
		b(n) &= a(h) \cdot b({\eta}) + b(h) \cdot a({\eta}).
	\end{align}
	Based on Theorem~\ref{theorem: CLT-DFT-iid}, as $UV$ goes to infinity, $a({\eta})[k,l]$ and $b({\eta})[k,l]$ converge in distribution to two independent Gaussian distributions. Since the linear combination of independent Gaussian variables is still a Gaussian variable, both $a(n)[k,l]$ and $b(n)[k,l]$ will also converge in distribution to Gaussian distributions.
	
	\textbf{For the statement (ii):} As (\ref{eqn: DFT_conv_coefficient}) shows, each Fourier coefficient of $n$ is a linear combination of the Fourier coefficients of \iid~noise $\eta$ at the same position. Consequently, when the Fourier coefficients of $\eta$ at different positions become independent as $UV$ goes to infinity, Fourier coefficients of $n$ at different positions will also become independent. Furthermore, (\ref{eqn: DFT_conv_coefficient}) implies that the Fourier coefficients of $n$ at the same position, \ie, $a(n)[k,l]$ and $b(n)[k,l]$, will jointly converge to a multivariate Gaussian distribution as $UV$ goes to infinity. Since there is  
	\begin{equation}
		\mathbb{E}\big[ a(n)[k,l] \cdot b(n)[k,l] \big] = 0,
	\end{equation}
	which means that $a(n)[k,l]$ and $b(n)[k,l]$ are uncorrelated, we can conclude that they will become two independent Gaussian variables as $UV$ goes to infinity.
\end{proof}

\section{Proof of the Theorem~\ref{theorem: DFT-SE}} \label{proof: DFT-SE}
\begin{proof}
	\textbf{For the statement (i):}
	Without loss of generality, we focus this proof on a single component (\ie, the $a[k,l]$-related term) of the loss function. Ignoring the position indices $[k,l]$ for simplicity, we can obtain
	\begin{align}   \label{eqn: SE} % statistical equivalence
		& \, \, \mathbb{E} \left[ \varphi\Big(a\big(f_\theta(x)\big) - a(y)\Big) \right]
		\notag
		\\
		& = \int_{-\infty}^{\infty} \varphi\Big(a\big(f_\theta(x)\big) - a(z) - a(n)\Big) p\big(a(n)\big) \,{\rm d}a(n)
		\notag
		\\
		& = \phi\Big(a\big(f_\theta(x)\big) - a(z)\Big),
		% = & \left.(\varphi * p)(t)\right|_{t = a_{f_\theta(x)} - a_z},
	\end{align}
	where $\phi(t) = \varphi(t) * p(t)$. This equation indicates that utilizing a function $\varphi(t)$ with noisy coefficients is statistically equivalent to utilizing its blurred version $\phi(t)$ with clean coefficients. 
	Extending this equation to all components of the loss function, we can obtain the statistical equivalence as %given by~(\ref{eqn: DFT-SE}). 
	\begin{equation}    \label{eqn: DFT-SE-appendix}
		\mathbb{E}\left[L_{\varphi}\big(f_{\theta}(x), y\big)\right] 
		= 
		L_{\phi}\big(f_{\theta}(x), z\big).
	\end{equation}
	
	\textbf{For the statement (ii):}
	Since $p(t)$ denotes the probability density function for a zero-mean Gaussian distribution, its derivative $p'(t)$ satisfies $p'(t)=-p'(-t)$ and $p'(t)<0$ for $t>0$. Using this property, we can get 
	\begin{align}
		\phi'(t) 
		& = \varphi(t) * p'(t)
		\notag
		\\
		& = \int_{-\infty}^{\infty} \varphi(t - \tau) p'(\tau) \,{\rm d}\tau
		\notag
		\\
		& = \int_{0}^{\infty} \varphi(t - \tau) p'(\tau) \,{\rm d}\tau
		+ \int_{-\infty}^{0} \varphi(t - \tau) p'(\tau) \,{\rm d}\tau
		\notag
		\\
		& = \int_{0}^{\infty} \varphi(t - \tau) p'(\tau) \,{\rm d}\tau 
		- \int_{0}^{\infty} \varphi(t + \tau) p'(\tau) \,{\rm d}\tau
		\notag
		\\
		& = \int_{0}^{\infty} \big(\varphi(t - \tau) - \varphi(t + \tau)\big) p'(\tau) \,{\rm d}\tau,
	\end{align}
	which is negative for $t<0$, zero for $t=0$, and positive for $t>0$. This means that  $t^{\star} = 0$ is the minimal point of the function $\phi(t)$. Consequently, $a\big(f_\theta^{\star}(x)\big) = a(z)$ and $f_{\theta}^{\star}(x)=z$ are the minimal points of (\ref{eqn: SE}) and (\ref{eqn: DFT-SE-appendix}), respectively. 
\end{proof}

\section*{Acknowledgment}
This work is supported in part by the Research Grants Council of Hong Kong (GRF 17200321, 17201822). The authors would like to thank Mr. Tao Liu, Mr. Chutian Wang, Mr. Chao Qu and Dr. Shuo Zhu for useful discussions about this paper. \revision{The source codes and trained models are publicly available at the linkage https://github.com/haosennn/IR-NSF.}

\ifCLASSOPTIONcaptionsoff
  \newpage
\fi



\bibliographystyle{IEEEtran}
\bibliography{ref,IEEEabrv}
\end{document}